   \documentclass[12pt,a4paper,twoside]{article}
\usepackage{./Paper_preamble}
\usepackage{bm}

 \usepackage{sectsty}
 \paragraphfont{\mdseries\scshape}
 \makeatletter
\def\@seccntformat#1{\csname the#1\endcsname.\quad}
\makeatother
 \sectionfont{\normalsize\centering\mdseries\scshape}

\makeatletter
\newcommand{\srcsize}{\@setfontsize{\srcsize}{3pt}{3pt}}
\makeatother
\makeatletter
\newcommand{\srcsizetwo}{\@setfontsize{\srcsizetwo}{2pt}{2pt}}
\makeatother

\newcommand{\gsii}{$\textup{[GS]}$}

\DeclareMathOperator{\ess}{ess}
\newcommand{\relations}{\operatorname{rel}}
\newcommand{\novel}{\mathfrak f}

\newcommand{\posint}{\mbb Z_{\mdoubleplus}}
\newcommand{\nnint}{{\mbb Z}_{\mplus}}
\newcommand{\posrat}{\mbb Q_{\mdoubleplus}}
\newcommand{\posreal}{\mbb R_{\mdoubleplus}}

\newcommand{\nnrat}{{\mbb Q}_{\mplus}}
\newcommand{\nnreal}{{\mbb R}_{\mplus}}

\newcommand\mdoubleplus{\text{\srcsize$+\mkern-2mu+$}}
\newcommand\mplus{\text{\srcsize$+$}}
\newcommand{\lightercolor}[3]{
    \colorlet{#3}{#1!#2!white}
}
\lightercolor{gray}{50}{lightgray}
\usepackage{accents}

\usepackage{upgreek}

\newcommand{\spann}{\operatorname{span}}
\newcommand{\reg}{\operatorname{reg}}
\newcommand{\nov}{\operatorname{nov}}
\newcommand{\test}{\operatorname{test}}

\newcommand{\Ext}{\operatorname{ext}}

\newcommand{\precb}{\mathbin{\prec}}

\newcommand{\preceqb}{\mathbin{\preceq}}

\newcommand{\countof}{\mathbin{\sharp}\hskip1pt}

\newcommand{\ext}{\mathrel{\mc R}}
\newcommand{\sext}{\mathrel{\mc P}}
\newcommand{\next}{\mathrel{\mc I}}
\newcommand{\supext}{{\ext}}

\newcommand{\extb}{\mathbin{\mc R}}
\newcommand{\sextb}{\mathbin{\mc P}}
\newcommand{\nextb}{\mathbin{\mc I}}

\newcommand{\hext}{\mathrel{\hat{\mathrel{\mathcal R}}}}
\newcommand{\hsext}{\mathrel{\hat{\mathrel{\mathcal P}}}}
\newcommand{\hnext}{\mathrel{\hat{\mathrel{\mathcal I}}}}

\newcommand{\hextb}{\mathbin{\hat{\mathbin{\mathcal R}}}}
\newcommand{\hsextb}{\mathbin{\hat{\mathbin{\mathcal P}}}}

\newcommand{\aext}{\mathrel{\acute{\mathrel{\mathcal R}}}}
\newcommand{\asext}{\mathrel{\acute{\mathrel{\mathcal P}}}}
\newcommand{\anext}{\mathrel{\acute{\mathrel{\mathcal I}}}}

\newcommand{\aextb}{\mathbin{\acute{\mathbin{\mathcal R}}}}
\newcommand{\asextb}{\mathbin{\acute{\mathbin{\mathcal P}}}}

\newcommand{\cext}{\mathrel{\check{\mathrel{\mathcal R}}}}

\newcommand{\cextb}{\mathbin{\check{\mathbin{\mathcal R}}}}

\newcommand{\total}{\textup{total}}

\newcommand{\mbbd}{{\mathds D}}
\newcommand{\mbbdp}{{\mathds D^{\novel}}}
\newcommand{\mbbdpp}{{\mathfrak D}}

\newcommand{\mbbc}{{\mathds C}}
\newcommand{\mbbcp}{{\mathds C^{\novel}}}
\newcommand{\mbbcpp}{{\mathfrak C}}

\newcommand{\mbbt}{{\mathds {T}}}
\newcommand{\mbbtp}{{\mathds{T} ^ \novel }}
\newcommand{\mbbtpp}{{\mathfrak{T}}}

\newcommand{\mbbi}{{\mathds L}}
\newcommand{\mbbip}{{\mathds{L}^{\novel}}}
\newcommand{\mbbipp}{{\mathfrak L}}

\newcommand{\mbbj}{\mathds J}
\newcommand{\mbbjp}{{\mathds {J}^{\novel}}}
\newcommand{\mbbjpp}{\mathfrak I}

\newcommand{\current}{{C^\star}}

\newcommand{\past}{{D^\star}}

 \newcommand{\lbc}{\left\{}
\newcommand{\rbc}{\right\}}
\newcommand{\lb}{\left\{}
\newcommand{\rb}{\right\}}

\renewcommand{\ij}{{(i, j)}}

\newcommand{\xy}{{(x, y)}}
\newcommand{\yx}{{(y, x)}}

\newcommand{\yz}{{(y,z)}}
\newcommand{\zy}{{(z,y)}}

\newcommand{\xz}{{(x,z)}}
\newcommand{\zx}{{(z,x)}}

\newcommand{\xw}{{(x,w)}}
\newcommand{\wx}{{(w,x)}}

\newcommand{\yw}{{(y,w)}}
\newcommand{\wy}{{(w,y)}}

\newcommand{\zw}{(z,w)}
\newcommand{\wz}{(w,z)}

\newcommand{\xpw}{(x^{\prime},w)}

\newcommand{\xpy}{(x^{\prime},y)}

\newcommand{\xpz}{(x^{\prime},z)}

\newcommand{\dd}{{(\cdot,\cdot)}}

\renewcommand{\v}{{\mathbf{ v }}}

\newcommand{\bmu}{\bm{\upmu}}

\newcommand{\fourpru}{\textit{4}-\textup{prudence}}
\newcommand{\threepru}{\textit{3}-\textup{prudence}}
\newcommand{\parthreediv}{\textup{partial-\textit{3}-diversity}}
\newcommand{\Parthreediv}{\textup{Partial-\textit{3}-diversity}}
\newcommand{\condtwodiv}{\textup{conditional-\textit{2}-diversity}}
\newcommand{\Condtwodiv}{\textup{Conditional-\textit{2}-diversity}}
\newcommand{\twodiv}{\textit{2}-\textup{diversity}}

\newcommand{\fourdiv}{\textit{4}-\textup{diversity}}

\newcommand{\threediv}{\textit{3}-\textup{diversity}}

\newcommand{\fourjac}{\textup{\textit{4}-Jac}}
\newcommand{\threejac}{\textup{\textit{3}-Jac}}

\newcommand{\half}{\frac{1}{2}}
\DeclareMathOperator{\rank}{rank}

\DeclareMathOperator{\image}{image}
\usepackage{blindtext}
\makeatletter
\renewcommand\maketitle
  {\begin{center}\mdseries\large
    {\@title}%
    \par\medskip\medskip
    {\normalsize\@author}%
    \par\medskip\medskip\normalfont
    \begin{small}
\emph{Australian Institute for Business and Economics}
\end{small}
\par
\begin{small}
\emph{University of Queensland}
\end{small}
   \end{center}
  }
  \makeatother
  \usepackage{fancyhdr}
\title{\MakeUppercase{Second-order Inductive Inference:\\ an axiomatic
    approach}\footnote{I thank Itzhak Gilboa for encouraging this project at an
    early stage. I am indebted to Sacha Bourgeois-Gironde for many conversations
    on the neuroscientific evidence for case-based reasoning and for introducing
    me to the related references that appear in the text.

    This version has time stamp \currenttime~(AEST), \today. The most recent
    version (as well as past versions) can be found at
    \url{https://arxiv.org/abs/1904.02934} .}}

\author{\large\textsc{By Patrick H. O'Callaghan}\footnote{Email address
    \href{mailto:p.ocallaghan@uq.edu.au}{\texttt{p.ocallaghan@uq.edu.au}} and
    ORCID iD \href{http://orcid.org/0000-0003-0606-5465}{
      \texttt{0000-0003-0606-5465}} }}

  \date{}
    
\begin{document}
  \pagenumbering{gobble}
  \maketitle

  \pagestyle{fancy}
\renewcommand{\abstractname}{\vspace{-\baselineskip}} \thispagestyle{plain}

\begin{abstract}
  Consider a predictor who ranks eventualities on the basis of past cases: for
  instance a search engine ranking webpages given past searches. Resampling past
  cases leads to different rankings and the extraction of deeper
  information. Yet a rich database, with sufficiently diverse rankings, is often
  beyond reach. Inexperience demands either ``on the fly'' learning-by-doing or
  \emph{prudence}: the arrival of a novel case does not force (i) a revision of
  current rankings, (ii) dogmatism towards new rankings, or (iii)
  intransitivity.

  For this higher-order framework of inductive inference, we derive a suitably
  unique numerical representation of these rankings via a matrix on
  eventualities $\times$ cases and describe a robust test of
  prudence. Applications include: the success/failure of startups; the veracity
  of fake news; and novel conditions for the existence of a yield curve that is
  robustly arbitrage-free.
\end{abstract}
\setlength{\epigraphwidth}{11.5cm}
\epigraph{From the past, the present acts prudently, lest it spoil future
  action.}{\emph{Titian:  Allegory of Prudence}
}
\section{Introduction}\label{sec-introduction}
Experience is the basis of prediction.  Yet outside of stylised settings, even
the most experienced forecasters do not claim  access to ``the full
model'' that is closed  with respect to the relevant states of the
world.
\begin{example}\label{eg-full-model}
  The canonical large-world setting is that of the global financial markets.  In
  2019, all models ommited details of COVID-19. Similarly, in 2007, a clear
  description of the sub-prime mortgage crisis was beyond reach. The central
  role of simulation and bootstrap methods in empirical finance points to a
  prevalence of inductive reasoning.\footnote{Consider
    \citet{Cowles-Forecasting}, \citet{White-Data_snooping},
    \citet{FF-Luck_vs_skill} and \citet{HL-Lucky_factors}.} That is, to
  forecasting on the basis of past cases as opposed to a full description of
  future states. At the same time, market makers need to set prices that
  are robust to changes that open the door to exploitation via arbitrage.
\end{example} 
\pagenumbering{arabic} \setcounter{page}{2} Incomplete models are a key
motivation for recent axiomatic updates to the standard Bayesian framework such
as ``reverse Bayesianism'' of \citet{KV_Reverse_Bayes}. This and other work
(\citet{KV-Awareness_of_U}, \citet{HR-Knowledge_of_U} and
\citet{GKMQT_Robust_experiments}) on unawareness and robustness in the
state-space sense provide part of our inspiration for the present upgrade to the
axiomatic foundations of inductive inference in \citet[henceforth
\gsii]{GS_Inductive_inference}. By taking rankings as primitive, we extend
\gsii\ to accommodate a forward-looking version of the second-order induction of
\citet{AG-Second-order_induction}. 

In the present framework, the basic building blocks of the model are
observations or (synonymously) past cases. A given past case may be empirical or
theoretical, and the predictor's model is naturally bounded in size and scope by
the predictor's experience. Our contribution is to extend the framework of
\gsii\ to model less experienced predictors that are prudent. That is to agents
that are able to proactively engage in second-order induction: by ``pulling
themselves up by the bootstraps'' and looking at how their model extends to
novel cases. Thus allowing them to survive their initial phase of inexperience.

In the remainder of this section, we informally introduce: the model of
\cref{sec-model}; the axioms and matrix representation of
\cref{sec-axioms-theorem}; and the applications to which we return in
\cref{sec-discussion}. Proofs of main results appear in appendices
\ref{sec-proof-main} to \ref{sec-proof-foureq}.

\paragraph{Synopsis of model and results with examples.} The predictor is endowed with a
qualitative plausibility ranking of eventualities given her current database
$D^{\star}$ of past cases. Moreover, the same is true for every finite
resampling of $D^{\star}$. We identify conditions on the resulting family of
(ordinal) rankings for the existence of a suitably unique
real-valued matrix $\mathbf v$ on eventualities$\times$cases that \emph{represents} the
information in these rankings. The form of this representation is linear on
databases (\ie\ additive over cases) and separable on eventualities, so that for
every database $D$, eventuality $y$ is
more likely than $x$ if, and only if,
\begin{linenomath*}
  \begin{equation}\label{eq-similarity}
  \sum_{c\,\in D} \mathbf  v(x,c) \leq \sum_{c\,\in D} \mathbf v(y,c).
\end{equation}
\end{linenomath*}
The similarity weight $\mathbf{v}(x,c)$ is the degree of support that case $c$
lends to $x$.
\begin{example}As a canonical example, consider predicting the slope coefficient
  $\beta_{i}$ from the regression of asset $i$'s returns on the market portfolio
  \citep[in the two-pass method of][]{FM-Two_pass}. Then values $\beta_{i}$ are
  eventualities and $D^{*}$ is the current sample of past returns. In this
  setting, $\mathbf v$ is an empirical log-likelihood function that we generate
  via a generalised notion of bootstrapping of cases in $D^{\star}$.
\end{example}
Key to the contribution of \gsii\ is an endogenous notion of case types: a
partition of cases according to the marginal information they contribute to a
given database. This marginal contribution is measured in terms of the impact on
the rankings of eventualities.\footnote{This notion of case type is therefore
  close to Quine's ``perceptual similarity'' (see \cref{sec-discussion}).}  Case
types are analogous to states in the sense that they form the model's
dimensions: the lower the dimension the less the experience.

\begin{example}[Search Engine Results Page, SERP]\label{eg-search_engine}
  Advertising aside, when users conduct a web search, the search engine compiles
  a ranking $\preceqb_{D^{\star}}$ of (web)pages $x, y, z, \dots$ on the basis
  of its database $D^{\star }$ of past cases: searches of past users plus
  feedback from subsequent clicks.  Resampling yields other databases $D$ and
  other rankings.  At one extreme, past cases may be so similar to one another
  that the same plausibility ranking arises regardless of how the data is
  resampled.  At the other, past cases may be sufficiently rich that resampling
  yields every feasible plausibility ranking of eventualities.
\end{example}

A rich set of past cases is at the heart of the diversity axiom of \gsii--which
we refer to as \fourdiv. This restricts the model to predictors whose current
data is sufficiently rich that a resampling exercise generates \emph{all
  $4!=24 $ strict (\ie\ total) rankings of every subset of four
  eventualities}. In this paper, we accommodate less inexperienced predictors by
only replacing \fourdiv\ with \condtwodiv. Given the other basic axioms of
\gsii, \condtwodiv\ turns out to be equivalent to requiring that, for every
three distinct eventualities $x$, $y$ and $z$, resampling generates at least $3$
of the $3! = 6$ possible distinct strict rankings of this triple. \Condtwodiv\
is minimal in the sense that, in its absence, we lose both existence and
uniqueness of the similarity representation (see \cref{eg-lexicographic}).  

To compensate for a lack of experience, we introduce a subtle and more flexible
notion of cases that allows us to capture the predictor's potential awareness of
her limited experience.  Formally, related notions in the literature go by the
name of unforeseen consequences in \citet{GQ-Surprises} and shadow propositions
in the setting dynamic awareness in \citet{HP-Dynamic_awareness}. Via a
content-free case $\novel $, the predictor can explore the impact on her model
of the arrival of a novel case type. In effect, this involves meta-analysis of
how her similarity function will evolve over time. Hence the reference to
second-order induction.

The prudent predictor  ensures her model is robust to the arrival of novel case
types. She ensures that, when a novel case arrives, she can accept the ranking
it generates without finding herself in the potentially costly position of
generating intransitive rankings when she combines past and novel cases.

\begin{example}[Second-order inductive inference]\label{eg-second-order}
  Consider a search-engine startup seeking to establish itself in the face of
  incumbents with the experience of Google. The start-up engages in second order
  induction when it is learning the similarity function itself. This may include
  learning the values $\mathbf{v}(x,c)$ of \cref{eq-similarity}, but it may
  also involve costly updates of the model structure ``on the fly'', \eg\
  redefining case types, rankings, \etc. The startup is prudent if, \emph{ex
    ante}, it structures its model to ensure that it is relatively costless to
  extend to novel case types$:$ once they arrive.
\end{example}

Prudence is only worthwhile when revisions of the predictor's model `on the fly',
once the novel case arrives, is costly.  
Consider ``zero-day attacks'' in the setting of cyber security.
\Citet{Hota_et_al-Cyber_security} highlight the essence of time when a novel
attack on a computer network arrives; and that such attacks are novel precisely
because cyber-security experts have already built in solutions to known
vulnerabilities. Also, tradeoffs between time, cost and learning are nowhere
more important than in finance. For a bond-market setting, we are able to
provide a formal equivalence between prudence and arbitrage pricing in
\cref{sec-discussion}.  In \cref{sec-discussion}, we also discuss: other
applications; empirical evidence linking intransitivity, memories and novelty;
and connections with the literature on second-order induction in more detail.

\section{Model}\label{sec-model}
Following \gsii, let the nonempty set $X$ denote the conceivable
\emph{eventualities} of the present prediction problem and let $ \relations(X)$
denote the set of binary relations or rankings on $X$.  For instance, for a
search engine, we identify webpage $x$ with the eventuality ``page $x$ is the
desired webpage''.  The predictor is equipped with her current memory
$\current$: the union of a finite set of past cases $\past$ and a
\emph{variable} or \emph{free case} $\novel$.  The cases in $\past$ collectively
represent the forecaster's relevant observations or experience.  Our first and
most fundamental modification of the primitives of \gsii~is the inclusion of
$\novel$ in the current memory
$\current$.  
\begin{remark*}On a computer, a natural implementation of this setup is the
  following.  Take every case $c \in \current $ to consist of a pair
  $p \times m:$ a pointer $p$ that references a memory location and the memory
  content $m$. Each $c\in \past$ is identical to a case in the setting of
  \gsii. But, for $c = \novel$, there is no meaningful memory content, so $m$ is
  ``empty'' or assigned an arbitrary null value.  From another perspective,
  cases in $ \past$ are constant (of arity zero) whereas $\novel$ is a variable
  (of positive arity).

  For every nonempty subsample $ D \subseteq \past $, the predictor has
  sufficient information to determine a well-defined ranking $\preceqb_ D $ in
  $ \relations(X)$.  In contrast, since $\novel = \acute{p} \times \acute{m}$
  has no meaningful memory content, $\preceqb_\novel$ is indeterminate and a
  free variable in $\relations(X)$. The prudent predictor gains a better
  understanding of her current model by assigning a ranking to $\acute{m}$ and
  exploring the extensions of \cref{def-extension}.
\end{remark*}

Like \gsii, we accommodate a forecaster that goes beyond her current memory and
includes hypothetical cases
$\mbbc$ that she may not have experienced, but which, through reasoning,
interpolation or resampling, she can clearly describe. These hypothetical cases
are formally constant, like members of $\past$. 
%
With \emph{case resampling and subsampling} from the literature on bootstrapping
in mind, let
\[\mbbd\defeq \lbc D\subseteq \mbbc: \countof D<\infty\rbc \] denote
the set of (finite) \emph{determinate or constant databases}.  (These are
referred to as \emph{memories} in \gsii.)  Like $\past$, each $D\in \mbbd $
contains no copies of $\novel$.

Let $[\novel] $ denote a set of copies of $\novel $. Finally, let
$\mbbcp \defeq \mbbc \cup [\novel]$ and let $\mbbcpp$ denote a member of
$\{\mbbc, \mbbcp\}$.  Let $\mbbdp $ denote the corresponding set of all finite
subsets of $ \mbbcp $ and take
  \begin{linenomath*}
\begin{equation*}
  \mbbdpp = \left\{
\begin{array}{ll}
 \mbbd & \text{if, and only if, $\mbbcpp= \mbbc$, and}\\
\mbbdp  &\text{otherwise.}
\end{array}\right.
\end{equation*}
\end{linenomath*}
The predictor is endowed with a well-defined plausibility ranking $\preceqb_D $
in $ \relations(X)$ for each $D$ in $\mbbd$.  Denote the symmetric part by
$\simeq _ D$ and asymmetric part by $ \precb _ D $.  In a minor departure from
\gsii, the primitive of our model is a point in $\relations(X)^{\mbbd}$
  \begin{linenomath*}
\[\preceqb_{\mbbd} \defeq \langle\preceqb_{D}:D\in \mbbd\rangle\,.\footnote{The
    present approach is equivalent to taking
    $\preceqb_{\mbbd} = \left\{D \times \preceqb_{D}: D \in
      \mbbd\right\}$. This way we maintain pairwise distinctness of
    $\preceqb_{C}=\preceqb_{D}$ such that $C \neq
    D$. I thank Maxwell B. Stinchcombe for bringing this point to my attention.}
\]
  \end{linenomath*}
For each $C$ in $ \mbbdp \bs \mbbd $, the fact that for some
$ c \in [ \novel ] $, $ c \in C $ means that $\preceqb_{C} $ is indeterminate,
free variable in $\relations(X)$. Although, in isolation each such
$ \preceqb_{C} $ is free, when the axioms we introduce hold, the potential
values of the variable
$ \preceqb _ \mbbdp \defeq \langle \preceqb _ C : C \in \mbbdp \rangle $ are
constrained by the current values of the constant $\preceqb_\mbbd$.





\paragraph{Case types.}
As in \gsii, two past cases $ c , d \in \mbbc $ are of the same \emph{case type}
if, and only if, the marginal information of $ c $ is everywhere equal to the
marginal information of $ d $. Formally, $ c \sim ^{ \star } d $ if, and only
if, for every $ D \in \mbbd $ such that $ c , d \notin D $,
$ \preceqb _ { D \cup \{ c \} } = \preceqb _ { D \cup \{ d \} }$.  By
observation 1 of \gsii, $ \sim ^{ \star }$ is an equivalence relation on
$ \mbbc $ and as its collection of equivalence classes generate a partition
$\mbbt$ of $\mbbc$.

We extend $ \sim ^{ \star } $ to $ \mbbcp $ by taking $ [ \novel ]$ to be an
equivalence class of its own, so that, for every $ c \in \mbbc $,
$ c \nsim ^{ \star } \novel $.  We let $\mbbtp$ denote the corresponding
partition of $\mbbcp$.

Like \gsii, we also extend $\sim^{\star}$ to $\mbbdp$ by treating databases that
contain the same number of each case type as equivalent. That is,
$C \sim ^{\star } D$ if, and only if, for every $t \in \mbbtp$, the numbers
$\countof (C \cap t)$ and $\countof (D \cap t) $ of that case type coincide. To
enable a translation of each database to counting vectors
$t \mapsto \countof (D \cap t) $, we impose a


\begin{assumption*}
For every  $ t \in \mbbtp$, there are infinitely many cases in $ t $.
\end{assumption*}


Our key definition is the following. 

\begin{definition}\label{def-extension}
  $\extb \defeq \langle \extb_{D}: D\in \mbbdpp \rangle $ 
  is an \emph{extension}, and in particular a $Y$-extension, of $ \preceq_{\mbbd}$ if, for some nonempty
  $ Y \subseteq X $, the following all hold$:$

\begin{enumerate}

\item\label{item-binary-rel} for every $ D\in \mbbdpp $, $\ext_{D} $ belongs
  to $ \relations (Y)$,

\item \label{item-preserving} for every $ D \in \mbbd$ and every $x,y\in Y$,
  $x \ext_{D} y $ if, and only if, $x \preceq_{D} y$,

\item \label{item-dimension} for every $D\in \mbbdpp$ and 
every $c,d \in \mbbcpp\bs D$,  if $ c \sim ^ \star d $ then
  $ \extb _ { D \cup \{c\} } = \extb _ { D \cup \{d\}}$.
\end{enumerate}
An extension $\ext_{\mbbdpp}$ is  \emph{proper} if $\mbbdpp = \mbbdp$ and
improper otherwise.
\end{definition}
By assigning rankings to the free case $\novel$, (potential) extensions of
$\preceq_{\mbbd}$ simulate the arrival of novel information.  Part
\ref{item-binary-rel} of the definition of an extension ensures that, for every
proper $Y$-extension $\ext$ and every $D\in \mbbdp$, $\ext_{D}$ is a
well-defined binary relation on $Y$. For every extension $\ext_{\mbbdpp}$ and
every $D \in \mbbdpp$, let $ \next_{D}$ and $ \sext_{D}$ respectively denote the
symmetric and asymmetric parts of $\ext_{D}$.

Part \ref{item-preserving} of the definition implies that, for every
$Y$-extension $\ext$ and every $D\in \mbbd$, $ \ext$ is simply the restriction
$\preceqb_{D}\cap Y^{2} $ of $\preceq_{D}$ to $Y$.  We therefore refer to part
\ref{item-preserving} of the definition of an extension as the preservation or
\emph{nonrevision} condition.

Part \ref{item-dimension} of the definition ensures that, for proper extensions
$\ext $, the partition $\mbbtp$ of case types generated by $\sim^{\star}$ is at
least as fine the partition generated by the equivalence relation generated by
$ \ext$. Two cases $ c,d \in \mbbcp $ are \emph{equivalent with respect to
  $ \ext $}, written $ c \sim^{\extb} d $, if, for every $ D \in \mbbd $ such
that $ c,d \notin D $, $ \extb _ { D \cup \{c\} } = \extb _ { D \cup\{d\} }$.
This notion allows us to partition the set of proper extensions as follows.


\begin{definition*}\label{def-novel}

  A proper extension $\ext $ is either \emph{regular} or \emph{novel}.  It is
  novel whenever $[\novel]$ is a distinct equivalence class of $\sim^{\supext}$,
  so that, for every $ c \in \mbbc $, $ c \nsim ^ \supext \novel $.

\end{definition*}



For novel extensions, $\novel $ mimmicks the potential arrival of new
information. Yet novel extensions need not feature qualitatively new rankings
(\ie\ rankings that do not feature in $\preceq_{\mbbd}$). In \cref{lem-insep},
we show that this is because a novel case type is characterised by the
quantitative notion of a similarity weight.

For every regular extension $\ext$, there exists $c \in \mbbc$ such that
$c \sim^{\extb} \novel$. Thus, there are as many regular extensions as there are
past case types (\ie\ $\countof \mbbt$).  Yet every regular $Y$-extension $\ext$
is \emph{equivalent} to the unique improper $Y$-extension
$ \langle \preceqb_{D} \cap Y^{2}: D \in \mbbd\rangle$ in the sense of 
\begin{observation}\label{obs-reg-eq}%
  For every regular $Y$-extension $\ext$ and improper $Y$-extension $\aext$, for
  every $ C \in \mbbdp $, there exists $ D \in \mbbd $ such that
  $C \sim^{\extb} D$ and $\extb_{C} = \aextb_{D}$.\footnote{This means that
    there is a canonical embedding of
    $\left\{C \times \extb_{C}: C \in \mbbdp\right\}$ in
    $\{D \times \aextb_{D}: D \in \mbbd\}$. The converse embedding follows from
    the nonrevision condition of \cref{def-extension}.} 
\end{observation}

\begin{proof}[Proof of \cref{obs-reg-eq}]\label{proof-reg-eq}
  Fix $Y\subseteq X$ nonempty and $\aext$ regular.  \Wlog, take $ C \in \mbbdp
  \bs \mbbd $, so that $ C $ contains at least one copy of $ \novel
  $.  For any $ c \in C \cap [ \novel ] $, the fact that $ \aext
  $ is regular implies that $ c \sim^{\aextb} c _ 1 $ for some $ c _ 1 \in \mbbc
  $.  The richness assumption ensures that we may choose $ c _ 1
  $ from the complement of $ C $.  Then, since neither $c$ nor
  $c_{1}$ belong to $ C _ 1 \defeq C \bs \{c\} $, $c \sim^{\aextb}
  c_{1}$ implies $ \aextb_{C} = \aextb _ { C _ 1 \cup \{c_{1}\} }$.  If $ c
  $ is the unique member of $ C \cap [ \novel
  ]$, then the proof is complete.  Otherwise, using the fact that $ C
  $ is finite, we may proceed by induction until we obtain a set $ C _ n
  $ such that $ C _ n \cap [\novel ] $ is empty and $ D \defeq C _ n \cup \{ c _
  1 , \dots , c _ n \} $ belongs to $ \mbbd
  $.  Part \ref{item-preserving} of \cref{def-extension} then implies $ \aextb _
  { D } = \preceqb _ { D } \cap Y^{2}$, so that, since
  $\ext$ is improper, $\aextb_{D}=\extb_{D}$.  Finally, since $C \sim ^{\aextb}
  D$, $ \aextb _ { C } = \aextb_{D} $, as required.
\end{proof}

\section{Axioms and main theorem}\label{sec-axioms-theorem}

\paragraph{The basic axioms\hskip-10pt}~of \gsii, which we rewrite in terms of
extensions, are the following. In each of these axioms, $\ext $ is an arbitrary
$Y$-extension of $\preceq_{\mbbd}$.

\begin{taggedblank}{$\textup{A}0$}[Transitivity axiom for $\ext$]\label{T}

  For every $D \in \mbbdpp $, $ \ext_{D} $ is transitive.

  \end{taggedblank}

\begin{taggedblank}{$\textup{A}1$}[Completeness axiom  for $\ext$]\label{K}

  For every $ D \in \mbbdpp $, $\ext_{D} $ is complete.

\end{taggedblank}

\begin{taggedblank}{$\textup{A}2$}[Combination axiom  for $\ext$]\label{C}

  For every disjoint $ C, D\in \mbbdpp $ and every $x,y\in Y$, if
  $ x \ext_{C} y $ and $ x \ext_{ D } y $, then $ x \ext_{ C \cup D} y
  $$;$ and  if $x \sext_{C} y$ and $x \ext_{D} y$, then $x\sext _{C \cup D } y$.

\end{taggedblank}

\begin{taggedblank}{$\textup{A}3$}[Archimedean axiom  for $\ext$]\label{A}

  For every disjoint $ C,D \in \mbbdpp $ and every $x,y\in Y $, if
  $x \sext_{D} y$, then there exists $ k \in \posint $ such that, for every
  pairwise disjoint collection $\lb D _ j : \text{ $ D _ j \sim ^ \supext D $
    and $ C \cap D _ j  =\emptyset$}\rb _ 1 ^ k$ in $\mbbdpp$,
  $ x \sext_{ C \cup D_{1} \cup \cdots \cup D_{k} } y.$

\end{taggedblank}




\paragraph{The diversity axioms\hskip-10pt}~that now follow require that $\mbbd$
is sufficiently rich to support $Y$-extensions $\ext$ with a variety of
\emph{total} orderings: \ie\ complete, transitive and antisymmetric
($x \ext_{D} y $ and $y \ext_{D} x$ implies $x = y$). 
Let $ \total ( \ext )$ denote the set $ \lb R : \text{for some
  $ C \in \mbbdpp $, $ R = \extb _{ C }$ is total}\rb$ of of total orders that
feature in $ \ext $. For $ k = 4 $, the following axiom is a restatement of the
diversity axiom of \gsii.
 
\begin{taggeddiv}{\hskip-5pt}[$k$-Diversity axiom]
  For every $ Y \subseteq X $ of cardinality $ n= 2, \dots, k $, every regular
  $ Y$-extension $\ext$ of $ \preceqb _{ \mbbd }$ is such that
  $\countof \total(\ext) = n!$\hskip.5pt.
\end{taggeddiv}
We say \emph{$k$-diversity holds on $Z$} if the axiom holds with $Z$ in the
place of $X$. By lowering the bar for the required number of total orders, the
following axiom substantively weakens \threediv\ and \emph{a fortiori} \fourdiv.

\begin{taggedblank}{$\textup{A}4$$'$}[Partial \threediv]\label{p3d}
  For every $ Y\subseteq X $ with cardinality $ n = 2 $ or $  3 $, every regular
  $X$-extension $ \ext $ of $ \preceq _{ \mbbd }$ is such that $ \countof
  \total ( \ext ) \geq n $.

\end{taggedblank}


\Cref{eg-zaslavski} of the appendix provides an example of $\preceq_{\mbbd}$
satisfying \ref{T}–\ref{p3d}, but not \threediv. \Parthreediv\ is clearly
stronger than \twodiv\ and, moreover, it plays the dual role of guaranteeing
uniqueness of the representation and allowing us to avoid restrictions on the
cardinality of $ X $. \Cref{eg-lexicographic} shows that \parthreediv\ is the
weakest axiom with these properties.  Moreover, the observation below shows
that, in our setting, \parthreediv\ is equivalent to

\begin{taggedblank}{\textup{A}$4$}[Conditional-\twodiv]\label{c2d}
  For every three distinct elements $ x , y , z \in X $, within one of the sets
  $ \{ D' : x \prec _{D ' } y \}$ and $ \{ D' : y \prec_{D ' } x \} $ there
  exists  both $C$ and $D$ such that $ z \prec _{ C } x $ and
  $ x \prec _{ D } z $. If $ \countof X = 2 $, then \twodiv\ holds on $X$.

\end{taggedblank}

\begin{theoremEnd}[link to proof]{observation}\label{obs-c2d}
  For $\preceq_{\mbbd}$ satisfying \ref{T}--\ref{A}, \condtwodiv\ and
  \parthreediv\ are equivalent.
\end{theoremEnd}
\begin{proofEnd}\label{proof-obs-c2d}
  This follows directly from \cref{prop-c2dQ} and the construction of
  $\preceq_{\mbbj}$.
\end{proofEnd}
\paragraph{The prudence axiom \hskip-10pt} that follows is our final requirement. It
is distinguished by the fact that it imposes structure on novel extensions. As
we will see in the proof of the main theorem (see \cref{lem-insep}), novel
extensions are characterised by a cardinal notion. In contrast, the notions of
testworthiness and perturbation that we now introduce are ordinal in nature.

\begin{definition*}\label{def-testworthy}

  A proper extension $ \ext $ of $\preceqb_{\mbbd}$ is \emph{testworthy} if it satisfies
  \textup{\ref{K}--\ref{A}} and, for some $D\in \mbbd$ such that $\ext_{D}$
  is total, $ \extb _{ \novel }$ is the inverse of $\extb_{D}$.\footnote{Recall
    that the inverse $ \ext _{ D } ^{ - 1 }$ of $ \ext _{ D }$ satisfies
    $ x \ext _{ D } ^{ - 1 } y $ if, and only if, $ y \ext _{ D } x$.}
\end{definition*}
We now introduce perturbations. The motivation is that, for any
given extension $\ext$, the predictor knows $\ext_{\mbbd}$ and freely chooses
$\ext_{\novel}$, but the rankings at other databases are somewhat arbitrary. For
although the rankings the predictor associates with members of
$\mbbdp\bs (\mbbd \cup \{\novel\}) $ are constrained, what
matters is not the actual rankings, but rather their potential for consistency
with the axioms.
\begin{definition*}
  Let $\ext$ and $\aext$ be extensions of $\preceqb_{\mbbd}$. $\aext$ is a
  \emph{perturbation} of $\ext$ if $\aextb_{\novel } =
  \extb_{\novel}$. Moreover, $\aext$ is a \emph{nondogmatic} perturbation if
  $\countof \total (\aextb) \geq \countof \total ( \extb)$.
\end{definition*}
We are interested in perturbations that are nondogmatic because may reveal
intransitivities that the predictor's inexperience conceals.
The prudent predictor may exploit them and revise her model \emph{before new cases arrive}.
\begin{prudence*}\label{P}

  For every $Y\subseteq X$ with cardinality $ 3$ or $ 4 $, every testworthy
  $ Y $-extension of $ \preceqb _{ \mbbd }$ that is novel has a
  nondogmatic perturbation that satisfies  \ref{T}–\ref{A}.
\end{prudence*}

Given that the extensions we consider are nonrevisionistic, it is natural to ask
whether \ref{T}--\ref{A} are superfluous in the presence of \fourpru. Our response
is twofold. Firstly, in practice, we would expect \ref{T}--\ref{A} to hold more
frequently than \fourpru\ which is more cognitively demanding. Secondly, as we
will see in the proof of the main theorem, when $ \mbbt $ is infinite, for some
$ Y \subseteq X $, the set of testworthy $ Y $-extensions that are novel may be
empty. For every such $Y$, the following observation confirms that \fourdiv\
holds on $Y$.
\begin{theoremEnd}[link to proof]{observation}[on testworthy
  extensions]\label{obs-testworthy}
  Let $\preceq_{\mbbd}$ satisfy \ref{T}–\ref{c2d}. For every $Y\subseteq X$ of
  cardinality $ 3$ or $ 4$, the set of testworthy $Y$-extensions is
  nonempty. If, for some $Y$, every testworthy $Y$-extension is regular, then
  \fourdiv\ holds on $Y$.
\end{theoremEnd}
\begin{proofEnd}\label{proof-obs-c2d}
  Respectively, these two statements follow via \cref{prop-central-testworthy}
  and \cref{lem-test-empty-fourdiv}.
\end{proofEnd}
It is also natural to  ask whether \fourpru\ is simply requiring that, on $Y$ such that
$\preceq_{\mbbj}$ fails to satisfy \fourdiv, there exists a testworthy
$Y$-extension that is novel and  satisfies \fourdiv. In the proof of the
theorem that now follows, we show that this is not the case.

\paragraph{The main theorem\hskip-5pt} that now follows involves real-valued
function $ \mathbf{v} $ on the product $X \times \mbbc $. We view $\mathbf{v}$
as a matrix and $ \mathbf{v} ( x , \cdot )$ as one of its rows.  The matrix
$ \mathbf{v} $ is a \emph{representation of $ \preceq _{ \mbbd }$} whenever it
satisfies
\label{sec-main}
 \begin{linenomath*} 
\begin{equation}\notag\label{eq-rep-main}
  \left\{
  \begin{array}{l}
    \text{for every $ x , y \in X$ and every $ D \in \mbbd $,}\\
    x \preceq _ D y \quad \text{if, and only if,} \quad \sum _ { \, c \,\in\, D} \mathbf{v} ( x
    , c ) \leq \sum _ {\, c \,\in\, D } \mathbf{v} ( y , c )  .
  \end{array}\right.
\end{equation}
\end{linenomath*}
The matrix $\mathbf{v}$ \emph{respects case equivalence} (with respect to
  $\preceq_{\mbbd}$) if, for every $c,d\in \mbbc$, $c \sim^{\star} d$ if,
and only if, the columns $\mathbf{v}(\cdot,c)$ and $ \mathbf{v}(\cdot,d)$ are equal.

\begin{theorem}[Part I, Existence]\label{thm-main}
  Let there be given $X$, $\mbbcp$, $\preceqb_ \mbbd$ and
  associated extensions, as above, such that the richness condition holds. Then
  \ref{ax-main} and \ref{wrap-main} are equivalent.

\begin{enumerate}[label=\textup{(\ref{thm-main}.\roman*)}]

\item\label{ax-main}

  \ref{T}--\ref{c2d} and \fourpru\ hold for $\preceq_{\mbbd}$.
  
\item\label{wrap-main} There exists a matrix
  $ \mathbf{v} : X \times \mbbc \rightarrow \R $ satisfying \ref{rep-main} and \ref{rows-main}$\,:$
  \begin{enumerate}[label=\textup{(\ref{thm-main}.\alph*)}]
  \item\label{rep-main} $ \mathbf{v} $ is a representation of
    $ \preceq _ { \mbbd }$ that respects case equivalence$\,;$

  \item\label{rows-main} no row of $\mathbf{v}$ is dominated by
    any other row, and for every three distinct elements $x,y, z \in X$ and
    $\lambda \in \R$, $\mathbf{v} (x, \cdot) \neq \lambda
    \mathbf{v}(y,\cdot) + (1-\lambda) \mathbf{v}(z,\cdot)$\,.\footnote{Observe that
$\mathbf{v}(x,\cdot)- \mathbf{v}(z,\cdot)$ and
$\mathbf{v}(y,\cdot)-\mathbf{v}(z,\cdot)$ are noncollinear if, and only if,
the affine independence condition of \ref{rows-main} holds.}

\end{enumerate}

\end{enumerate}                 %
\end{theorem}


Our uniqueness result is identical to that of \gsii.  \setcounter{theorem}{0}
\begin{theorem}[Part II, Uniqueness]
  If \ref{ax-main} $[$or \ref{wrap-main}$]$ holds, then the matrix
  $ \mathbf{v} $ is unique in the following sense$\,:$ for every other matrix
  $ \mathbf{u} : X \times \mbbc \rightarrow \R $ that represents
  $\preceqb_{\mbbd}$, there is a scalar $ \lambda > 0 $ and a matrix
  $ \beta : X \times \mbbc \rightarrow \R$ with identical rows (\ie\ with
  constant columns) such that $ \mathbf{u} = \lambda \mathbf{v} + \beta$.
  
\end{theorem}



The proof of \cref{thm-main} appears in online appendix \ref{sec-proof-main}. It
relies upon a translation from the abstract database/memory set up of the model
to the setting of rational vectors similar to \gsii. We show that the translated
theorem \cref{thm-mainQ} (see \cref{sec-proof-main}) is equivalent to
\cref{thm-main}. The proof of \cref{thm-mainQ} is then the subject of online
appendix \ref{sec-proof-mainQ}.

\paragraph{We appeal to the following corollary}\hskip-7pt
when we connect the present framework to the concept of arbitrage in finance. Central to this connection is
\begin{definition} 

  For $ Y \in 2 ^ { X } $, the matrix
  $ v^{ \dd } : Y^{2}\times \mbbc \rightarrow \R $ satisfies the \emph{Jacobi
    identity} whenever, for every $ x , y , z \in Y $, the rows satisfy
  $ v ^{ \xz } = v ^{ \xy } + v ^{ \yz }$.
\end{definition}
In \cref{lem-insep}, of \cref{sec-proof-mainQ}, we show that, for a given
extension $\ext$, \ref{K}--\ref{A} and \twodiv\ yield a \emph{pairwise
  representation} $v^{\dd}$ of $ \ext $.\footnote{That is, for every $x,y\in X$
  and $D \in \mbbd$, $x \extb_{D} y$ if, and only if
  $ \sum_{c \,\in D}v^{\xy}(c) \geq 0$.} 

   
\begin{theoremEnd}[link to proof]{corollary}[a characterisation of
  prudence]\label{cor-foureq}
   Let the number of case types be finite and let $\preceq_{\mbbd}$ satisfy
   \ref{c2d}. Then $\preceq_{\mbbd}$ satisfies \fourpru\ if, and only if,
   $\preceq_{\mbbd}$ has a pairwise representation $v^{\dd}$ that satisfies the
   Jacobi identity. Moreover, for every other pairwise representation
   $u^{\dd} $, there exists $\lambda >0$, such that $u^{\dd} = \lambda v^{\dd}$.
 \end{theoremEnd}
 \begin{proofEnd}
   This follows from \cref{thm-foureq}, \cref{lem-induction} and the fact
   that, via \cref{lem-test-empty-fourdiv}, \fourpru\ implies \ref{T}--\ref{A} when
   the number of case types is finite. \label{proof-cor-foureq}
\end{proofEnd}
\section{Discussion}\label{sec-discussion}
We begin by restating the existence part of the main theorem of \gsii.
\begin{theorem*}[Existence]\label{thm-gsii}
  Let there be given  $X$, $ \mbbc $ and $\preceqb_ \mbbd$, as above, such that the
  richness condition holds. Then \ref{ax-gsii} and \ref{wrap-gsii} are equivalent.

\begin{enumerate}[label=\textup{(\roman*)}]

\item\label{ax-gsii}

  \textup{\ref{T}--\ref{A}} and \textup{\fourdiv} hold for $\preceq_{\mbbd}$.

\item\label{wrap-gsii} There exists a matrix
  $ \mathbf{v} : X \times \mbbc \rightarrow \R $ satisfying \ref{rep-gsii} and \ref{rows-gsii}$\,:$
  \begin{enumerate}[label=\textup{(\alph*)}]
  \item\label{rep-gsii}
  $ \mathbf{v} $ is a representation of $ \preceq _ { \mbbd }$ that respects case equivalence$\,;$

\item\label{rows-gsii} if $\countof X < 4$, then no row is dominated by an
  affine combination of the other rows, and for every four distinct elements
  $x,y,z,w \in X$ and every $\lambda , \mu, \theta \in \R$ such that
  $ \lambda +\mu + \theta = 1$,
  $\mathbf{v}(x,\cdot ) \not \leq \lambda \mathbf{v}(y,\cdot )+\mu
  \mathbf{v}(z,\cdot)+ \theta \mathbf{v}(w,\cdot)$.
\end{enumerate}  

\end{enumerate}
\end{theorem*}

Although diversity axioms play an important technical role, they are not
obviously behavioural. Instead, diversity axioms impose restrictions on what is
beyond the predictor's control and on what is central to inductive inference:
experience. Our main contention is that $ \current $ may not be so rich as to
support $\preceqb_{\mbbd}$ satisfying \fourdiv. That is to say, there may exist
$ Y \subseteq X $ such that $ \countof Y = 4$, and such that the data is
insufficiently rich to support all $ 4 ! = 24 $ strict rankings. A casual
comparison of condition \ref{rows-main} and \ref{rows-gsii} confirms that the
present framework achieves the main purpose of accommodating the less
experienced.

For the remainder of this discussion, we take both $X$ and the set $\mbbt$ of
case types to be finite and of cardinality $m$ and $n$ respectively. Via case
equivalence, for any $\preceq_{\mbbd}$ (satisfying \ref{wrap-main} or
\ref{wrap-gsii}) we may efficiently summarise rankings using a real-valued
$m\times n$ matrix $\mathbf{v}$ on the product of $X$ and case types $ \mbbt$.

\paragraph{Comparing the complexity of \condtwodiv\ and \fourdiv,}\hskip-8pt in
the presence of the other axioms, provides a measure of the value of
experience. As an estimate, we compare condition \ref{rows-main} with
\ref{rows-gsii} of Gilboa and Schmeidler's theorem.  Verifying \ref{rows-gsii}
involves checking $n$ affine dominance constraints: one for each case type. This
is well-known to be equivalent to the complexity of linear programming with real
variables \textup{\citep{DR-Linear_programming}}. In the absence of knowledge
regarding the sparsity of $\mathbf v$, the fastest algorithm for achieving this
is of order $n^{3}$ \textup{\citep{LS_Linear_programming}}. Since this holds for
every subset of four distinct eventualities, checking \ref{rows-gsii} is of
order $\binom{m}{4}n^{3}$.  In contrast, verifying
$\mathbf{v}(x,\cdot )\not \leq \mathbf{v}(y,\cdot)$ takes at most $n$ steps for
each of the $\binom{m}{2}$ subsets of $2$ distinct elements. Likewise, for every
three distinct elements $x,y,z$ in $X$, checking for noncollinearity of two
vectors takes at most $n$ steps. Thus, a na\"{i}ve algorithm for checking
\ref{rows-main} is of order $\binom{m}{3}n$. Even for $m = 4 $ and $n = 4$, the
difference is stark: $\binom{m}{3}n = 16$ versus $ \binom{m}{4} n^3 = 64$. (The
threshold $4$ is important as we show in \cref{prop-fourdiv-empty} below.)

\paragraph{A robust test of prudence\hskip-7pt} is available precisely when
\ref{T}--\ref{c2d} hold and \fourdiv\ fails. For, when $n<\infty$, the existence
of a similarity representation satisfying \cref{wrap-main} is equivalent to
\fourpru. The test is robust in that, generically, predictors that fail to check
for the arrival of new cases also fail to satisfy \cref{wrap-main}. This is
because, the row differences
$v^\xy = - \mathbf{v}(x,\cdot) + \mathbf{v}(y,\cdot)$, the Jacobi identity
$v^{\xz} = v^{\xy} + v^{\yz}$ holds for every $x,y,z \in X$. This in turn
implies that the hyperplanes $H^{\{x,y\}}, H^{\{y,z\}}$ and $H^{\{x,z\}}$, to
which the row differences are normal, are congruent. Congruence implies the
hyperplanes are not in general position. (When $n = 2$, the same argument can be
made in terms of extensions: see \cref{eq-matrix-3} of the proof of
\cref{thm-foureq}.) In other words, there is a zero (conditional) probability of
the predictor striking lucky and appearing to be prudent when in fact they are
not engaging in this form of second order inductive inference.

\paragraph{Novelty, memory and transitivity.} An early test and taxonomy of
intransitive behaviour is due to
\citet{Weinstein-Intransitivity,Weinstein-Transitivity}. Weinstein points out
that intransivity can sometimes be rational in complex situations.
Interestingly, he shows that young are people significantly less transitive in their
choices. He also points out that the law itself is designed to accommodate
irresponsible under-age decisions. In the psychology literature,
\citet{BR-Novelty_and_intransitivity} show that presenting novel objects is more
likely to trigger intransitive choices.

More recently, \citet{Enkavi-Hippocampal_dependence} provide evidence that
people with a specific form of memory impairment (lesions in the hippocampus of
the brain)
are significantly more likely to violate transitivity in pairwise choices of
chocolate bars: even though they rank numbers transitively.\footnote{The
  hippocampus is associated with learning and
  memory. \Citet{Hassabis-Hippocampal_dependence} and
  \citet{Schacter-Hippocampal_dependence} present evidence showing that the
  hippocampus plays an important role in imagining future experiences on the
  basis of past ones. \Citet{Enkavi-Hippocampal_dependence} go further by
  showing that it plays a role in the value-based decision making framework of
  \citet{Rangel-Value-based_neurobiology}.} Whilst this literature does not yet
offer a direct test of the present model, it supports the case-based framework
of constructing preference as well as our premise that violations of
transitivity are often driven by novelty, or equivalently impaired memory.

In line with the case-based approach, the experimental evidence of
\citet{Enkavi-Hippocampal_dependence} suggests that agents are constructing
their preferences on the basis of past experience. Moreover, it seems natural to
interpret agents with hippocampal impairment as inexperienced predictors. The
fact that impaired agents then make intransitive decisions is very much in line
with what our model predicts as they are, in effect, facing a novel situation
and are required to construct their preferences on the fly. It appears that
impaired agents are also failing to be prudent, though it is not obvious that
chocolates warrant the additional neural computation that accompanies \fourpru.
\begin{remark*}
  A closer look at the relationship between the proportion $\rho$ of hippocampal
  impairment and the percentage $\sigma$ of intransitive choices \citep[in
  figure 2 of][]{Enkavi-Hippocampal_dependence} suggests another interpretation.
  For $\rho $ above $ \frac{1}{4}$, $\sigma $ is above $ 20\%$$:$ twice as high as
  it is for $\rho < \frac{1}{4}$. Memories and the rankings they generate are
  latent variables to the observable $\rho$ and this threshold is where
  \condtwodiv\ fails to hold and our model breaks down. The $16$ cases in the
  data are thus partitioned into three groups$:$ two that satisfy \threediv\
  ($\rho < 0.05 $ and $\sigma< 5\%$)$;$ $12$ intermediate cases that satisfy
  \condtwodiv\ ($0.05 \leq \rho \leq 0.25 $ and $5\%\leq \sigma \leq 10\%$)$;$ and
  $2$ severe cases that fail to satisfy \condtwodiv\ ($0.25 < \rho$ and
  $10\% < \sigma$).
\end{remark*}

\paragraph{The success or failure of startups. \hskip-7pt} Inexperience raises
significant barriers to entry. Overcoming these barriers is either the result of
making mistakes and learning by doing ``on the fly'' or the result of being
prudent. Which form of second-order induction bears out in practice will depend
on many factors.

The following proposition confirms our thesis that experienced predictors (\ie\
those that satisfy \fourdiv) have indeed encountered a high number of case
types. As the main theorem of \gsii\ shows, experienced predictors have no need
for the additional structure of extensions. Unless the prediction problem
changes (\eg\ new eventualities become relevant),  they have no need to engage in second order induction. This saving in
cognitive effort is the prize that experience confers.

\begin{theoremEnd}{proposition}[experience and case types]\label{prop-fourdiv-empty}
  If $\preceqb_{\mbbd}$ satisfies \ref{K}--\ref{A} and \fourdiv, then
  $ n \geq \min \{4, m\} $, and, for every $Y \subseteq X$ of cardinality $m'$
  and regular $Y$-extension $\ext $, the number $n'$ of equivalence classes of
  $\sim^{\ext}$ satisfies $ n' \geq \min \{4,m'\}$.
  \end{theoremEnd}
  \begin{proofEnd}
    Via \cref{lem-axiomsQ}, \ref{K}--\ref{A} and \fourdiv\ hold for
    $\preceq_{\mbbd}$ if and only if \ref{KQ}--\ref{AQ} and \fourdiv\ hold for
    $\preceq_{\mbbj}$. Let $Y\subseteq X$ be of cardinality $m' = 1, 2, 3$ or
    $4$ and let $\ext$ be a regular $Y$-extension. Via \cref{lem-insep}, there
    exists a pairwise representation $v^{\dd} $. For $m' = 1$, $n' = 1 $ because
    $\ext_{J}$ is constant on $\mbbjp$.  For $m' = 2$, $n' \geq 2$, since via
    part \ref{D-insep} of \cref{lem-insep}, $G^{\xy}$ and $G^{\yx}$ are both
    nonempty.

    By way of contradiction, first suppose $n' = 2$ and $m' \geq 3$. Via
    \cref{eg-zaslavski} of appendix \ref{sec-proof-mainQ},
    $\total(\ext) \leq 4$. In contrast, \fourdiv\ requires $\total(\ext) = 6$.
    The remaining case is where $n' = 3$ and $m' \geq 4$. If the rank
    $\mathbf r$ of $v^{\dd}$ satisfies $\mathbf r \geq 3$, then the kernel
    $A^{Y}$ of $v^{\dd}$ is zero-dimensional. Then $0$ is the unique element of
    $ A^{Y}$. Thus, the positive kernel $A^{Y}_{\mdoubleplus}$ of $v^{\dd}$ is
    empty. Then Zaslavski's theorem implies that $\total(\ext)< 4!$, so that
    \fourdiv\ fails to hold.  If $\mathbf r \leq 2$, then an application of the
    rank version Zaslavski's theorem (in particular \cref{eq-zaslavski-4} with
    $\acute{\mathbf r } = \mathbf r = 2$) yields
  \begin{linenomath*}
    \[ \lvert \mc G_{\mdoubleplus} \rvert \leq 1 - 6 + 15 - 20 + 15 + 6 + 1 =
      12. \]
\end{linenomath*}
Thus, once again \fourdiv\ fails to hold. Thus $n' \geq \min \{4, m'\}$, as
required.  Finally, since $Y\subseteq X$, $m\geq m'$, and, via part
\ref{item-dimension} of \cref{def-extension}, $n \geq n'$.
\end{proofEnd}
In contrast, \condtwodiv\ implies no restrictions on the cardinality of
$\mbbt$ beyond $n \geq 2$ and this is also a virtue of \twodiv.

\paragraph{When is prudence worth the trouble?}
The simple answer to this question is: when revising a model ``on the fly'',
once a novel case arrives, is costly.  The following is our main example of such
a setting.
\begin{example}\label{eg-arbitrage}
  Consider a fair market maker of zero-coupon (treasury) bonds.\footnote{Similar
    to a fair insurer, the fair market maker sets the market spread to zero.}
  The compound-interest formula for the accumulation process of such a bond
  is
  \begin{linenomath*}
    \[a^{\xy}= \left(1+r^{\{x,y\}}\right)^{-x + y}, \]
  \end{linenomath*}
  where $r^{\{x,y\}}$ is the implied yield on a forward contract that accrues
  interest between dates $x$ and $y$. If $ x$ is later than $y$, then the
  contract is to sell, and the market maker pays this yield, so that
  $r^{\{y,x\}} = r^{\{x,y\}}$.  Let $X \subseteq \R_{\mplus}$ index a suitable
  sequence of trading dates with $0 \in X$ being the spot date.  It is well
  known that the (normalised) spot bond price
  $x \mapsto b(x) = \left(1+r^{\{0,x\}}\right)^{-x} $ is arbitrage-free if, and
  only if, the log-accumulation process satisfies
  \begin{linenomath*}
  \begin{equation}\label{eq-no-arbitrage}
    \text{for every $x,y,z \in X$,}\quad   \log a^{\xy} = \log a^{\xz} +  \log a^{\zy}
    \,.\footnote{To see this, suppose that, for some
      $x< y$, she sets $ a^{\xy} < a^{\xz}a^{\zy}$.  Another trader would do well
      to sell the forward contract $\xy$, buy the spot contract $\xz$ and sell the
      spot contract $\zy$.  A risk-free arbitrage opportunity is also
      available if the reverse inequality holds.}
  \end{equation}
\end{linenomath*}
This no-arbitrage condition is a special case of the Jacobi
identity. We now explain how a market maker might infer the
accumulation process from past cases.

Let each case $c \in D^{\star}$ consist of market-relevant data from a given
time interval (a block of time periods) in the past. These blocks are chosen so
that, for any finite resampling $D $ of cases from $D^{\star}$, the sequence of
cases that makes up $ D $ is exchangeable. The free case $\novel $ has no
additional structure beyond that of
\cref{sec-model}.  
Next, for every finite resampling $D$ and every date $x$ and $y$, let
$x \preceq_{D} y$ if, and only if, in answer to the question ``At which date
will the price be higher?'' the market maker finds that $y$ is more plausible
than $x$.

For the following corollary, we introduce the empirical implied yield function
that maps triples $x \times y \times c $ to $ r^{\{x,y\}}_{c} \in \R$.  This
function is characterised by three conditions: for time intervals of length
zero, the yield is zero; fair pricing; and case equivalence.  These are
respectively formalised as follows: for every $x,y\in X$ and every
$c , d\in \mbbc$, $ r^{\{x,x\}}_{c} = 0$; $r^{\{y,x\}}_{c} = r^{\{x,y\}}_{c}$;
and $c \sim^{\star} d$ if, and only if, $r^{\{x,y\}}_{c} = r^{\{x,y\}}_{d}$. The
empirical bond price function $B : X \times \mbbd \rightarrow \R$ maps every
pair $x \times D$ to
\[ B(x,D) = \prod_{c\, \in D}\left( 1 + r^{\{0,x\}}_{c}
  \right)^{-\frac{x}{ \lvert D \rvert}}.
  \]

When $D^{\star}$ belongs to $\mbbd$, the number of case types is finite, and we have
\begin{corollary}\label{cor-arbitrage}
  Let $\preceq_{\mbbd}$ satisfy \ref{c2d}. Then $\preceq_{\mbbd}$ satisfies
  \fourpru\ if, and only if, there exists empirical implied yield and
  empirical bond price functions, such that
  \begin{linenomath*} 
    \begin{equation}\tag{$*$}\label{eq-bond-rep}
 \left\{
  \begin{array}{l}
    \text{for every $ x , y \in X$ and every $ D \in \mbbd $,}\\
      x \preceq_{D} y \quad \text{if, and only if,}\quad B(x,D) \leq B(y,D).
  \end{array}\right.
      \end{equation}
    \end{linenomath*} 
    Moreover, for every $D\in \mbbd$, the bond price
    $B(\cdot,D)$ is arbitrage-free.
\end{corollary}
\begin{proof}[Proof of \cref{cor-arbitrage}]
  Via \cref{thm-foureq}, there exists a pairwise Jacobi representation
  $v^{\dd}: X^{2} \times \mbbc \rightarrow \R$ such that, for every
  $c,d\in \mbbc$, $v^{\dd}(c) = v^{\dd}(d)$ if, and only if, $c \sim^{\star}
  d$. Moreover, $v^{\dd}$ is unique upto multiplication by a positive
  scalar. Then, for every $x, y, z \in X$, $v^{(x,x)}(\cdot) = 0$,
  $v^{\yx} = - v^{\xy}$, and $v^{\xy} = v^{\xz} + v^{\zy}$.

  Recalling that $0 \in X$, for every $x\in X$ and $D \in \mbbd$, let
  \begin{linenomath*}
    \begin{equation}\label{eq-bond-vxy}
      \textstyle B(x,D) = \exp\left( -\frac{1}{\lvert D \rvert} \sum_{c\,\in D} v^{(x,0)}_{c}
      \right).
      \end{equation}
\end{linenomath*}
For the proof of \eqref{eq-bond-rep}, recall that $x \preceq_{D} y$ if,
and only if, $\sum_{c\,\in D} v^{\xy}(c) \geq 0$. Since
$- v^{(y,0)}_{c} = v^{(0,y)}_{c}$ and, via the Jacobi identity,
$v^{(x,0)} + v^{(0,y)} = v^{\xy}$, we have:
\begin{linenomath*}
  \[\textstyle
    - \log B(x,D) + \log B(y,D) = \frac{1}{\lvert D\rvert} \sum_{c\,\in D}
    \left(v^{(x,0)}_{c} - v^{(y,0)}_{c}\right) = \frac{1}{\lvert D\rvert}
    \sum_{c\,\in D} v^{\xy}_{c}.
\]
\end{linenomath*}
It remains for us to confirm that the bond price is a suitable function of the
empirical yield function. For every $x,y \in X$ and $c \in \mbbc$,
$ \log a^{\xy}_{c} = -v^{\xy}(c)$: so that, as the solution to
$ (y - x )\log (1+r^{\{x,y\}}_{c}) = -v^{\xy}_{c} = v^{\yx}_{c}$, for
$x \neq y$,
\begin{linenomath*}
  \begin{equation}\label{eq-gross-yield-vxy}
    1+ r^{\{x,y\}}_{c} = \exp\left(\frac{v^{\yx}_{c}}{y-x}\right) =
    \exp\left( \frac{v^{\xy}_{c}}{x-y}\right) = 1+r^{\{y,x\}}_{c}.
    \end{equation}
\end{linenomath*}
We therefore observe that $r^{\{x,y\}}_{c} = r^{\{y,x\}}_{c}$. For $x = y$,
$v^{\xy}_{c}= 0$ ensures that we can take $r^{\{x,x\}}_{c} = 0$. Finally, note
that for $c \sim^{\star} d$, the property $r^{\{x,y\}}_{c} = r^{\{x,y\}}_{c}$ is
inherited from $v^{\xy}_{c} = v^{\xy}_{d}$, so that we have an empirical implied
yield function.

The fact that, for every $D$, $B(\cdot, D)$ is arbitrage-free follows by virtue
of the fact that $v^{\dd}$ satisfies the Jacobi identity. 
\end{proof} 
In the present setting, because we are modelling a normalised bond price, we
obtain a stronger uniqueness result relative to part II of \cref{thm-main}. If
$\tilde B$ is another function that satisfies the present corollary, then, for
every $D \in \mbbd$, the spot price $\tilde B(0,D) = 1$. Thus, via
\cref{eq-bond-vxy} and \cref{cor-foureq}, for some $\lambda >0$,
$\tilde B = \textup{e}^{-\lambda} B$.

We now point out an interesting implication of \condtwodiv\ the recent
prevalence of negative interest rates. \Wlog, fix $x<y$. Then, given \fourpru,
\twodiv\ implies that there exists $c,d \in \mbbc$ such that
$v^{\xy}(c) < 0 < v^{\xy}(d)$. This is equivalent to
$v^{\yx}(d) < 0 < v^{\yx}(c)$, and, via \cref{eq-gross-yield-vxy},
$r^{\{x,y\}}_{d} < 0 < r^{\{x,y\}}_{c}$. That is, \twodiv\ requires that the
market maker's data is rich enough to contain at least one case where the yield
between date $x$ and $y$ is negative (as well as one where it is
positive). 
\Condtwodiv\ extends this notion to require that
$r^{\{x,y\}}_{D} < 0 < r^{\{x,y\}}_{C}$ for some $C$ and $D$ such that
$r^{\{x,z\}}_{C}\cdot r^{\{x,z\}}_{D} >0$.
  \end{example}

  \paragraph{Discussion of second-order induction.}\vskip-8pt
  The market maker of \cref{eg-arbitrage} engages in second-order induction when
  she acts prudently. She reflects on her model by checking that the basic
  axioms of \gsii\ will continue to hold when a novel case arrives. By way of
  contrast, suppose  the bond price of the market maker is such that
  $\preceq_{\mbbd}$ is consistent with the basic axioms, but not \fourpru. Then
  when a novel case arrives, she may be exposed to arbitrage and need to
  respecify her entire model ``on the fly''. Such a step corresponds to the
  intermittent respecification of her similarity weighting function
  $\mathbf{v}(x,c)$, that \citet[p.10324]{AG-Second-order_induction}
  describe. In \citet{AG-Second-order_induction}, the ``leave-one-out''
  technique of cross-validating the model by omitting a case of each type is
  intuitively and operationally close to our inclusion of the free case
  $\novel$. The difference is that by allowing $\novel $ more degrees of
  freedom, our market maker can study novel extensions and peer into the future
  through the lens of her current model. She can exploit the intervals of time
  inbetween the arrival of novel cases by continuously engaging in second-order
  induction.

  Through an example, we now show that the present framework provides the
  flexibility to accommodate second-order induction without sacrificing the
  computational or normative advantages that additive similarity functions
  provide.


  \begin{example*}[second-order induction, \gsii, p.12] \label{eg-john_mary} Let
    $c$ denote a case where Mary chooses restaurant $x$ over restaurant $y$. In
    the absence of any further information, it is tempting to assume some
    similarity between John and Mary. The predictor then finds it plausible that
    John prefers $x$ to $y$ given $\{c\}$. A separate database $D$ contains no
    choices between $x$ and $y$. Thus, in the absence of further information,
    $x$ and $y$ appear equally likely based on $D$. Additivity of the similarity
    function (or \ref{C})) implies it is plausible that John prefers $x$ to $y$
    given $\{c\} \cup D$. The violation of \ref{C} arises when a more careful
    examination of the contents of $D$ reveals many choices between other pairs
    of restaurants where John and Mary consistently differ.
\end{example*}
Quine's notion of perceptual similarity \citep{Quine-Roots_of_reference} offers
a check on the predictor's inference about John's choice given $\{c\}$. John and
Mary may just as well be two drivers passing through an intersection at
different times. Although their situations are broadly speaking very similar, if
one faces a red light and the other a green light, their responses will
differ. In the restaurant setting, some pivotal information is omitted from $c$.
Observing that the evidence in $c$ in favour of $x$ over $y$ is somewhat weak, a
prudent predictor instead recasts $\{c\}$ as a database $C$ that combines past
observations with copies of the pivotal novel case $\novel$.  With a more
refined model, that explicitly allows for omitted variables, the predictor can
check to see if her model extends to higher dimensions without violating the
basic axioms.

\begin{remark}
Observe that predictors that only fail to satisfy \ref{C} can, with some
additional regularity conditions, still be represented by a nonlinear function
$u: X \times \mbbd \rightarrow \R $ such that for every $D \in \mbbd$ and every
$x,y \in X $, $x \preceq_{D} y$ if, and only if, $ u(x,D)\leq u(y,D)$
\citep[see][]{OCallaghan-Parametric_continuity}. But predictors that satisfy
\ref{T}--\ref{c2d} but not \fourpru\ can also be represented by such a function:
because $\preceq_{\mbbd}$ is complete and transitive for each $D \in \mbbd$. The
present framework allows us to disentangle the latter kind of predictor from
those who, for good reason, fail to satisfy the combination axiom. (See \gsii\
for  examples of such reasons.)
\end{remark}

\paragraph{On the veracity of false news.} How should a predictor check whether
her model consistently extends to higher dimensions (when novel cases arrive)?
If our model is a guide then, the most useful rankings that she might wish to
examine are those that are far from her own. This is because our definition of
testworthy extensions involves assigning to the novel case $\novel$ the inverse
of some total ranking $\preceq_{D}$.  This may offer some rationale for why
information that differs from our own is intrinsically valuable. Testworthy
extensions play a vital role in taming the complexity of our proof. It seems
plausible that something similar may be at play when agents encounter radically
different information from their own on social media: \emph{even if it is
  fake}. This may help to explain why false news is significantly more veracious
than real news online \citep{Vosoughi-Roy-Aral-Veracity}. The fact that real
news is typically closer to what we have observed in the past means that it is
of less value to the prudent predictor that finds it costly to imagine worlds
that are far from her own.

\bibliographystyle{apalike} \bibliography{references}
\makeatletter
\def\@seccntformat#1{Appendix\,\csname the#1\endcsname.\quad}
\makeatother
\begin{appendices}
\section{The proof of \cref{thm-main}}\label{sec-proof-main}



As in \gsii, we  translate the model into one where
databases are represented by vectors, the dimensions of which are
case types. To allow us to focus on aspects of the present model, 
proceed directly to rational vectors and present the axioms  and a corresponding
theorem (\cref{thm-mainQ}) which, as we confirm, holds if,
and only if,  \cref{thm-main} does. The proof of \cref{thm-mainQ} can be found
in \cref{sec-proof-mainQ}.

\paragraph{Case types as dimensions.}
From our definition of case types in  \cref{sec-model}, $ \mbbt = \mbbc_{/\sim^{\star}} $ and
$ \mbbt^\novel \defeq \mbbt\cup [ \novel ]$. Let $ \mbbtpp $ be a free variable
in $ \{ \mbbt , \mbbtp \}$.  When no possible confusion should arise, we use
$ \novel $ as shorthand for $ [ \novel ] $.  It is straightforward to show that
the following construction would work if instead we were to work with any
partition $\textup T$ of $\mbbc $ that is at least as fine as $\mbbt$. The
present construction is the one with the lowest feasible number $\countof \mbbt$
of dimensions.

\paragraph{Translation to  counting vectors.}
Let $ \nnint $ denote the set of nonnegative integers and $ \posint $ those that
are (strictly) positive.  Let $ \mbbi \subseteq \nnint ^ { \mbbt } $ denote the
set of counting vectors $ L : \mbbt \rightarrow \nnint $ such that
$ \lbc t : L ( t ) \neq 0 \rbc$ is finite and let $\mbbip$ denote the
corresponding subset of $ \nnint^{\mbbtp}$. Then let
\begin{linenomath*}
\begin{equation*}
\mbbipp = \left\{
\begin{array}{ll}
 \mbbi & \text{if, and only  if, $\mbbtpp= \mbbt$, and}\\
\mbbip  &\text{otherwise.}
\end{array}\right.
\end{equation*}
\end{linenomath*}
Modulo notation, the following construction is identical to \gsii.  For every
$D\in \mbbd$, let $L_{D}: \mbbt \rightarrow \nnint $ denote the function
$t \mapsto L_{D}(t) = \countof (D \cap t )$.  For each $D \in \mbbd$, let
$\preceqb_{L_{D}} \defeq \preceqb_{D}$.  We need to establish that
$ \preceqb _{ \mbbi } \defeq \langle \preceqb _{L } : L \in \mbbi \rangle $ is
well-defined. For every $L\in \mbbi$, the richness assumption (on $\mbbtp$)
guarantees the existence of $D\in \mbbd$ such that $L_{D}= L$. By definition,
$\sim^{\star}$ is such that, for every $C,D \in \mbbd$, $C\sim^{\star}D$ if, and
only if, $L_{C}= L_{D}$. Straightforward mathematical induction on the
cardinality of $C$ shows that $C\sim^{\star} D$ implies
$\preceqb_{C} = \preceqb_{D}$.  This construction of $ \preceqb_{\mbbi}$
ensures that the same notion of equivalence that we introduced in \cref{obs-reg-eq}
also applies here. Thus, $ \preceqb_{\mbbi} \equiv \preceqb _{\mbbd} $.

\paragraph{Translation to rational vectors.}
Similarly, let $\nnrat$ denote the nonnegative rationals and $\posrat$ those
that are (strictly) positive.  Take $ \mbbj \subseteq \nnrat^{\mbbt}$ to be the
set of rational vectors with $ \lbc t \in \mbbt : J ( t ) \neq 0 \rbc $ finite
and take $ \mbbjp $ to denote the corresponding subset of $\nnrat ^{ \mbbtp } $.
For each $ J \in \mbbj $, by virtue of the fact that $\posint$ is well-ordered
and $J$ has finite support, there exists (unique) minimal $ k_{J} \in \posint $
such that $ L_{J} \defeq k_{J} J $ belongs to $ \mbbi $.  Let
$ \preceqb_{J}\defeq \preceqb_ { L_{J} }$. (This definition acquires meaning
below once we translate and apply the combination axiom.)  In this way,
$ \preceqb _{ \mbbj } = \langle \preceqb_ J : J \in \mbbj \rangle $ is
well-defined, and we may introduce axioms for $\preceq_{\mbbj}$ directly: \ie\
without first introducing axioms for $\preceq_{\mbbi}$. We first demonstrate
that $\preceqb_{\mbbj}$ and $\preceqb_{\mbbd}$ are equivalent. First note that,
for every $I,J\in \mbbj$ such that $L_{I} = L_{J} $,
$\preceqb_{I}=\preceqb_{J}$. Then, let $L^{\prime} = L_{J}$ and take any $D$
such that $L_{D} = L^{\prime}$. Then $\preceqb_{J} = \preceqb_{D}$.  The reverse
embedding follows by virtue of the fact that $\mbbi\subset \mbbj$. Thus,
$ \preceqb_{\mbbj}\equiv \preceqb_{\mbbd}$.



\paragraph{Construction of extensions of $\preceqb_{\mbbj}$.}
We follow common practice by letting $ 2 ^{ X } $
denote the collection of nonempty subsets $ Y \subseteq X $.
For each $ Y \in 2 ^{ X } $, we will denote the set of regular, novel and
testworthy $Y$-extensions (of $\preceqb_{\mbbd}$ or $\preceqb_{\mbbj}$) by
$ \reg ( Y , \cdot )$, $ \nov ( Y , \cdot)$ and
$ \test ( Y , \cdot)$ respectively. Recalling that every
$Y$-extension is either regular or novel, let $ \Ext (Y,\cdot)$
denote the set of all $Y$-extensions. We now clarify
what it means to be an extension of $\preceqb_{\mbbj}$.

For each $t \in \mbbtpp$, we take  $\delta_{t} : \mbbtpp \rightarrow \R$ to be
the function satisfying $ \delta_{t} ( s ) = 1 $ if $ s = t $ and $ \delta_{t} ( s ) = 0 $
otherwise. (When $\mbbtpp$ is finite, these are simply the basis vectors for $\R^{\mbbtpp}$.)  When we wish to emphasise that the  vectors belong to in
$\R^{\mbbtp}$, then, for each $\mbbtp $, we will write $\delta_{t} ^{ \novel }$.
Let
\begin{linenomath*}
\begin{equation*}
 \mbbjpp  = \left \{
\begin{array}{ll}
\mbbj & \text{ if, and only if, $\mbbtpp = \mbbt$, and}    \\
 \mbbjp & \text{ otherwise.}
\end{array}\right.
\end{equation*}
\end{linenomath*}
For every $I \in \mbbj$ and $J \in \mbbjpp$, we write $I \equiv J $ whenever
$I = J $ or $J = I\times 0 $. (In the latter case, $J(t)= I(t)$ for every
$t\in \mbbt$ and $J(\novel )= 0 $.) This notion reflects the fact that, for the purposes
of the present model, such $I $ and $J$ are equivalent.
\begin{definition}\label{def-extensionQ}
 
$\extb =\langle \extb_{J}: J \in \mbbjpp \rangle$ is an
  \emph{extension}, and in particular a $Y$-\emph{extension}, of
  $ \preceqb _{ \mbbj }$ if, and only if, for some nonempty
  $ Y \subseteq X $ both the following hold
\begin{enumerate}

\item for every $J\in \mbbjpp$, $\extb_{J}\in \relations(Y)$,
  $ \nextb_{J} \defeq \extb_{J} \cap \extb_{J} ^ { -1 }$ and
  $ \sextb_{J} \defeq \extb_{J} \bs \extb_{J} ^ { -1 } $$;$

\item for every $ J \in \mbbj $ and $ L \in \mbbjpp $ such that $J \equiv
  L$, $\extb_{L} = \preceqb_{J} \cap (Y ^ 2) $. \label{item-preservingQ}
\end{enumerate}
An extension $\ext_{\mbbjpp}$ (of $\preceq_{\mbbj}$) is proper if $\mbbjpp=
\mbbjp$ and improper otherwise.  A proper extension is either regular or
novel. $\ext$ is novel if, for every $ s \in \mbbt $, there exists $ I $ in $
\mbbj$ such that, for $ J = I \times 0 $ (in $\mbbjp$), we have $\extb _ { J +
  \delta _{s} ^{ \novel } } \neq \extb _ { J + \delta _ { \novel} ^{ \novel }
}$\,.
 \end{definition} 

 For every regular $Y$-extension $ \ext$ of $\preceq_{\mbbd}$ such that $Y=X$,
 \cref{obs-reg-eq} implies $\extb\equiv \preceqb_{\mbbd} $. And, via
 $\preceqb_{\mbbj}\equiv \preceqb_{\mbbd}$ and transitivity of equivalence, we
 conclude that $\ext$ is equivalent to $\preceq_{\mbbj}$.  Two sets of
 extensions are isomorphic if there exists a canonical isomorphism between
 equivalent extensions. 


\begin{theoremEnd}[no link to proof]{lemma}[proof on \cpageref{proof-nov-iso}]
  \label{lem-nov-iso}

  For every $Y \in 2 ^ { X } $, $\reg(Y, \preceqb_{\mbbj}) $ is isomorphic to
  $ \reg(Y,\preceqb_{\mbbd})$ and
  $\nov ( Y , \preceqb_{\mbbj} ) $ is isomorphic to $ \nov ( Y , \preceqb _ \mbbd )$.


\end{theoremEnd}
\begin{proofEnd}
  \label{proof-nov-iso}

  We show that there exists a canonical embedding (a structure preserving
  injection) of $ \nov ( Y , \preceqb _ \mbbj ) $ into
  $ \nov ( Y , \preceqb _ \mbbd ) $. The fact that this map is also
  surjective follows from the fact that $ \nov ( Y , \preceqb _ \mbbd ) $
  can be embedded in $ \nov ( Y , \preceqb _ \mbbj ) $ in precisely the same
  way. The proof that the two sets of regular extensions are isomorphic follows
  via a similar argument plus the observation that every $Y$-extension is either
  regular or novel.

  Take $ \extb \in \nov ( Y , \preceqb _ \mbbj )$ and define
  $ \hextb = \langle \hextb_{C} : C \in \mbbdp \rangle$ via the property: for
  each $ C \in \mbbdp $, $ \hextb_{C}\defeq \extb_{J} $ if, and only if,
  $ L_{C}= L_{J} $, where, as before, $t \mapsto  L_{C}(t) $ counts the number of cases of
  type $t$ in $C$ and $ L_{J} = \kappa_{J} J \in \mbbip $ for some minimal
  $\kappa_{J} \in \nnint$. Now, for any
  $ \extb ' \neq \extb $ in $ \nov ( Y , \preceqb _ \mbbj ) $, there exists
  $J\in \mbbjp$ such that $\extb'_{J}\neq \extb_{J}$. If we define $ \hext' $
  analogously, so that it is equivalent to $\ext'$, then
  $ \hextb ' \neq \hextb $. As a consequence, the canonical mapping
  $\extb\mapsto \hextb$ is injective. If we can show that $ \hext $ does
  in fact belong to $ \nov ( Y , \preceqb _ \mbbd )$, then we have constructed
  the required embedding. The fact that $ \hext $ satisfies
  \ref{item-binary-rel} and \ref{item-preserving} of \cref{def-extension}
  follows immediately from \cref{def-extensionQ}. The proof that \cref{item-dimension}
  of \cref{def-extension} holds is as follows. Take any $ c , c ' \in \mbbcp $
  and $ D \in \mbbdp $ such that $ c \sim ^ \star c ' $ and
  $ c , c ' \notin D $. First, observe that
  $ D \cup \{c\} \sim ^ \star D \cup \{c'\} $, and moreover, for some $ t \in \mbbtp $
  we have $ c , c ' \in t $. Then, for every $t\in \mbbtp$,
  $ \lvert D \cup \{c\} \rvert = \lvert D \cup\{c'\} \rvert = L$ for some
  $ L \in \mbbip\cap \mbbjp $. Thus
  $ \hextb _ { D \cup \{c\} } = \hextb _ { D \cup \{c'\}}$, as required for $ \hext $
  to be an extension of $ \preceqb _{ \mbbd }$. Finally, via \cref{def-extensionQ},
  the definition of a novel extension ensures that the induced equivalence
  relation $\sim^{\extb}$ on $\mbbcp $ satisfies $ c \not \sim^{\extb} \novel $ for
  every $ c \in \mbbc $. Since $ \sim ^{ \hextb}$ inherits this property,
  $ \hext $ is novel.
\end{proofEnd}

\paragraph{Axioms and theorem.} 
We now restate the axioms for $Y$-extensions $ \ext $ of $ \preceqb  _ \mbbj $.
\begin{enumerate}[label=\textup{A\arabic*}$^\flat$]
 \setcounter{enumi}{-1}
  
\item\label{TQ}

  For every $ J \in \mbbjpp $, $\ext _ J $ is transitive on $ Y $.

\item\label{KQ}

  For every $ J \in \mbbjpp $, $ \ext _ J $ complete on $ Y $.

\item\label{CQ}

  For every $ I , J \in \mbbjpp $, every $ x , y \in Y $ and every
  $ \lambda , \mu \in \posrat $, if $ x \ext _ I y $  and
  $ x \ext_{J} y $, then $ x \ext_{ \lambda I + \mu J } y $$;$ moreover, if $x
  \sext_{I} y$ and $x \ext_{J} y$, then $x \sext_{\lambda I + \mu J}$.

\item\label{AQ}

  For every $ I , J \in \mbbjpp $ and every $ x , y \in Y $ if
  $ x \mathrel{\sext} _ J y $, then there exists $0 < \lambda < 1$ such that,
  for every $ \mu \in \mbb Q \cap ( \lambda , 1 ) $,
  $ x \mathrel{\sext _{ ( 1 - \mu ) I + \mu J }} y $.

\end{enumerate}

For $k = 2, 3, 4$, $k$-diversity is defined for extensions of $\preceq_{\mbbj}$
in exactly the same way. We continue to use the term $k$-diversity in this
setting.  The following are \condtwodiv\ and \parthreediv\ respectively.
\begin{enumerate}[label=\textup{A\arabic*}$ ^ \flat$,resume]
\item\label{c2dQ} For every three distinct elements $ x , y , z \in Y $, one of
  the two subsets $ \{ J' : x \prec _{J ' } y \}$ and
  $ \{ J' : y \prec_{J ' } x \} $ of $\mbbj $ contains both $I$ and $J$ such
  that $ z \prec _{ I } x $ and $ x \prec _{ J } z $. If $ \countof Y = 2 $,
  then \twodiv\ holds on $Y$.

\end{enumerate}
\begin{enumerate}[label=\textup{A\arabic*}$^{'\flat}$,resume]
  \setcounter{enumi}{3}
\item\label{p3dQ}
  For every $ Y^{\prime} \subseteq Y $ with cardinality $ n = 2 $ or $  3 $, every
  $Y^{\prime}$-extension $ \ext $ of $ \preceq _{ \mbbj }$ is such that $ \countof
  \total ( \ext ) \geq n$.
  

\end{enumerate}


A proper extension $ \ext $ of $ \preceqb _{ \mbbj }$ is \emph{testworthy} if
it satisfies \ref{KQ}--\ref{AQ} and, for some $J\in \mbbj$ such that
$\ext_{J\times 0 }$ is total, $\extb_{\novel} = \extb_{J\times 0}^{-1}$. Thus,
for each $ Y \in 2 ^ { X } $,
$ \test ( Y , \preceqb _{ \mbbd } ) \simeq \test ( Y , \preceqb _{ \mbbj
})$. For any pair of extensions $\ext $ and $\hext $, $\hext $ is a perturbation
of $\ext$ if $\extb_{\novel} = \hextb_{\novel}$. Moreover, $\hext$ is a
nondogmatic perturbation if
$\countof \total (\hext) \leq \countof \total (\ext) $.
\begin{enumerate}[label=\textit{4}-\textup{P}$^{\flat}$]
\item\label{PQ} For every $Y \subseteq X$ of cardinality $3$ or $4$, every
  testworthy $ Y $-extension of $ \preceqb _{ \mbbj }$ that is novel has a
  nondogmatic perturbation that satisfies \ref{TQ}–\ref{AQ}.
\end{enumerate}


The following result corresponds to claim 2 of \gsii. Its proof is a consequence
of mathematical induction and the combination axiom.
\begin{lemma}\label{lem-coneQ}

  If $ \ext _ \mbbjpp $ and $ \hext _ \mbbipp $ are equivalent and the latter
satisfies \ref{C}, then for every $ J \in \mbbjpp $ and every rational number
$ q >0 $, we have $ \extb _{ q J }  =  \extb _{J} $.

\end{lemma}
The fact that $\preceqb_{\mbbj}\equiv \preceqb_{\mbbd}$ immediately implies
that $\preceqb_{\mbbj}$ satisfies \ref{TQ}, \ref{KQ} and \ref{c2dQ} if, and
only if, the corresponding axiom holds for $\preceq_{\mbbd}$. In general, we
have the following result, which then also yields the equivalence for the
prudence axiom.

\begin{theoremEnd}[no link to proof]{lemma}[proof on \cpageref{proof-axiomsQ}]
  \label{lem-axiomsQ}

  For $ \extb _{ \mbbjpp } \equiv \hextb _{ \mbbdpp }$, $ \ext _ { \mbbjpp} $ satisfies
  \ref{CQ}--\ref{AQ} if, and only if, $ \hext_{\mbbdpp} $ satisfies
  \ref{C}--\ref{A}.
\end{theoremEnd}

\begin{proofEnd}
  \label{proof-axiomsQ}

  Fix $ \extb _ { \mbbjpp} \equiv \hextb _ { \mbbdpp}$ and assume that
  $ \hextb _{ \mbbdpp} $ satisfies \ref{C}. We show that $ \extb _ { \mbbjpp}$
  satisfies \ref{CQ}.  Fix $ x , y \in Y $ and $ J \in \mbbjpp $ such that
  $ x \ext _{ J } y $ and $ x \ext _{ J '} y $. Fix
  $ \lambda , \mu \in \posrat $ and let $ \kappa $ be the smallest positive
  integer such that both $ L \defeq \kappa \lambda J $ and
  $ L ' \defeq \kappa \mu J ' $ belong to $ \mbbipp $. Then, by
  \cref{lem-coneQ}, we have both $ x \ext _{ L } y $ and $ x \ext _{ L ' } y
  $. Moreover, for $D, D'$ such that $L_{D}= L$ and $L_{D'}= L'$ , we have
  $ x \hext _{ D } y $ and $ x \hext _{ D ' } y $ and, by \ref{C},
  $ x \hext _{ D \cup D ' } y $. Finally, since
  $ L_{D} + L_{D'} = \kappa (\lambda J + \mu J ' )$, one further application of
  \cref{lem-coneQ} yields $ x \sext _{ \lambda J + \mu J ' } y $, as required
  for \ref{CQ}.

  The proof that “\ref{C} implies \ref{CQ}” is \emph{mutatis mutandis} a special
  case of the above argument and ommitted.  We now assume $ \hextb _{ \mbbdpp} $
  satisfies \ref{C} and \ref{A} and prove that $ \extb _{ \mbbjpp }$ satisfies
  \ref{AQ}.  Fix $ x , y \in X $ such that $x \sext _{ J } y $ for some
  $ J \in \mbbj $ and take any $ J ' \in \mbbjpp $.  Then, by the construction
  of $ \extb _ \mbbjpp $, there exists $ L, L ' \in \mbbi $ such that
  $ j J = L $ and $ j ' J ' = L ' $ for some $ j , j ' \in \posint $.  By
  \cref{lem-coneQ}, $ \extb_{L} = \extb_{J} $ and
  $ \extb _ { L '} = \extb _ {J '} $.  Moreover, by construction, for some $D$
  and $D'$ such that $L_{D}=L$ and $I_{D'}= L'$, $ \hextb_{D} = \extb _ J $ and
  $ \hextb _ { D'} = \extb _ {J '} $.  We therefore conclude that
  $ x \hsext _{ D} y $, so that \ref{A} implies the existence of
  $ \kappa \in \posint $ and
  $\{D_{l}: D_{l}\sim^{\hextb} D\}_{1}^{\kappa}$ such that
  $ x \hsext _ { D_{1}\cup\cdots \cup D_{\kappa} \cup D '} y $.  Then, by the
  construction of $ \extb _{ \mbbjpp }$,
  $ x \sext _ { \kappa L_{D} + L_{D'}} y $.  Let
  $ \nu \defeq \frac { 1 } { \kappa j + j '}$ and take $ \lambda = \nu j '$, so
  that $0 < \lambda < 0 $ and $ 1-\lambda = \nu \kappa j $. In fact, since
  $ \lambda \in \mbb Q$, we have
  \begin{linenomath*}
    \[ K \defeq (1-\lambda ) J +\lambda J ' \in \mbbjpp .\]
  \end{linenomath*}  
  Simplifying, we obtain $ K = \nu ( \kappa L + L ')$.  Since
  $ \nu \in \posrat $ and $ \kappa L + L ' \in \mbbjpp $, \cref{lem-coneQ}
  implies $ \extb _ { K } = \extb _ { \kappa L + L ' }$.  This allows us to
  conclude that $ x \sext _ K y $.  Finally, take any
  $ \mu \in \mbb Q \cap ( 0 ,\lambda )$.  From basic properties of the real
  numbers, there exists $ \xi < 1 $ such that $ \mu = \xi\lambda $ and,
  moreover, $ \xi $ is rational.  Next, note that the definition of $ K $
  implies $ \xi ( K - J ) = \xi\lambda ( J ' - J ) $.  Adding $ J $ to each side
  of the latter and applying the definition of $ \mu $ yields
\begin{linenomath*}
  \[(1 - \xi ) J +  \xi K =  ( 1 - \mu ) J + \mu  J ' . \]
\end{linenomath*}
  Then, since $ x \sext _ J y $ and $ x \sext _ K y $, \ref{CQ} implies
  $x \sext _ { ( 1 - \mu ) J + \mu J '} y $, as required for \ref{AQ}.

Conversely, we now assume that $ \ext _{ \mbbjpp } $ satisfies \ref{CQ} and
\ref{AQ} and prove that \ref{A} holds.  Take $ D , D ' \in \mbbd $ such
that $ x \hsext_{D} y $ and any other $ D' \in \mbbd $. Let $L=L_{D}$ and
$L' = L_{D'}$.  Then, by construction, $ x \sext_{L} y $ and, by \ref{AQ}, there
exists $\lambda \in \mbb Q \cap ( 0, 1) $ such that
$ x \sext _ { (1 -\mu ) L +\mu L'}  y $.  Then, since $\mu $ is rational,
$\mu = \nicefrac { j } { k } $ for some $ j , k \in \posint $.  Let
$ q : = (1 - \mu ) /\mu = ( k - j )/ j $ and let $ \kappa = j q $, so that
$ \kappa = k - j $.  The fact that $ 0 < \mu < 1 $ ensures that
$ \kappa \in \posint $.  To complete the proof, we show that
$ x \sext _ { \kappa L + L'} y $, for then the existence of
$D_{1}, \dots , D_{\kappa}$ such that
$ x \hsext _ { D_{1}\cup \cdots \cup D_{\kappa} \cup D'} $ immediately follows.  Together
$ x \sext _ { (1 -\mu ) L +\mu L' } y $ and \cref{lem-coneQ} imply
$ x \sext _ { q L + L ' } y $.  Similarly, together $ x \sext _{ L } y $ and
\cref{lem-coneQ} imply $ x \sext _ { ( j - 1 )q L } y $.  Then, since
  $ ( j - 1 )q L + ( q L + L ' ) = j q L + L ' $ and $ \kappa = j q $, an
  application of \ref{CQ} yields $x \sext_{\kappa L + L '} y $, as required.
\end{proofEnd}
The matrix $\mathbf{v}: X \times \mbbt \rightarrow R$ is a \emph{representation of
  $ \preceq _{ \mbbj }$} whenever it satisfies
\begin{linenomath*}
\begin{equation}\tag{$\flat$}\label{eq-rep-main}
  \left\{
  \begin{array}{l}
    \text{for every $ x , y \in X$ and every $ J \in \mbbj $,}\\
    x \preceq_{J} y \quad \text{if, and only if,} \quad
\sum _ { \, t \,\in\, \mbbt} \mathbf{v} ( x
    , t ) J(t) \leq \sum _ {\, t \,\in\, \mbbt } \mathbf{v} ( y , t ) J(t)  .
  \end{array}\right.
\end{equation}
\end{linenomath*}

We observe that, via the definition of case types, there exists a representation
of $\preceq_{\mbbd}$ that respects case equivalence if, and only if there
exists a representation of $\preceq_{\mbbj}$. The above translation and results 
imply that \cref{thm-main} is equivalent to
\begin{theorem}\label{thm-mainQ}
  Let there be given $X$, $\mbbtp$, $ \preceqb _ { \mbbj } $ and associated extensions,
  as above. Then \ref{ax-mainQ} and \ref{wrap-mainQ} are equivalent.

\begin{enumerate}[label=\textup{(\ref{thm-mainQ}.\roman*)}]

\item\label{ax-mainQ}

 \ref{TQ}--\ref{c2dQ} and
 \ref{PQ} hold for   $\preceqb _ { \mbbj } $ on $X$.

\item\label{wrap-mainQ} There exists a matrix
  $ \mathbf{v} : X \times \mbbt \rightarrow \R $ that satisfies both$\,:$
  \begin{enumerate}[label=\textup{(\ref{thm-mainQ}.\alph*)}]
  \item\label{rep-mainQ}
  $ \mathbf{v} $ is a representation of $ \preceq _ { \mbbj }$$\,;$ and

\item\label{rows-mainQ} no row of
  $\mathbf{v}$ is dominated by any other row, and, for every three distinct elements
  $x,y, z \in X$, $\mathbf{v}(x,\cdot)-\mathbf{v}(z,\cdot)
  $ and $\mathbf{v}(y,\cdot)-\mathbf{v}(z,\cdot)$ are noncollinear (\ie\ linearly independent).
\end{enumerate}
\end{enumerate}
Moreover, $\mathbf{v}$ is unique in the sense of \cref{thm-main} part II, with
\ref{wrap-mainQ} replacing \ref{wrap-main} and $\mbbt$ replacing $\mbbc$.
\end{theorem}

\section{The proof of \cref{thm-mainQ}}\label{sec-proof-mainQ}
The present proof follows a similar structure to that of \gsii. That is, we
begin with the proof for the case of arbitrary (nonempty) $ X $ and finite
$ \mbbt $. We then show that we can ``patch'' the proof together to account for
the case where $0 <  \countof X < \infty = \countof \mbbt$. We omit the proof for the
case where both $X$ and $\mbbt$ are arbitrary since it is identical to that of \gsii.




\begin{step}[for $\countof \mbbt < \infty$, we characterise \ref{KQ}–\ref{AQ},
  \twodiv\ and novel extensions]\label{step-twodiv}
  For any pair of vectors $\acute v, J: \mbbtpp \rightarrow \R$ and consider the
  inner product
\begin{linenomath*}
  \[ \langle \acute v, J \rangle \defeq \sum_{t \,\in \,\mbbtpp} v(t) J(t).\]
\end{linenomath*}
Since $\mbbtpp $ is finite, the linear operator
$J \mapsto \langle \acute v , J \rangle$ is real-valued.  Let
$\acute H \defeq \{J : \langle \acute v , J \rangle = 0 \}$, let
$\acute G \defeq \{J: \langle \acute v, J \rangle > 0\}$, and let
$\acute F \defeq \{J : \langle \acute v, J \rangle \geq 0\}$. Similarly, we take
$\acute H_{\mdoubleplus} = \acute H\cap \posreal^{\mbbtpp} $ to be the
(strictly) \emph{positive kernel} of $\langle \acute v , \cdot \rangle$ and, for
$0 \not \leq \acute v \not \leq 0$, $\acute G_{\mdoubleplus} $ and
$\acute F_{\mdoubleplus}$ are, respectively, the open and closed half-spaces of
$\posreal ^{\mbbtpp}$ associated with $\acute v$. (For such $\acute v$, we also
refer to $\acute H_{\mdoubleplus}$ as a hyperplane in $\posreal^{\mbbtpp}$.) For
any $Y\in 2^{X}$ and matrix $\acute v : Y^{2}\times \mbbtpp\rightarrow \R$, each
row vector $\acute v^{\xy}: \mbbtpp \rightarrow \R$ the associated spaces are
$\acute H^{\{x,y\}}_{\mdoubleplus}$, $\acute G^{\xy}_{\mdoubleplus}$ and
$\acute F^{\xy}_{\mdoubleplus}$ respectively.

    On occasion we refer to the non-negative counterpart of these sets
    $\acute{H}_{\mplus},\acute{G}_{\mplus}$ and $\acute{F}_{\mplus}$ in
    $\nnreal^{\mbbtpp}$. It is natural to ask why we do not work with the latter
    sets throughout. The answer is that $\posreal^{\mbbtpp}$ has the same
    topological structure as $\R^{\mbbtpp}$. By working with hyperplanes in
    $\posreal^{\mbbtpp}$, the key result,  Zaslavski's theorem, from the literature on
    arrangements (of hyperplanes) applies without
    modification. We briefly introduce this literature following the next lemma
    which corresponds to lemma 1 of \gsii\ and gives meaning to the statement
    ``the arrangement generated by an extension''.

 \begin{theoremEnd}{lemma}[two-diverse pairwise representation]\label{lem-insep}
 
   Let $\aext$ be a $Y$-extension of $\preceq_{\mbbj}$.  $ \aext $ satisfies
   \ref{KQ}--\ref{AQ} and \twodiv\ holds on $Y$, if, and only if, there exists a
   matrix $\acute{v}^{(\cdot,\cdot)} : Y^{2} \times \mbbtpp \rightarrow \R $
   such that, for every $x,y\in Y$, row
   $\acute{v}^{(x,y)}: \mbbtpp \rightarrow \R$ and its associated spaces
   $\acute H^{\{x,y\}}_{\mdoubleplus}$ and $\acute G^{\xy}_{\mdoubleplus}$
   satisfy
   \begin{enumerate}[label=\textup{(\roman*)}]
   \item \label{K-insep}
     $ \acute H^{\{x,y\}}_{\mdoubleplus} \cap \mbb Q^{\mbbtpp} = \{J : x \anext_{J} y\}$ and
     $ \acute G ^{\xy}_{\mdoubleplus} \cap \mbb Q^{\mbbtpp}= \{ J: x \asext_{J} y\}$,
   \item\label{D-insep}

  $ \acute{G}^{\xy}_{\mdoubleplus}$ and $ \acute{G}^{\yx}_{\mdoubleplus} $ are
  both nonempty if $x \neq y$ and both empty otherwise,

\item\label{skew-insep}

  $\acute H^{\{y,x\}}_{\mdoubleplus} =\acute H^{\{x,y\}}_{\mdoubleplus}$ (and in particular
  $ \acute{v} ^ \yx = - \acute{v} ^ \xy $),



\item\label{unique-insep}

  $\acute H^{\{x,y\}}_{\mdoubleplus}$ is the unique hyperplane in $\posreal^{\mbbtpp}$ that
  separates $\{J : x \asextb_{J} y\}$ and $ \{ J: y \asextb_{J} x\}$
  ($\acute{v}^{\xy}$ is unique upto multiplication by a positive scalar).

\end{enumerate}
Moreover, $ \aext $ is novel if, and only if, for every $ t \neq \novel $,
$\acute{v}^{\dd}(t) \neq \acute{v}^{\dd}(\novel)$.

\end{theoremEnd}
\begin{proofEnd}
  
  In addition to \ref{KQ}–\ref{AQ}, the proof of lemma 1 of \gsii\ only appeals
  to \twodiv. That lemma, like the present one, does not require \ref{TQ}, or
  any diversity condition stronger than \twodiv\ since it is only result about
  distinct pairs of elements $x$ and $y$ in isolation.

  \Wlog\ we suppress reference to the acute accent.  Modulo notation, lemma 1 of
  \gsii\ and its proof show that \ref{KQ}–\ref{AQ}, in addition to \twodiv\ on
  $Y$, imply the existence $v^{\dd}$ with rows satisfying properties
  \ref{K-insep}–\ref{unique-insep} of the present lemma.

  We now prove the converse: that properties \ref{K-insep}–\ref{unique-insep}
  imply the axioms hold.  For \ref{KQ}, fix arbitrary $J \in \mbbj $. For every
  $x,y \in Y$, the fact that $\langle v^{\xy},\cdot \rangle$ is real-valued and,
  via property \ref{skew-insep}, equal to $-\langle v^{\yx},\cdot \rangle$,
  ensures that $J$ belongs to one of the sets $ H_{\mdoubleplus}^{\{x,y\}}$,
  $G_{\mdoubleplus}^{\xy}$ and $G_{\mdoubleplus}^{\yx}$. Then property
  \ref{K-insep} and the fact that $J $ belongs to $\mbb Q^{\mbbtpp}$ completes
  the proof. \ref{CQ} and \ref{AQ} hold by virtue of the fact that
  $\langle v^{\xy},\cdot \rangle$ is linear on $\posreal^{\mbbtpp}$. Finally, we
  prove that \twodiv\ holds. Take any distinct $x,y \in Y$. Then, via property
  \ref{D-insep}, both $G_{\mdoubleplus}^{\xy} $ and $G_{\mdoubleplus}^{\yx}$ are
  nonempty. By continuity of $\langle v^{\xy}, \cdot \rangle$,
  $G_{\mdoubleplus}^{\xy}$ and $G_{\mdoubleplus}^{\yx}$ are also open in
  $\posreal^{\mbbtpp}$. As such they each contain a rational vector, so that, by
  property \ref{K-insep}, \twodiv\ holds on $\{x,y\}$.
 
  We now prove the characterisation of novel extensions. Let $\ext$ be a novel
  $Y$-extension with matrix representation $v^{\dd}$ satisfying parts
  \ref{K-insep}–\ref{unique-insep} of the lemma.  Fix arbitrary
  $ t \neq \novel $.  Then \cref{def-extensionQ} implies the existence of
  $ J \in \mbbj $ and $ L = J \times 0 \in \mbbjp $ such that
  $ \extb _ { L + \delta_{t} } \neq \extb _ { L + \delta_{\novel} }$.  \Wlog,
  consider the case where, for some $ x , y \in Y $, it holds that both
  $ y \ext _ { L + \delta _ t } x $ and $ x \sext _ { L + \delta_ \novel } y
  $. Equivalently,
\begin{linenomath*}
  \[ \langle v^{\xy}, L + \delta_ t \rangle \leq 0 <\langle v^{\xy}, L + \delta_
  \novel\rangle \]
\end{linenomath*}
 which, via linearity of $\langle v^{\xy}, \cdot \rangle$, we
 may rearrange to obtain
 \begin{linenomath*}
  \begin{equation}\label{eq-nov}
    v^{\xy} ( t ) \leq -\langle v^{\xy}, L \rangle < v ^ \xy ( \novel ) .
\end{equation}
\end{linenomath*}
Thus, $v^{\xy} (t) \neq v^{\xy}(\novel)$, as required for the lemma.

For the converse argument, fix arbitrary $t \in \mbbt$. Then
$\acute v^{\dd} (t) \neq \acute v^{\dd} (\novel)$ implies the existence of
distinct $x,y \in Y$ such that
$ \acute v^{\xy}(t) \neq \acute v^{\xy}(\novel) $. We show that there exists
$J \in \mbbj $ and $L= J\times 0 \in \mbbjp$ satisfying \cref{eq-nov}. For then,
by retracing (in reverse order) the arguments that lead to \cref{eq-nov}, we
arrive at the conclusion that
$ \extb _ { L + \delta_{t} } \neq \extb _ { L + \delta_{\novel}} $, as required
for $\ext$ to be novel.

Take $ \mu = v ^ \xy ( t ) $ and $ \xi = v ^ \xy ( \novel ) $ and consider the
case where $ \mu < \xi < 0 $.  Since $\mu<0$, \twodiv\ implies that there exists
$s\in\mbbt$ such that $v^{\xy}(s)$ is positive.  Then, for some
$ \lambda \in \posrat $, $ - \lambda v ^ \xy ( s ) \in ( \mu , \xi ) $.  Let
$ L = \lambda \delta ^{ \novel } _ s $ and observe that
\begin{linenomath*}
  \[ \mu < - \langle v^{\xy}, L \rangle < \xi ,\]
\end{linenomath*}
  as required. \emph{Mutatis mutandis}, the case where both $ \mu $ and $ \xi $
  are positive is the same.  If $ \mu \leq 0 \leq \xi $, then take $ L = 0 $, so
  that $\mu \neq \xi$ yields
  $ \extb _ { L + \delta ^{ \novel } _ t } \neq \extb _ { L + \delta ^{ \novel }
    _ \novel } $.
\end{proofEnd}
We refer to a matrix $v^{\dd}$ that satisfies condition \ref{D-insep} of
\cref{lem-insep} (for every distinct $x,y\in Y$) as a \emph{$2$-diverse
  matrix}. If $v^{\dd}$ satisfies all the conditions of \cref{lem-insep} with
respect to a given extension $\ext$ and, in addition, $x=y$ implies
$v^{\xy}= 0$, then $v^{\dd}$ is a $2$-diverse pairwise representation of
$\ext$.\footnote{Note that for $\countof X=2$, \cref{lem-insep} constitutes a
  proof of \cref{thm-mainQ}. For we may take
  $\mathbf{v}: X \times \mbbt \rightarrow \R$ such that $\mathbf{v}(x,\cdot)=0$
  and set $\mathbf{v}(y,\cdot)=v^{(x,y)}$, so that
  $v^{(x,y)}= -\mathbf{v}(x,\cdot) + \mathbf{v}(y,\cdot)$.}

\paragraph{An \emph{arrangement}\hskip-4pt} is a collection of hyperplanes in
$\posreal^{\mbbtpp}$ or $\R^{\mbbtpp}$. (A \emph{hyperplane in
  $\posreal^{\mbbtpp}$} is the positive kernel of some nonzero vector.) The
chief result, from the mathematics of arrangements, to which we extensively
appeal is Zaslavsky's theorem two form of which we state below.  For any given
extension $\ext$, Zaslavsky's theorem allows us to use information about the
intersections of hyperplanes in the arrangement to identify
$ \countof \total (\ext) $.  It does so by counting the collection
$\mc G_{\mdoubleplus}$ of open and connected subsets of
$ \R^{\mbbtpp} \bs \bigcup \{ H_{\mdoubleplus} : H_{\mdoubleplus} \in \mc A \} $
are the \emph{chambers} or \emph{regions} of the arrangement.  In the present
setting, each chamber corresponds to a \emph{CAR ranking} of the elements of
$Y$: a complete, antisymmetric and reflexive (but possibly intransitive) ranking
$R$. Every CAR ranking can be succinctly represented as a CAR list. For
instance, take $Y = \{ x, y , z\}$ and $ x \mathbin{R} y \mathbin{R} z $, then
the corresponding list
\begin{linenomath*}
\begin{equation*}
l = \left \{
\begin{array}{ll}
  (x,y,z)& \text{if $R$ is transitive}\\
  (x,y,z,x)& \text{if $R$ is intransitive.}\footnotemark
\end{array}\right.
\end{equation*}
\end{linenomath*}
  The notation extends without exception to sets of cardinality $4$. \Eg,
  $(x,y,z,x,w)$ represents the CAR ranking that is intransitive over
  $ \{x,y,z\}$ and such that $w$ dominates every other member.

\paragraph{The intersection semilattice\hskip-4pt}
of any arrangement $\mc A$ is the partially ordered (\emph{by reverse
  inclusion}) set $\mc L$ of intersections of members of $\mc A$. The unique
minimal element is obtained by taking the intersection $A^{\emptyset}$ over the
empty subarrangement $\mc A^{\emptyset}$ of $\mc A$ to obtain the ambient space
itself. That is $A^{\emptyset}= \R^{\mbbtpp}$ or $\posreal^{\mbbtpp} $,
depending on whether we are considering the lattice $\mc L$ or the lattice
$\mc L_{\mdoubleplus}$ respectively. In \gsii, as a consequence of \fourdiv,
$\mc H_{\mdoubleplus}$ is always central. In our setting, it is only $\mc H$
that is guaranteed to be central. In general an arrangement is central, if, and
only if, its intersection semilattice has a unique maximal element
\citep[proposition 2.3]{Stanley-Arrangements}. Thus, if $\mc H_{\mdoubleplus}$
is \emph{centerless}, then $\mc L_{\mdoubleplus}$ is a meet semilattice with
multiple maxima: as in \cref{eg-zaslavski}.  Extending our notation: if
$Y = \{x,y,z,w\}$, then the unique intersection $A^{Y}$ is the (nonempty)
\emph{center} of $\mc A^{Y} = \mc H$. By $A^{\{x,y,z\}}$, we mean the
intersection over $\mc A^{\{x,y,z\}} \defeq \{H^{\{i,j\}}: \text{$i\neq j$ in
  $\{x,y,z\}$}\}$. Finally, by $A^{\{x,y\}\{z,w\}}$, we mean the intersection
over $\mc A^{\{x,y\}\{z,w\}} \defeq \{H^{\{x,y\}},H^{\{z,w\}}\}$.

\paragraph{Zaslavski's theorem} provides two distinct methods for counting the
number of regions in an arrangement. The first states that
\emph{$\countof \mc G$ is equal to the sum of the absolute values of the
  M\"{o}bius function $\bmu: \mc L \rightarrow \mbb Z$} which is defined
recursively via
\begin{linenomath*}
  \begin{equation}\label{eq-mobius}
\bmu(A) = \left \{
\begin{array}{ll}
  1  & \text{if $A = A^{\emptyset} $}\\
  -\sum\{\bmu(B):  A \subsetneq B\} & \text{otherwise.}\footnotemark
\end{array}\right.
\end{equation}
\end{linenomath*}
The above definition of Zaslavski's theorem is explicitly provided by
\citet{Sagan-Zaslavski}. Specialised to the present setting, the more common
\citep[see ][]{OT-Arrangements,Dimca-Arrangements,Stanley-Arrangements} ``rank''
version of Zaslavski's theorem is
\begin{linenomath*}
  \[\countof \mc G  =\sum_{\substack{\mc A \subseteq \mc H\\
        \mc A \textup{\ central}}} (-1)^{\lvert \mc A \rvert - \rank(\mc A)},\]
  \end{linenomath*}
  where \emph{central} means that $\bigcap \{H: H \in \mc A\}$ is nonempty, and
  $\rank (\mc A)$ is the dimension of the space spanned by the normals to the
  hyperplanes in $\mc A$.\footnote{Equivalently, $\rank(\mc A) $ is the
    dimension of the orthogonal complement in $\R^{\mbbtpp}$ (or
    $\R^{\mbbtpp}_{\mdoubleplus}$) of the intersection over $\mc A$. Since the
    intersection over the empty arrangement $\mc A^{\emptyset}$ is the ambient
    space $\R^{\mbbtpp}$ (or $\posreal^{\mbbtpp}$), the only vector in
    $\R^{\mbbtpp}$ that is orthogonal to the ambient space is $0$,
    $\rank (\mc A ^{\emptyset}) = 0$.}

\begin{example}[a comparison of $\mc L$ and $\mc L_{\mdoubleplus}$]\label{eg-zaslavski}
  Let $ X = \{ x , y , z \}$, $\mbbt = \{s,t\}$, and
  $ u ^{ (x,y) } = 1 \times - 1$ and $ u ^{(y,z) } = 2 \times - 1$ denote
  vectors in $ \R^{\mbbt} $. We now appeal to the Jacobi identity and take
  $u ^{(x,z) } = u ^{ (x,y) } + u ^{ (y,z) } = 3 \times - 2 $ and extend to the
  remaining pairs in $X^{2}$ using \cref{lem-insep}.  Since these vectors are
  pairwise noncollinear and $\countof \mbbt = 2$, the associated arrangement
  $\mc H_{\mdoubleplus} = \left\{ H_{\mdoubleplus}^{\{x,y\}},
    H_{\mdoubleplus}^{\{y,z\}}, H_{\mdoubleplus}^{\{x,z\}}\right\}$ consists of
  three pairwise disjoint lines that partition $\posreal^{\mbbt}$ and
  $\mc G_{\mdoubleplus}$ has cardinality $4$.  We now confirm this using
  Zaslavski's theorem.
  \begin{figure}
    \begin{center}
  \begin{tikzpicture}
    \node (max) at (0,1.5) {\color{lightgray}$\{x,y,z\}$};
    \node (d) at (-2,0.25) {$\{x,y\}$};
    \node (e) at (0,0.25) {$\{y,z\}$};
    \node (f) at (2,0.25) {$\{x,z\}$};
    \node (min) at (0,-1) {$\emptyset$};
    \draw (min) -- (d);
    \draw (min) -- (e);
    \draw (min) -- (f);
    \draw[color = lightgray] (max)  -- (d);
    \draw[color = lightgray] (max)  -- (e);
    \draw[color = lightgray] (max)  -- (f);    
\end{tikzpicture}
\caption{\label{fig-hasse-Y3-r2} The intersection semilattice
  $\mc L_{\mdoubleplus} = \mc L \bs A^{\{x,y,z\}}$.}
\end{center}
\end{figure}

In the present setting $A^{\emptyset}_{\mdoubleplus} = \posreal^{\mbbt}$ and,
via \cref{eq-mobius}, $\bmu(A^{\emptyset}_{\mdoubleplus}) = 1$. Then, since
$\posreal^{\mbbt}$ is the unique element in $\mc L_{\mdoubleplus}$ that
(strictly) contains each hyperplane in $\mc H_{\mdoubleplus}$, \cref{eq-mobius}
yields $\bmu(A) = -\bmu(A^{\emptyset}_{\mdoubleplus}) $ for each
$A \in \mc H_{\mdoubleplus}$.  Now since the hyperplanes in
$\mc H_{\mdoubleplus}$ are disjoint, there are no further elements in
$\mc L_{\mdoubleplus}$ and we observe that
  \begin{linenomath*}
    \[\lvert \mc G_{\mdoubleplus} \rvert= \sum_{A \in \mc L_{\mdoubleplus}} \lvert
      \bmu(A) \rvert = 4.\]
  \end{linenomath*}

  In contrast, although the structure of $\mc L$ is otherwise isomorphic to
  $\mc L_{\mdoubleplus}$, since $\{0\} \subset \R^{\mbbt}$ is a subset of every
  hyperplane in $\mc H$, $\{0\}$ is the center $A^{\{x,y,z\}} $ of $\mc H$ and
  the maximal element of $\mc L$. Via \cref{eq-mobius} and the calculations of
  the previous paragraph, we obtain
  $\bmu(A^{\{x,y,z\}}) = - \left(\bmu(A^{\emptyset}) -3 \bmu(A^{\emptyset})
  \right) = 2$.  Thus,
\begin{linenomath*}
  \[  \countof\mc G = \sum_{A \in \mc L} \lvert \bmu(A) \rvert = 6 = 3!.\]
\end{linenomath*}
\end{example}
\begin{remark}[The relationship between $\mc L $ and $\mc L_{\mdoubleplus}$]\label{rem-lattice}
  Let $\aext $ be a $Y$-extension with $2$-diverse representation
  $\acute{u}^{\dd}$. Since, for every distinct $x$ and $y $ in $Y$,
  $\acute{H}^{\{x,y\}}$ contains the origin, $\acute{\mc H}$ is centered. As we
  see in \cref{eg-zaslavski}, this is not the case for
  $\acute{\mc H}_{\mdoubleplus}$ where $\acute{\mc H}_{\mdoubleplus}$ is
  centerless and each of its members is maximal in
  $\acute{\mc L}_{\mdoubleplus} $.

  In \gsii, \fourdiv\ guarantees that, for every $Y\subseteq X$ of cardinality
  $2,3$ or $4$, the improper $Y$-extension generates a centered arrangement in
  $\posreal^{\mbbt}$. The fact that $\posreal^{\mbbtpp}$ is open in
  $\R^{\mbbtpp}$ ensures that the dimension of any $L \in \acute{\mc L}$ is
  equal to its counterpart $L _{\mdoubleplus} \in \acute{\mc L}_{\mdoubleplus}$
  provided the latter exists.  Thus, $ \acute{\mc L}_{\mdoubleplus}$ and
  $ \acute{\mc L}$ are isomorphic if, and only if,
  $\acute{\mc H}_{\mdoubleplus}$ is centered. For the same reason,
  $\acute{\mc G}_{\mdoubleplus} $ and $\acute{\mc G} $ are isomorphic if,
  and only if, $\acute{\mc H}_{\mdoubleplus}$ is centered.
\end{remark}
We now abstract a useful property from \cref{eg-zaslavski}.
\begin{theoremEnd}{proposition}[polar opposite
  rankings]\label{prop-pairwise-extremal}
  If $\preceqb_{\mbbj}$ satisfies \ref{KQ}–\ref{AQ} and \twodiv, then, for every
  $Y\subseteq X$ of cardinality $3$ or $4$, the improper $Y$-extension $\extb $
  is such that, for some $J,L \in\mbbj$, $\extb_{J} = \extb_{L}^{-1} $ belongs
  to $ \total(\ext)$.
\end{theoremEnd}
\begin{proofEnd}
  Fix $\countof Y= 3$ or $4$, via \cref{lem-insep}, let $ v^{\dd}$ denote the
  $2$-diverse matrix representation of the improper $Y$-extension $\ext$. Let
  $\mc H_{\mdoubleplus}$ denote the associated arrangement of hyperplanes.  For
  every distinct $x,y\in X$, \cref{lem-insep} implies that
  $H^{\{x,y\}}_{\mdoubleplus}$ intersects $ \posreal^{\mbbt}$. Then, similar to
  \cref{eg-zaslavski}, the $ 1 \leq n \leq \binom{\countof Y}{2} $ distinct
  hyperplanes of $\mc H_{\mdoubleplus} $ cut $\posreal^{\mbbt}$ into at least
  $n+1$ regions.  At least one pair $ G$ and $G^{*}$ in $\mc G_{\mdoubleplus}$
  are therefore separated by all $n$ distinct members of $\mc
  H_{\mdoubleplus}$. Take $J \in G$, so that, for every distinct $x,y \in Y$,
  $\langle u^{\xy}, J \rangle \neq 0$.  Thus $\ext_{J}$ is antisymmetric,
  complete and, via \ref{TQ}, total. Next, take $L \in G^{*}$, so that since $J$
  and $L$ are separated by every hyperplane in $\mc H_{\mdoubleplus}$,
  $\extb_{J}= \extb_{L}^{-1}$.
\end{proofEnd}
\begin{example}[insufficiency of \twodiv]\label{eg-lexicographic}
  Let $X = [0,1]^{2}$ and let $\leq^{\textup{lex}}$ denote the lexicographic
  ordering on $X$.  Let $\mbbt = \{s,t\}$, and, for each $J \in \mbbj $, let
\begin{linenomath*}
\begin{equation*}
\preceqb_{J} = \left\{ 
  \begin{array}{ll}
    X^{2} & \text{if  $J(s) = J( t )$$;$}\\
    \leq^{\textup{lex}} & \text{if  $J(s) < J( t )$$;$}\\
    (\leq^{\textup{lex}})^{-1} & \text{otherwise.}
    \end{array}\right.
\end{equation*}
\end{linenomath*}
Recall that if $\preceqb_{J} = X^{2}$, then $\preceqb_{J}$ is symmetric and
hence equal to $\simeq_{J}$.  Thus, for every distinct $x,y\in X$,
$ H^{\{x,y\}}_{\mdoubleplus}= \{J \in \R^{\mbbt}_{\mdoubleplus} : J(s) =
J(t)\}$. Via \cref{lem-insep}, $\preceqb_{\mbbj}$ has a two-diverse matrix
representation $v^{\dd}$. But via \cref{lem-c2dQ}, below, \condtwodiv\ fails to
hold. The fact that $\preceqb_{\mbbj}$ fails to satisfy part \ref{rep-mainQ} of
\cref{thm-mainQ} follows from the fact that $\preceqb_{J}$ is lexicographic for
every $J$ outside $H$.
\end{example}
\end{step}
\begin{step}[for $\mbbt<\infty$, characterisations of \ref{c2dQ}]\label{step-condtwodiv}
  A matrix $v^{\dd}$ that satisfies the conditions of the next lemma is a
  \emph{conditionally-$2$-diverse (pairwise) representation}.

  \begin{theoremEnd}{lemma}[conditionally-$2$-diverse representation]\label{lem-c2dQ}
   Let $\ext$ be a $Y$-extension of $\preceq_{\mbbj}$ with $2$-diverse matrix
    representation $v^{\dd}$.  Then \ref{c2dQ} holds on $Y$ if, and only if, for
    every three distinct elements $x,y,z\in Y$, $v^{\xz} $ and $v^{\yz}$ are
    noncollinear.
  \end{theoremEnd}
  \begin{proofEnd}
\label{proof-lem-c2dQ}
    Let \ref{c2dQ} hold on $Y$. Since $v^{\dd}$ is a $2$-diverse pairwise
    representation, $v^{\xz}, v^{\yz}\not \leq 0$.  By \ref{c2dQ}, one of
    $G_{\mplus}^{\xz}$ and $G_{\mplus}^{\zx}$ contains both $J,L$ such that
\begin{linenomath*}
  \[\langle v^{(y,z)} ,  L \rangle <0< \langle v^{(y,z)} , J\rangle.\]
\end{linenomath*}
\Wlog, suppose $J, L $ belongs to $ G_{\mplus}^{\xz}$. Then
$\langle v^{\xz},\cdot \rangle$ is positive on $\{L,J\}$, so that, for every
$\lambda \in \R $ $v^{\xz} \neq \lambda v^{\yz}$, as required.

    Conversely, let $x,y,z \in Y$ be such that $v^{\xz}$ and $v^{\yz}$ are
    noncollinear.  Then $H_{\mdoubleplus}^{\{x,z\}}\neq H_{\mdoubleplus}^{\{y,z\}}$, and
    there exists $L \in H_{\mdoubleplus}^{\{x,z\}}\bs H_{\mdoubleplus}^{\{y,z\}}$.
    \Wlog, therefore, suppose
    $L \in H^{\{x,z\}}_{\mdoubleplus} \cap G^{\yz}_{\mdoubleplus}$.  Since
    $v^{\dd}$ is $2$-diverse, there exists $s,t\in \mbbt$ such that
    $v^{\xz}(s)<0<v^{\xz}(t)$. Noting that $L \in \posreal^{\mbbtpp}$, so that
    $ v^{\xz} \neq \delta_{s},\delta_{t}$, let $\psi_{s}$ and $\psi_{t}$ be the
    convex paths from $L $ to $\delta_{s}$ and $\delta_{t}$ respectively. For
    sufficiently small  $\lambda > 0$,
    $ \langle v^{\yz},\psi_{s'} (\lambda) \rangle$ remains positive for
    $s'= s, t$ and, moreover, since $L \in H^{\{x,z\}}_{\mdoubleplus}$,
\begin{linenomath*}
    \[\langle v^{\xz}, \psi_{s}(\lambda)\rangle<0< \langle v^{\xz},
      \psi_{t}(\lambda) \rangle. \]
\end{linenomath*}
Finally, since $L$ has finite support, a finite sequence of perturbations of the
elements of $\psi_{s}(\lambda)$ and $\psi_{t}(\lambda)$ yields (rational-valued)
members of $\mbbjpp$ with the same properties, as required for \ref{c2dQ}.
\end{proofEnd}
The following is a translation of \cref{obs-c2d}.
\begin{theoremEnd}{proposition}[on \ref{c2dQ} and \ref{p3dQ}]\label{prop-c2dQ}
  For $\preceqb_{\mbbj}$ satisfying \ref{TQ}--\ref{AQ}, \condtwodiv\ and
  \parthreediv\ are equivalent.
\end{theoremEnd}
\begin{proofEnd}
  \label{proof-prop-c2dQ}
  When $X=2$, \ref{c2dQ} and \ref{p3dQ} are identical to \twodiv.  Let
  $ Y = \{ x, y , z \} \subseteq X$ and let $\ext$ denote the improper
  $Y$-extension of $\preceq_{\mbbj}$. We begin by assuming \ref{c2dQ} and
  showing that $\countof Y + 1 = 4 \leq \countof \total ( \ext )$. Via
  \cref{lem-c2dQ}, there are three distinct hyperplanes
  $H_{\mdoubleplus}^{\{x,y\}}, H_{\mdoubleplus}^{\{y,z\}}$ and
  $H_{\mdoubleplus}^{\{x,z\}}$ in the associated arrangement
  $\mc H_{\mdoubleplus}$. Then, as in \cref{eg-zaslavski},
  $\emptyset = \posreal^{\mbbt}$ is the unique element of $\mc L_{\mdoubleplus}$ that
  lies below each member of $\mc H_{\mdoubleplus}$. Thus, via \cref{eq-mobius},
  $\bmu(A^{\emptyset}) = 1 $ and $\bmu(A) = - \bmu(A^{\emptyset})$ for all three hyperplanes
  $A \in \mc H_{\mdoubleplus}$. Thus, Zaslavski's theorem implies that
  $\countof \mc G_{\mdoubleplus}$ is bounded below by $4$.  Thus
  $\total(\ext) \geq 4$, and since, for every $Y$-extension $\aext$,
  $\countof \total (\aext ) \geq \countof \total (\ext)$, \ref{p3dQ} holds.

  Conversely, suppose \ref{p3dQ} holds and, once again let $\ext$ denote the
  improper $Y$-extension of $\preceq_{\mbbj}$, so that
  $\total(\ext) \geq \countof Y = 3$. Now \ref{p3dQ} implies \twodiv, so that,
  via \cref{lem-insep}, there exists a $2$-diverse matrix representation with
  associated arrangement $\mc H_{\mdoubleplus}$. It is not the case that
  $\lvert \mc H_{\mdoubleplus} \rvert = 1$, for this would imply that
  $\countof \total(\ext) = 2$. \Wlog, suppose
  $H_{\mdoubleplus}^{\{x,y\}} \neq H_{\mdoubleplus}^{\{y,z\}}$. Observe that
  \ref{TQ} then implies
  $H_{\mdoubleplus}^{\{x,z\}} \neq H_{\mdoubleplus}^{\{x,y\}}$ and
  $H_{\mdoubleplus}^{\{x,z\}} \neq H_{\mdoubleplus}^{\{y,z\}}$. This implies
  that $v^{\xy}, v^{\yz} $ and $v^{\xz}$ are pairwise noncollinear. Finally, an
  application of \cref{lem-c2dQ} then yields \ref{c2dQ}.
\end{proofEnd}
  \end{step}

  \begin{step}[for $\mbbt<\infty$, a characterisation of \fourpru]\label{step-prudence}
   The following Jacobi identity plays a central role in the proof of \gsii.
\begin{definition*}

  For $ Y \in 2 ^ { X } $, the matrix
  $ v^{ \dd } : Y^{2}\times \mbbtpp \rightarrow \R $ satisfies the \emph{Jacobi
    identity} whenever, for every $ x , y , z \in Y $,
  $ v ^{ \xz } = v ^{ \xy } + v ^{ \yz }$.
\end{definition*}
For any given $Y$-extension $\ext$, the \emph{Jacobi identity holds for
  $ \ext $} whenever it holds for some pairwise representation $v^{\dd}$ of
$ \ext $. Moreover, in this case, $v^{\dd} $ is a \emph{Jacobi
  representation}. Finally, if the $Y$-extension $\ext$ is improper and the
Jacobi identity holds for $ \ext $, we simply say that the \emph{Jacobi identity
  holds on $ Y $}.  Consider
\begin{k-jac*}

  For every $ Y \subseteq X $ with $ 3 \leq \countof Y \leq k $, the Jacobi
  identity holds on $ Y $.
   
\end{k-jac*}
We will work with \threejac\ and \fourjac\ in particular. The following lemma is
the special case of \cref{thm-foureq} of \cref{sec-proof-foureq} where $\mbbt$
is finite. When $\mbbt$ is finite, for every $Y$, the set of testworthy
$Y$-extensions that are novel is nonempty. In this case, \fourpru\ implies
\ref{TQ}--\ref{AQ} via nonrevision of rankings (part \ref{item-preservingQ}  of
\cref{def-extensionQ}).
 \begin{lemma}\label{lem-foureq}
   For $\preceq_{\mbbj}$ satisfying \ref{c2dQ}, \fourpru\ holds if, and only if,
   \fourjac\ holds.
 \end{lemma}
 \end{step}
\setcounter{step}{3}
\begin{step}[for arbitrary $X$ and $\mbbt$, the induction
  argument]\label{step-induction}
  The present step is closely related to lemma 3 and claim 9 of \gsii. There the
  authors establish that, when \fourdiv\ holds, $3$-Jac is a necessary and
  sufficient condition for the (global) Jacobi identity to hold on $ X $.
  \gsii\ relies on the fact that \fourdiv\ implies linear independence of
  $\{v^{\xy},v^{\yz}, v^{\zw}\}$ for every four distinct elements
  $x,y,z,w \in X$. In the present setting, where \condtwodiv\ only implies
  linear independence of pairs $\{v^{\xy},v^{\yz}\}$, the Jacobi identity
  requires \fourjac.

\begin{theoremEnd}{lemma}[Jacobi representation]\label{lem-induction} Let
  $\preceq_{\mbbj}$ have a conditionally-$2$-diverse representation $u^{\dd}$.
  Then \fourjac\ holds if, and only if, $\preceq_{\mbbj}$ has a Jacobi
  representation $v^{\dd}$.  Moreover, for every Jacobi representation
  $ \mathbf v ^{ \dd }$ of $\preceq_{\mbbj}$ there exists $\lambda > 0$
  satisfying $ \mathbf v ^{ (\cdot, \cdot) } = \lambda v ^{(\cdot, \cdot)} $.

\end{theoremEnd}

\begin{proofEnd}\label{proof-lem-induction}
  Note that, when $1 \leq \lvert X \rvert \leq 2 $, \fourjac\ holds vacuously and
  $\preceq_{\mbbj}$ has a Jacobi representation via \cref{lem-insep}.  For
  general $X$, the fact that \fourjac\ is necessary for $\preceq_{\mbbj} $ to
  have a Jacobi representation follows simply because if the Jacobi identity
  holds on $X$, then it holds on every $Y\subseteq X$. For the sufficiency of
  \fourjac, we proceed by induction.  As in lemma 3 and claim 9 of \gsii, we
  assume that $X$ is well-ordered.  


  In the case that $ \lvert X \rvert \leq 4 $, we only need to show that
  $ v ^{ \dd } $ is unique. \Wlog, we take the initial step in our induction
  argument to satisfy $ \lvert X \rvert = 4 $. 
  Let $ \mathbf v ^{\dd }$ denote any other Jacobi representation of
  $\preceq_{\mbbj}$. By \cref{lem-insep}, for every distinct
  $ x , y \in Y ^{ 2 } $, there exists $ \lambda^{\{x,y\}}> 0 $ such that
  $ \mathbf v ^{ \xy } = \lambda ^{\{x,y\}} v ^{ \xy }$. We need to show that
  $ \lambda ^{\{x,y\}} = \lambda $ for every distinct $ x , y \in Y $.  Let
  $ Y = \{ x , y , z , w \}$.  By \cref{lem-c2dQ}, the set
  $\{ v ^{ \xy } , v ^{ \xz } , v ^{ \xw } \}$ is pairwise noncollinear. Then,
  since the Jacobi identity holds for both $ v ^{ \dd } $ and
  $ \mathbf v ^{\dd }$, we derive the equation
\begin{linenomath*}
  \begin{equation}\label{eq-jac-unique}
    (1 - \lambda ^{\{x,y\} })v ^{ \xy } + (1 - \lambda ^{\{y,z\}})v ^{\yz } =
    (1 - \lambda ^{\{x,z\}}) v ^{\xz}
  \end{equation}
\end{linenomath*}
  Suppose that $ 1 - \lambda ^{\{y,z\} } = 0 $. Then, either the other coefficients
  in \cref{eq-jac-unique} are both equal to zero (and our proof is complete), or
  we obtain a contradiction of \cref{lem-c2dQ}.  Thus, $ 1 - \lambda ^{\{y,z\} } $
  is nonzero and we may divide through by this term and solve for
  $ v ^{ \yz } $. First note that, since $ v ^{ \dd }$ is a Jacobi representation,
  $ v ^{ \yx }+ v ^{ \xy } = v ^{ (y,y) } = 0$. Then, since  
  $ v ^{ \yx} = - v ^{ \xy }$, 
\begin{linenomath*}
  \[
    v ^{ \yz } =
    \frac{ 1- \lambda ^{ x y }}{ 1 - \lambda ^{ y z }} v^{ \yx}
    +
    \frac{ 1 - \lambda ^{ x y }}{ 1 - \lambda ^{ y z }} v^{ \xz }.
  \]
  \end{linenomath*}
  We then conclude that both of the coefficients in the latter equation are
  equal to one. (This follows from linear independence of $ v ^{ \yx }$ and
  $ v ^{ \xz }$ together with the Jacobi identity
  $ v ^{ \yz } = v ^{ \yx} + v ^{ \xz }$.)  Thus,
  $ \lambda ^{\{x,y\}} = \lambda ^{\{y,z\}} = \lambda ^{\{x,z\}} $, as
  required. Repeated application of the same argument to the remaining Jacobi
  identities yields the desired conclusion, $ \mathbf v^{\dd}= \lambda v^{\dd}$.


For the inductive step, take $Y$ to be an initial segment of $X$. By the
induction hypothesis, there exists a Jacobi representation
$ \v ^{ \dd }: Y^{2} \times \mbbt \rightarrow \R $ of the improper
$Y$-extension $\extb = \preceqb_{\mbbj} \cap Y^{2} $ that is suitably unique.
\begin{claim}\label{claim-induction-well-defined}
  For every $ w \in X \bs Y $ and $W \defeq Y \cup \{w\}$, there
exists a Jacobi representation
$\hat{v}^{\dd}: W^{2} \times \mbbt \rightarrow \R$ of the improper
$W$-extension $\hextb $.
\end{claim}
\begin{proof}[Proof of \cref{claim-induction-well-defined}]
  Via \cref{lem-c2dQ}, there exists a conditionally $2$-diverse pairwise
  representation $u^{\dd}$ of $\preceq_{\mbbj}$. Fix any four distinct elements
  $x, x ' , y , z$ in $Y$.  \Cref{lem-insep} implies the existence of
  $\phi, \phi' \in \posreal $ such that
  $\phi u^{\xz} = \v^{\xz}$ and $\phi' u^{\xpz} = \v^{\xpz}$.  Let
  $Z = \{ x, y , z , w\}$ and $Z ' = \{ x' , y , z , w\}$. Since \threejac\
  holds, there exist positive scalars
  $ \alpha , \beta , \hat \beta, \gamma, \sigma $ and $ \tau $ such that
\begin{linenomath*}
\begin{align}\label{eq-xy}
  \alpha u ^{ \xw } + \beta u ^{ \wy } &= \gamma u ^{ \xy },\\
   \label{eq-yz}
    \hat \beta u ^{ \yw } + \sigma u^{ \wz } &= \tau u ^{ \yz }, \text{  and}\\
   \label{eq-xz}
  \gamma u ^{ \xy } + \tau u ^{ \yz } &= \v ^{ \xz } .
\end{align}
\end{linenomath*}
\addtocounter{linenumber}{-1}
Moreover, \fourjac\ ensures that we may take
$\beta = \hat \beta $.  Since $u^{\dd}$ is conditionally $2$-diverse,
$ \{ u ^{ \xy } , u ^{ \yz } \}$ is linearly independent, and the linear system
\cref{eq-xz} in the unknowns $ \gamma $ and $ \tau $ has a unique solution.
This, together with the induction hypothesis (which yields
$ \v ^{ \xy } + \v ^{ \yz } = \v ^{ \xz }$) implies that
$ \gamma u ^{ \xy } = \v ^{ \xy }$ and $ \tau u ^{ \yz } = \v ^{ \yz } $.
Similarly, for $ Z ' $, \fourjac\ yields
$ \alpha ' , \beta ' , \sigma ' , \gamma ' , \tau ' > 0 $ such that
\begin{linenomath*}
 \begin{align}\label{eq-xpy}
   \alpha ' u ^{ \xpw } + \beta ' u ^{ \wy } &= \gamma ' u ^{ \xpy },\\
\label{eq-yz-xp}
   \beta '  u ^{ \yw } + \sigma '  u ^{ \wz }& = \tau '  u ^{ \yz }, \text{ and }\\
\label{eq-xpz-y}
    \gamma '  u ^{ \xpy } + \tau ' u ^{ \yz } &= \v ^{ \xpz } .
 \end{align}
\end{linenomath*}
\addtocounter{linenumber}{-1} As in the arguments involving $ \gamma $ and
$ \tau $, the induction hypothesis yields
$ \gamma ' u ^{ \xpy } = \v ^{ \xpy } $ and $ \tau ' u ^{ \yz } = \v ^{ \yz }$.
We conclude that $ \tau = \tau' $.  Substituting for $\tau'$ in \cref{eq-yz-xp}
and appealing to linear independence of $\{u^{\yw}, u^{\wz}\}$ then yields the
desired equalities $\beta = \beta' $ and $\sigma = \sigma ' $.
  
  As a consequence of the above argument, for every $y , z \in Y$, take
  $ \hat{v}^{\yw} $ and $ \hat{v}^{\wz}$ to be the unique vectors in
  $\R^{\mbbt}$ that solve the equation
  $\hat{v} ^{ \yw } + \hat{v} ^{ \wz } = \v ^{ \yz } $. For every $y , z \in Y$,
  let $\hat{v}^{\yz} = \v^{\yz}$ and $\hat{v}^{(w,w)} = 0$.  Then the matrix
  $\hat{v}^{\dd} $ with row vectors
  $\left\{ \hat{v}^{\xy}: x , y \in W \right\}$ is a Jacobi representation of
  $\hext$.
\end{proof}
Our proof of \cref{claim-induction-well-defined} shows that the extension to $W$
holds for any initial subsegment of $Y$ consisting of four elements. Our proof
thereby accounts for the case where  $ X $ is infinite and  $w$  is a limit ordinal.
\end{proofEnd}
\end{step}

\begin{step}[for $\countof \mbbt < \infty $, the concluding arguments in the proof of
  \cref{thm-mainQ}]\label{step-conc-mainQ}
  In \cref{step-condtwodiv} we showed that \ref{KQ}–\ref{c2dQ} hold if, and only
  if, $\preceq_{\mbbj}$ has a conditionally $2$-diverse pairwise
  representation.  In \cref{step-prudence}, we showed that, when
  $\preceq_{\mbbj}$ satisfies \ref{TQ}–\ref{c2dQ}, a necessary and sufficient
  condition for \fourjac\ is \ref{PQ}.  In \cref{step-induction},
  via (mathematical) induction, we showed that conditionally $2$-diverse Jacobi
  representations of $Y$-extensions of $\preceq_{\mbbj}$ such that
  $\countof Y = 4$, can be patched together to obtain a conditionally
  $2$-diverse Jacobi representation of $\preceq_{\mbbj}$ (on all of $X$,
  regardless of cardinality). The fact that \ref{TQ} is necessary for a
  Jacobi representation follows from \gsii.  As a
  consequence, \ref{TQ}–\ref{c2dQ} and \ref{PQ} are necessary and sufficient for
  a conditionally $2$-diverse Jacobi representation of $\preceq_{\mbbj}$. The
  following argument then completes the proof of \cref{thm-mainQ}.

  Let $v^{\dd}$ be a (conditionally $2$-diverse) Jacobi representation of
  $\preceq_{\mbbj}$ and define $\mathbf{v}: X\times \mbbt \rightarrow \R$ as
  follows. Fix arbitrary $w\in X$, and let $\mathbf{v}(w,\cdot) = 0 $. Then, for
  every $x\in X$, let $\mathbf{v}(x,\cdot) = v^{\wx}$. Recalling that
  $v^{\wx}= -v^{\xw}$, note that, since the rows of $v^{\dd}$ satisfy the Jacobi
  identity, for every $x,y\in X$, we have
\begin{linenomath*}
    \[v^{\xy}= v^{\xw}+v^{\wy} = -\mathbf{v}(x,\cdot) + \mathbf{v}(y,\cdot).\]
\end{linenomath*}
To see that
  \ref{rep-mainQ} holds note that, for every $J\in \mbbj$, we have
  $x \preceq_{J} y$, if, and only if, $0 \leq \langle v^{\xy}, J \rangle $, if,
  and only if, $\langle v(x,\cdot), J \rangle \leq \langle \mathbf{v}(y,\cdot), J \rangle$.
  For \ref{rows-mainQ}, note that, since $v^{\dd}$ is a conditionally
  $2$-diverse pairwise representation, for every $x,y\in X$,
\begin{linenomath*}
  \[0 \not \leq v^{\xy} = - \mathbf{v}(x,\cdot) + \mathbf{v}(y,\cdot).\]
\end{linenomath*}  
Finally, for every $z\in X$, we have, for every $\lambda \in \R$,
$v^{\zx} \neq \lambda v^{\zy}$. Equivalently,
\begin{linenomath*}
  \[ v(x,\cdot) \neq (1-\lambda) \mathbf{v}(z,\cdot ) +  \lambda
    \mathbf{v}(y,\cdot) . \]
\end{linenomath*}
\Cref{thm-main} part II, on uniqueness, follows from \cref{lem-induction} and,
without modification, part 3 of the proof of theorem 2 of \gsii\ (see page 23).
\end{step}

\begin{step}[the case of arbitrary $X$ and $\mbbt$]\label{step-infinite}
  Note that the inner product we have been working with throughout is
  well-defined on the infinite-dimensional, linear subspace
  $\R^{\oplus \mbbtpp} \subset \R^{\mbbtpp}$ of vectors with finite
  support. Indeed, for every $v \in \R^{\mbbtpp}$ and
  $J \in \R^{\oplus\mbbtpp}_{\mplus}$, the inner product is a finite sum
\begin{linenomath*}
  \[\langle v, J\rangle = \sum_{\{t:J(t)>0\}}v(t)J(t).\]
  \end{linenomath*}
  Next note that $\mbbjpp $ is just the set of rational-valued vectors in
  $ \R^{\oplus \mbbtpp}$. As such, the notions of orthogonality and hyperplanes
  carry over to the present, infinite-dimensional, setting.  \Wlog, let $\aext$
  be an improper $Y$-extension of $\preceq_{\mbbj}$ and note that, for every
  $Y\subseteq X$ of cardinality $2,3$ or $4$, the representation
  $\acute{u}^{\dd}$ of $\aext$ has finite rank $\acute{\mathbf r}$. The
  \emph{essentialization} $\ess(\acute{\mc H})$ of the associated arrangement
  $\acute{\mc H}$ in $ \R^{\oplus \mbbt}$ is the arrangement we obtain via the
  by orthogonally projecting $\acute{\mc H}$ onto the
  ($\acute{\mathbf r}$-dimensional) span of $\acute{u}^{\dd}$. Let
  $\textup{p}: \R^{\oplus \mbbt}\rightarrow S = \spann\{\acute{u}^{\dd}\}$
  denote this projection. Then, for any $J\in \mbbjpp$, $J$ is the sum of
  $\textup{p}(J)$ and a term that belongs to the kernel of $u^{\xy}$, for every
  $x,y \in Y$. Thus, for every $J \in \R^{\oplus\mbbt}$,
  $\langle u^{\xy}, J \rangle_{\,\R^{\oplus\mbbt}} = \langle u^{\xy},
  \textup{p}(J) \rangle_{S}$. As such, all the structure of $\acute{\mc H}$ in
  $\R^{\oplus\mbbt}$ is preserved by $\ess(\acute{\mc H})$ in the
  finite-dimensional subspace $S$.

  For $n = \acute{\mathbf{r}}$, take $\mc T= \{T_{i}: i = 1, \dots , n\} $ to be
  an orthonormal basis for $S$. Then let $T: R^{n} \rightarrow S $ to be the
  matrix with columns $\mc T$. Then, for every $J \in \R^{n}$,
  $T(J) = J_{1} T_{1} + \cdots + J_{n}T_{n}$, where $J_{i} \in \R$ for each
  $i$. Via orthonormality of $\mc T$,

\begin{linenomath*}
\begin{equation*}
\langle  T_{i} , T_{j} \rangle_{S} = \left\{ 
\begin{array}{ll}
1 & \text{if $i = j$, and}\\
0 & \text{otherwise.}
\end{array}
\right.
\end{equation*}
\end{linenomath*} 
\begin{linenomath*}
It then follows that, for every pair of vectors $v,J \in
  \R^{n}$, we have
\[ \langle T(v) , T(J) \rangle_{S} = \sum_{i,j = 1}^{n} v_{i}J_{j} \,\langle  T_{i},
   T_{j} \rangle_{S} = \langle v, J \rangle_{\,\R^{n}}.\]
\end{linenomath*}
Now, for each row $u^{\xy}$ of $u^{\dd}$, let
$\mathbf{\acute{u}}^{\xy} \defeq T^{-1} (\acute{u}^{\xy})$, and let
$\mathbf{\acute{u}}^{\dd}$ denote the matrix with rows
$\{\mathbf{\acute{u}}^{\xy}: x,y\in Y \}$. Then let
$\mathbf{\aext}= \langle \mathbf{\aext}_{J}: J \in \mbb Q_{\mplus}^{n} \rangle$
denote the extension generated by $\mathbf{\acute{u}}^{\dd}$ in
$\mbb Q_{\mplus}^{n}$, for every $x,y\in Y$ and every
$J \in \mbb Q_{\mplus}^{n}$, satisfy $x \mathbin{\mathbf{\aext}}_{J} y $ if, and
only if, $\langle \mathbf{\acute{u}}^{\xy},J \rangle_{\,\R^{n}} \geq 0$. Then,
by construction, $\langle \mathbf{\acute{u}}^{\xy},J \rangle_{\,\R^{n}} \geq 0$
if, and only if $\langle \acute{u}^{\xy}, T(J) \rangle_{S} \geq 0$. Thus,
$x \mathbin{\mathbf{\aext}}_{J} y $ if, and only if, $x \aext_{T(J)} y$.

Finally, note that, with the exception of \cref{lem-induction} of
\cref{step-induction}, all results in the proof of \cref{thm-mainQ} involve
$Y$-extensions $\ext$ of $\preceq_{\mbbj}$ for $\countof Y = 2,3$ or $4$.  By
taking $\mbbt_{Y} = \{1,\dots, n\}$ and substituting for $\mbbt$ and then
working with $\mathbf{\aext}$ and $\mathbf{\acute{u}}^{\dd}$ in $\R^{n}$, all
these results extend to arbitrary $X$ and $\mbbt$.  And, since
\cref{step-induction} holds for arbitrary $\mbbt$ and $X$, the proof of
\cref{thm-mainQ} is indeed complete.
\end{step}

\section{Statement and proof of \cref{thm-foureq}}\label{sec-proof-foureq}
 \begin{theorem}\label{thm-foureq}
   For $\preceq_{\mbbj}$ satisfying \ref{TQ}--\ref{c2dQ}, \fourpru\ holds if,
   and only if, the Jacobi identity holds (for some pairwise representation of
   $\preceq_{\mbbj}$).
 \end{theorem}
 Via \cref{lem-induction}, it suffices to show that \fourpru\ is equivalent to
 \fourjac. Then by construction, \fourpru\ and \fourjac\ are conditions that
 apply independently to each subset $Y$ of $X$ that has cardinality $3$ or $4$.
 Throughout the present section, $Y\subseteq X$ has cardinality $3$ or $4$ and
 $\ext$ denotes the improper $Y$-extension of $\preceq_{\mbbj}$. This will allow
 us to apply the translation of \cref{step-infinite} and work in a finite
 dimensional space. But first we need to rule out the following case.
 \paragraph{Suppose there is no testworthy $Y$-extension that is
   novel,}\hskip-7pt In
 this case, \fourpru\ holds vacuously on $Y$. The follow lemma confirms that, in
 this case, \fourdiv\ holds and theorem 2 of \gsii\ applies, so that \fourjac\
 holds on $Y$.
\begin{lemma}\label{lem-test-empty-fourdiv}
  If every testworthy $Y$-extension of
  $\preceq_{\mbbj}$ is regular, then $\lvert \mbbt \rvert =
  \infty $ and \fourdiv\ holds on $Y$.
\end{lemma}

\begin{proof}[Proof of \cref{lem-test-empty-fourdiv}] \label{proof-test-empty-fourdiv}

  Let $ \lvert Y \rvert = 4 $. (The proof for case where $ \lvert Y \rvert = 3 $
  is similar and omitted.)  Via \cref{prop-central-testworthy}, the set of
  testworthy $Y$-extensions is nonempty. Let $ \hext $ be a testworthy
  $ Y $-extension, so that there exists $J $ in $ \mbbj$ such that
  $\hext_{J\times 0 }$ is total and equal to the inverse of
  $\hextb_{\novel}$. Let $ P \defeq \hsextb_{\novel}$. Thus $P$ is the
  asymmetric part of a total ordering of $Y$.  Via \cref{lem-insep}, let
  $\hat{v}^{\dd}$ be the matrix representation of $\hext$ and let $\hat{v}^{P}$
  denote the restriction of $\hat{v}^{\dd}$ to
  $P \times \mbbtp$.
  \begin{claim}
    \label{claim-test-empty}
    For every vector
    $ \eta^{P} = \langle \eta^{\xy} \in \R_{\mdoubleplus }: \xy \in P \rangle $,
    there exist $ s, t \in \mbbt $ such that
    $\hat{v}^{P}(s)= \eta^{P}$ and $\hat{v}^{P}(t)= - \eta ^{P}$.
  \end{claim}
  
  \begin{proof}[Proof of \cref{claim-test-empty}]
    \label{proof-test-empty}

    By way of contradiction, suppose there exists
    $ \acute{\eta}^{P} \in \R^{P}_{\mdoubleplus}$ such that, for every $s$ in
    $ \mbbt$, $\hat{v}^{P}(s) \neq \acute{\eta}^{P}$. That is $\eta^{P}$ such
    that, for every $ s $ in $ \mbbt $, there exists $ \xy $ in $ P $ such that
    $ \hat{v} ^{ \xy} (s) \neq \acute{\eta}^{\xy } $. This property will suffice
    for the existence of a testworthy $Y$-extension $\aext$ that is novel.
    Define $ \acute{v}^{\dd}: X^{2} \times \mbbtp \rightarrow \R$ as
    follows. For each $ \xy $ in $ P $, let
\begin{linenomath*}
  \begin{equation*}
  \acute{v} ^{ \xy } (s)\defeq
  \left\{
    \begin{array}{ll}
      \acute{\eta}^{\xy} & \text{if $s  = \novel $,}\\
     \hat {v}^{\xy}(s) & \text{otherwise}.
      \end{array}
    \right.
  \end{equation*}
\end{linenomath*}
  For every $ \xy $ in $P^{-1}$, since $\yx $ in $ P$, take
  $ \acute{v}^{\xy} = - \acute{v} ^{ \yx } $. Finally, since $P $ is the
  asymmetric part of a total ordering, for every remaining $ \xy $ in
  $ Y ^{ 2 }$, $x = y$, so let $ \acute{v} ^{ \xy } = 0 $.  Observe that, by
  construction, for every $s $ in $ \mbbt$,
  $\acute{v}^{\dd}(s) \neq \acute{v}^{\dd}(\novel)$. This allows us to appeal to
  \cref{lem-insep} and take $\aextb$ to be the novel extension that
  $\acute{v}^{\dd}$ generates.
  Moreover, $\aext $ is also testworthy. For together $P = \hsextb_{\novel}$ and
  $\hsextb_{\novel} = \hsextb_{J\times 0 }^{-1}$ imply that
  $ \aextb _{ \novel } = \aextb _{ J \times 0 } ^{ - 1 }$ since, for every
  $ \xy $ in $ P$, we have
\begin{linenomath*}
  \[ \langle \acute{v} ^{ \xy } , J \times 0 \rangle = \langle \hat{v} ^{ \xy } , J
    \times 0 \rangle< 0 < \acute{\eta}^{\xy} = \acute{v} ^{ \xy } (\novel).\]
\end{linenomath*}


This contradiction implies that, for every $ \eta ^{ P } \in \posreal^{P} $,
there exists $s $ in $\mbbt $ such that $\hat{v}^{P}(s) = \eta^{P}$.
\emph{Mutatis mutandis}, a repetition of the preceding argument by contradiction
confirms that there exists $t $ in $ \mbbt $ such that
$ \hat{v}^{P}(t) = - \eta^{P}$.
\end{proof}
\Cref{claim-test-empty} implies that, when every testworthy $Y$-extension is
regular, the cardinality of $\mbbt $ is equal to the cardinality of $\R^{P}$.
We now show that \fourdiv\ holds on $Y$. Let $ {R} $ denote an arbitrary
total ordering of $ Y $.  We show that, for some $ K $ in $ \mbbj $,
$ \langle \hat{v} ^{ \xy }, K \rangle \geq 0 $ if, and only if, $ \xy $ in
${R} $. \Cref{claim-test-empty} ensures that we can choose
$ s $ in $ \mbbt $ such that, for some $0 < \epsilon <1$
\begin{linenomath*}
  \begin{equation*}
    \hat{v} ^{ \xy } ( s ) =
    \left \{
      \begin{array}{l l}
        1+ \epsilon & \text{ if $ \xy $ in $ {R} \cap P  $},\\
        1-\epsilon & \text{if $ \xy $ in $ {R} ^{ - 1 } \cap P   $}.
      \end{array}
    \right.
  \end{equation*}
\end{linenomath*}
  Via \cref{claim-test-empty}, take $ t$ in $ \mbbt $ such that, for every
  $ \xy $ in $ P$, $ \hat{v}^{ \xy } ( t ) = -1 $.  Let
  $ K : = \delta _{ s } + \delta _{ t } $ in $ \mbbjp $, so that
  $\langle \hat{v}^{\xy}, K \rangle = \hat{v}^{\xy}(s) + \hat{v}^{\xy}(t)$. By evaluating
  terms and observing that $\epsilon > 0 $ we obtain
\begin{linenomath*}  
 \begin{equation*}
   \langle  \hat{v}^{\xy}, K \rangle =
   \left \{
      \begin{array}{l l}
        (1+ \epsilon) - 1 > 0  & \text{if $ \xy $ in $ {R} \cap P   $},\\
                 (1- \epsilon) - 1 < 0 & \text{if $\xy $ in ${R}^{-1}\cap P $}.
      \end{array}
    \right.
  \end{equation*}
\end{linenomath*}  
  Since $\xy $ in ${R}^{-1} \cap P^{-1} $ if, and only if,
  $\yx $ in ${R}\cap P$ (and, similarly,
  $\xy $ in ${R} \cap P^{-1}$ if, and only if $ \yx$ in
  ${R}^{-1}\cap P$), we appeal to $\hat{v}^{\xy} = - \hat{v}^{\yx}$ and obtain
\begin{linenomath*}
 \begin{equation*}
   \langle  \hat{v}^{\xy}, K \rangle =
   \left \{
      \begin{array}{l l}
        -( 1 + \epsilon) + 1 < 0 & \text{if $ \xy $ in $ {R}^{-1} \cap P^{-1}$,}\\
        -(1- \epsilon) + 1 > 0  & \text{if $ \xy $ in $ {R} \cap P^{-1}  $.}
      \end{array}
    \right.
  \end{equation*}
\end{linenomath*}
Since
$P$ is the asymmetric part of a total ordering we conclude that, for every
$x\neq y$, $\langle \hat{v}^{\xy}, K \rangle
$ has the right sign. Finally, for $x = y $, $\langle \hat{v}^{\xy}, K \rangle =
0$.
\end{proof}

As a consequence of \cref{lem-test-empty-fourdiv}, for the remainder of the proof of
\cref{thm-foureq} we work under the assumption that the set of testworthy
$ Y $-extensions that are novel is nonempty. 

\paragraph{For the remainder of this section, redefine
  $\mbbt : = \mbbt_{Y} $,}\hskip-7pt where recall  $\mbbt_{Y}$ is defined in \cref{step-infinite}.  Thus, in the present section, $\mbbt$ is the finite
set that indexes the translation to $\R^{n}$ of \cref{step-infinite}. Then
$\mbbtp = \mbbt \cup[\novel]$ as before and indeed all other notation remains
the same.  Whereas $\mbbt = \mbbt_{Y}$ is now finite, the set
$\mbbc_{/\sim^{\ext}}$ of equivalence classes of $\sim^{\ext}$ in $\mbbc $ may
be infinite as we have seen in \cref{lem-test-empty-fourdiv}.

 Since \ref{KQ}--\ref{c2dQ} hold, the pairwise representation
 $ u^{\dd}: Y^{2}\times \mbbt \rightarrow \R $ of $\ext$ is conditionally
 $2$-diverse.  \Cref{lem-c2dQ} implies that $ u^{\dd}$ has row rank
 $\mathbf r \geq 2$.


Recall the definition of central arrangements. We now show that the set of
testworthy extensions with a central arrangement is always nonempty.
\begin{lemma}\label{prop-central-testworthy}
  For every $J \in \mbbj$ such that $\ext_{J}$ is total, there exists a
  $Y$-extension $\aext $ with $\aextb_{\novel } = \extb_{J } ^{-1}$. Moreover,
  its arrangement $\acute{\mc H}_{\mdoubleplus}$ is central$;$ its
  representation $\acute{u}^\dd$ has rank $ \acute{\mathbf r}= \mathbf r
  $$;$ and \fourjac\ holds for $\aext$ if, and only if, it holds for $\ext$.
\end{lemma}
\begin{proof} Let $J \in \mbbj$ such that $\ext_{J}$ is total. If, for some
  $t \in \mbbt$, $ J(t) = 0 $, then, via $\binom{\countof Y}{2}$ applications of
  \ref{AQ} (or the continuity properties of the inner product), there exists
  there $L \in \posreal^{\mbbt}$ such that $\extb_{L} = \extb_{J}$. Thus,
  \withoutlog, take $J(t)>0$ for every $t\in \mbbt$. For some rational
  $0<\acute{\iota}<1$, let
  $\acute{J} = (1-\acute{\iota})J \times \acute{\iota} $, so that $\acute{J}$
  lies on the (relative) interior of $ \mbbjp $.  Now, for every $x,y$ in $ Y$,
  let
  \begin{linenomath*}
  \[\acute{\eta}^{\xy} := -\frac{1- \acute{\iota}}{\acute{\iota}}\langle
    u^{\xy}, J \rangle,
  \]
  \end{linenomath*}
  and let $\acute{u}^{\xy}\defeq u^{\xy}\times \acute \eta^{\xy}$, so that
  $\langle \acute{u}^{\xy}, \acute{J} \rangle = 0$. Let $\aext$ be the
  associated $Y$-extension, so that by construction
  $\aextb_{\novel} = \aextb_{J}^{-1}$.  Since
  $\acute{J} \in \acute{H}^{\{x,y\}}$ for every distinct $x,y\in Y$,
  $\acute{\mc H}$ is central. By construction, note that
  $\mathbf r \leq \acute{\mathbf r} \leq \mathbf r + 1$. Moreover, recalling the
  rank--nullity theorem, we have
  $\acute{\mathbf r} = \countof \mbbtp - \textup
  {dim}(\textup{ker}(\acute{u}^{\dd}))$. Now, $\acute{J} \in \R^{\mbbtp}$
  belongs to $ \textup{ker}(\acute{u}^{\dd})$, but not
  $\textup{ker}(u^{\dd})\times 0 \subset \R^{\mbbtp}$. Moreover, the latter set
  belongs to $\textup{ker}(\acute{u}^{\dd})$. Thus,
  $\textup{ker}(\acute{u}^{\dd})$ is of dimension one more than
  $\textup{ker}(u^{\dd})\subset \R^{\mbbt}$. Then since
  $\countof \mbbtp = \mbbt +1$, the rank--nullity theorem yields
  $\acute{\mathbf r} = \mathbf r$.  Finally, via linearity of the inner product,
  the Jacobi equations \eqref{eq-xy}--\eqref{eq-xz} hold for $u^{\dd}$ if, and
  only if, they hold for $\acute{\eta}^{\dd}$ and hence $\acute{u}^{\dd}$.
\end{proof}

\paragraph{The case where $Y$ has cardinality
  $3$,\hskip-10pt} follows from the next two lemmas.
  \begin{lemma}\label{lem-Y3-r2} 
    If $Y = \{x,y,z\} $ and $\mathbf r =
    2$, then \fourjac\ and \fourpru\ hold on $Y$.
  \end{lemma}
  \begin{proof}[Proof of \cref{lem-Y3-r2}]
    Fix arbitrary $J \in \mbbj$ such that $\ext_{J}$ is total.  We first apply
    Zaslavski's theorem to prove that the $Y$-extension $\aext$ of
    \cref{prop-central-testworthy} (which is testworthy relative to $J$)
    satisfies $\lvert \acute{\mc G}_{\mdoubleplus}\rvert = 6$.  Via
    \cref{lem-c2dQ}, $\lvert\acute{\mc H}_{\mdoubleplus}\rvert= 3$. And since
    every subarrangement of $\acute{\mc H}_{\mdoubleplus}$ is central, for every
    $k = 0, 1, 2, 3$ there are $\binom{ 3 }{k }$ ways to choose
    $\lvert \mc A \rvert = k $ hyperplanes from $ \acute{\mc
      H}_{\mdoubleplus}$. For $k < 3$, the rank of every subarrangement is
    $k$. For $k = 3$, the rank of the arrangement is $\acute{\mathbf r}$. So
    that, via \cref{prop-central-testworthy}, since $\acute{\mathbf r} = \mathbf
    r = 2$,  
\begin{linenomath*}
        \begin{equation}\label{eq-zaslavski-3}
                \begin{aligned}
                  \lvert \acute{\mc G}_{\mdoubleplus}\rvert =& \binom{3}{3}
                  (-1)^{3 -\acute{\mathbf{r}}}+\binom{3}{2}(-1)^{2-2}
                  +\binom{3}{1}(-1)^{1-1} + \binom{3}{0}(-1)^{0-0} = 6. 
                \end{aligned}
        \end{equation}
\end{linenomath*}
For both \fourjac\ and \fourpru, we require that every member of $\acute{\mc
  G}_{\mdoubleplus}$ is associated with a total ordering. This ensures that
\threediv\ holds for
$\aext$, so that, via lemma 2 of \gsii, \threejac\ holds for
$\aext$ and \cref{prop-central-testworthy} yields \threejac\ for
$\ext$. (When $\lvert Y \rvert =
  3$, \threejac\ and \fourjac\ coincide.) Also, for \fourpru, if $ \hext
$ is any novel, testable extension that satisfies $
\hextb_{\novel}=\extb_{J}^{-1}$, then
$\aext$ is a nondogmatic perturbation of
$\hext$ that satisfies \ref{TQ}-\ref{AQ}.
   
It therefore remains for us to show that every member of
$\acute{\mc G}_{\mdoubleplus}$ is associated with a total ordering of $Y$.
Since $\acute{\mc H}_{\mdoubleplus}$ is central,
$\acute{\mc G}_{\mdoubleplus} \equiv \acute{\mc G}$, we
suppress reference to $\mdoubleplus$. Let $\{\acute{G}^{r}: r = 1, \dots 6\}$ be
an enumeration of $\acute{\mc G}$ such that $\acute{G}^{r}$ and
$\acute{G}^{r+1}$ (\textup{mod} 6) are separated by a single hyperplane in
$\acute{\mc H}$. Moreover, let the first four members of this enumeration
intersect $\mbbj \times \{0\}$. Then, since
$\extb = \langle \extb_{I}: I \in \mbbj \rangle$ satisfies \ref{TQ},
$\acute{G}^{1}, \dots, \acute{G}^{4}$ are associated with total orders. Take
$L \in \acute{G}^{5} \cap \mbbjp$ and consider the affine path
$\lambda \mapsto \phi(\lambda)= (1-\lambda)\acute{L} + \lambda \acute{J}$, where
$\acute{J}$ is defined in \cref{prop-central-testworthy}. For some rational
$\lambda^{*}$ sufficiently close to, but greater than, $1$,
$L^{*} = \phi(\lambda^{*}) $ belongs to $\acute{G}^{*} \cap \mbbjp$ for some
$\acute{G}^{*} \in \acute{\mc G}$. Since $\phi(1) = \acute{J}$ belongs to the
center of $\mc H$, $L$ and $L^{*}$ are separated by all three members of
$\mc H$, so that $\aextb_{L^{*}} = \aextb_{L}^{-1} $. Thus, $\aextb_{L}$ is
transitive if, and only if, $\aextb_{L^{*}}$ is. Since $\acute{G}^{5}$ is
separated from $\acute{G}^{*}$ by $3$ hyperplanes,
$\acute{G}^{*} = \acute{G}^{2}$. Thus $\aextb_{L^{*}}$ is transitive.
\emph{Mutatis mutandis}, the same argument shows that $\acute{G}^{6}$ is also
associated with a total ordering of $Y$.
\end{proof}

\begin{lemma}\label{lem-Y3-r3}
  If $Y = \{x,y,z\} $ and  $ \mathbf{r} = 3 $, then neither \fourpru\ nor
  \fourjac\ hold.
\end{lemma}
\begin{proof}[Proof of \cref{lem-Y3-r3}]
  When $\mathbf r = 3$, $u^{\xy}, u^{\yz}$ and $ u^{\xz}$ are linearly
  independent, so that \threejac\ fails to hold.

  We now confirm that \threepru\ also fails to hold.  Via
  \cref{prop-central-testworthy}, $\acute{\mathbf r} = \mathbf r $.  The only
  term in \cref{eq-zaslavski-3} that we adjust for the present lemma is
  $\mathbf{r}=3$. Thus,
\begin{linenomath*}
  \begin{align*}
    \lvert \acute{\mc G}\rvert
    = (-1)^{0}+3(-1)^{0}+ 3(-1)^{0}+(-1)^{0} = 8.
  \end{align*}
\end{linenomath*}
\addtocounter{linenumber}{-1}
  Then there are $3!=6$ members of $\total (\aext)$, and the two additional
  regions of $\acute {\mc G}$ are associated with intransitive CAR rankings. It
  remains for us to show that every $Y$-extension $\hext$ with
  $\countof \total (\hext) = 6$ is of this form. 
  
  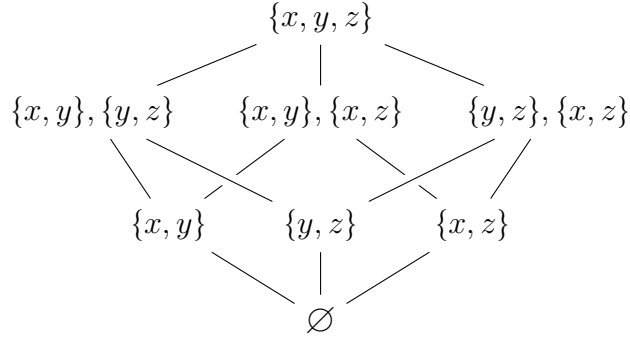
\begin{figure}
    \begin{center}
  \begin{tikzpicture}
    \node (max) at (0,3) {$\{x,y,z\}$}; \node (a) at (-3,1.75)
    {$\{x,y\},\{y,z\}$}; \node (b) at (0,1.75) {$\{x,y\},\{x,z\}$}; \node (c) at
    (3,1.75) {$\{y,z\},\{x,z\}$}; \node (d) at (-2,0.25) {$\{x,y\}$}; \node (e)
    at (0,0.25) {$\{y,z\}$}; \node (f) at (2,0.25) {$\{x,z\}$}; \node (min) at
    (0,-1) {$\emptyset$}; \draw (min) -- (d) -- (a) -- (max) -- (b) -- (f) (e)
    -- (min) -- (f) -- (c) -- (max) (d) -- (b); \draw[preaction={draw=white,
      -,line width=6pt}] (a) -- (e) -- (c);
\end{tikzpicture}
\caption{\label{fig-hasse-Y3-r3} The intersection lattice of a central
  arrangement for $\countof Y = 3$ and $ {\mathbf r} =3$.}
\end{center}
\end{figure}

In the Hasse diagram of \cref{fig-hasse-Y3-r3}, an increase in level corresponds
to a decrease in dimension: since $\acute{A}^{\{x,y,z\}} $ is nonempty, it is of
dimension at least zero.  Since $\acute{J}$ belongs to the interior of
$\posreal^{\mbbtp}$ and $\acute{A}^{\{x,y\}\{y,z\}}$ is at least
one-dimensional, $ \acute{A}^{\{x,y\} \{y,z\}}_{\mdoubleplus}$ is
one-dimensional.  Since
$\acute{A}^{\{x,y,z\}} \subset \acute{A}^{\{x,y\}\{y,z\}}$, the latter set is
nonempty whenever $\acute{\mc H}_{\mdoubleplus}$ is central. The same, of
course, applies to other members at the same level. Conversely, if
$\acute{A}^{\{x,y\}\{y,z\}}$ is empty, then so is $\acute{A}^{\{x,y,z\}}$. We
now use this to show that there is a unique form of $Y$-extension $\hext$ such
that $\countof \hat{\mc G}_{\mdoubleplus} = 6$ and, moreover, that any such
$\hext$ fails to satisfy \ref{TQ}.

Recall from \cref{eg-zaslavski} that $\bmu(A^{\emptyset}) = 1$ and
$\bmu(A^{\{i,j\}}) = -\bmu(A^{\emptyset})$. Then, since
$\hat{A}^{\{x,y\}\{y,z\}} \subset \hat{A}^{\{x,y\}} , \hat{A}^{\{y,z\}}$, we
have $\bmu(A^{\{x,y\}\{y,z\}}) = - \bmu(A^{\emptyset})(1 - 2) = 1$. Finally, by
the same argument, since $\hat{A}^{\{x,y,z\}}$ belongs to every member of
$\hat{\mc L}$, we have
$\bmu(A^{\{x,y,z\}}) = - \bmu(A^{\emptyset})(1 - 3 + 3) = -1$.  Thus, when
$\mathbf{r} = 3$, if $\countof \hat{\mc G}_{\mdoubleplus} = 6$, then, in
addition to $\hat{A}^{\{x,y,z\}}_{\mdoubleplus}$, at most one of the
intersections $\hat{A}^{\{i,j\}\{k,l\}}_{\mdoubleplus}$ is empty. Then,
\withoutlog, suppose $\hat{A}^{\{x,y\}\{y,z\}}_{\mdoubleplus} $ is nonempty. We
will show that \ref{TQ} fails to hold. Take
$ L \in \hat{A}^{\{x,y\}\{y,z\}}_{\mdoubleplus} $. If $L \in \mbbjp$, then since
$x \hnext_{L}y \hnext_{L} z$ and $\neg (x \hnext_{L} z)$, the proof is
complete. In any case, $L$ belongs to one $\hat{G}^{\xz}_{\mdoubleplus}$ or
$\hat{G}^{\zx}_{\mdoubleplus}$. \Wlog, suppose
$L \in \hat{G}^{\xz}_{\mdoubleplus}$. Then $\hat{H}^{\{x,y\}}$ and
$\hat{H}^{\{y,z\}}$ split $\hat{G}^{\xz}_{\mdoubleplus}$ into the four
(nonempty) open regions formed by expanding
\begin{linenomath*}
  \[\hat{G}^{\xz}_{\mdoubleplus}\bs \left(\hat{H}^{\{x,y\}} \cap
      \hat{H}^{\{y,z\}}\right) = \hat{G}^{\xz}_{\mdoubleplus}\cap
    \left(\hat{G}^{\xy}_{\mdoubleplus} \cup \hat{G} ^{\yx}_{\mdoubleplus}
    \right) \cap \left (\hat{G}^{\yz}_{\mdoubleplus} \cup
      \hat{G}^{\zy}_{\mdoubleplus} \right). \]
\end{linenomath*}
Take any $K$ in the final member
$ \hat{G}^{\xz}_{\mdoubleplus}\cap \hat{G}^{\yx}_{\mdoubleplus} \cap
\hat{G}^{\zy}_{\mdoubleplus}$ of this expansion. Then, since
$ x \hsext_{K} z \hsext_{K} y \hsext_{K} x$, we observe that $\hext$ fails to
satisfy \ref{TQ}. We conclude that only dogmatic perturbations of $\aext$
satisfy \ref{TQ}.
\end{proof}
\paragraph{The case where $Y$ has cardinality $4$.\hskip-4pt} Note that a
failure of \threejac\ on $Z \subset Y$ such that $\lvert Z \rvert = 3$ implies a
failure of \fourjac\ on $Y$. And since the arguments for the case where
$\lvert Y \rvert = 3$ account for the case where \threejac\ fails, we henceforth
assume that \threejac\ holds on $Y$. That is, our conditionally \emph{2}-diverse
representation $u^{\dd}$ will now satisfy equations \eqref{eq-xy}--\eqref{eq-xz}
with $\hat{\beta} = \beta$ if, and only if, \fourjac\ holds on $Y$.

First some some useful results that exploit \threejac.
\begin{proposition}\label{lem-2r3}
  If $ Y = \{x, y, z, w\} $ and \threejac\ holds for $u^{\dd}$, then, for every
  $\acute{v}^{\dd} = u ^{\dd} \times \acute{\eta}^{\dd}$ with rank
  $ \acute{\mathbf r}$, that satisfies \threejac,
  $2\leq \acute{\mathbf r}\leq 3$.
\end{proposition}
\begin{proof}[Proof of \cref{lem-2r3}]
  Via \cref{prop-c2dQ} and \cref{lem-Y3-r2}, $\mathbf r \geq 2 $. Indeed the
  span of $\{ u^{\xw}, u^{\yw}\}$ is two. Let $S $ denote the span of
  $ \{u^{\xw}, u^{\yw} ,u^{\zw}\}$.  Since $u^{\yw} = - u^{\wy}$ and $ u ^{\dd}$
  satisfies \threejac, equations \eqref{eq-xy}--\eqref{eq-xz} hold for
  $u^{\dd}$.  (If \fourjac\ fails to hold, then $\beta \neq \hat \beta$, but the
  equations still hold.) Thus $u^{\xy}$, $u^{\yz}$ and $u^{\xz}$ all belong to
  $S$ and $\mathbf{r} \leq 3$. Now note that above argument does not depend on
  the cardinality of $\mbbt$, thus take $\acute{\eta}^{\dd}$ to
  satisfy \threejac: indeed with the same parameters that feature in
  equations \eqref{eq-xy}--\eqref{eq-xz} for $u^{\dd}$. The preceding
  argument then extends \emph{mutatis mutandis} to $\acute{v}^{\dd}$ and
  $2 \leq \acute{\mathbf r} \leq 3$.
\end{proof}

\begin{proposition}\label{lem-arrangement-cardinality-rank}
  If $X = \{ x, y , z , w\}$, then $4 \leq \lvert \mc H \rvert \leq 6$ and
  these bounds are tight. If, moreover, \threejac\ holds for $u^{\dd}$ and
  $\lvert \mc H \rvert < 6 $, then $\mathbf{r} = 2$.
\end{proposition}
\begin{proof}[Proof of \cref{lem-arrangement-cardinality-rank}]
  The upper bound $\lvert \mc H \rvert \leq 6$ follows from the fact that there
  are $\binom{4}{2} = 6 $ ways to choose distinct pairs of elements from the
  four-element set $X$.  Via \cref{lem-c2dQ} only the following equalities are
  feasible: $H^{\{x,y\}} = H^{\{z,w\}}$, $H^{\{x,z\}} = H^{\{y,w\}}$ and
  $H^{\{x,w\}} = H^{\{y,z\}}$.  Via \cref{prop-pairwise-extremal}, \withoutlog, take
  $ G^{(x,y,z,w)}$ and $G^{(w,z,y,x)}$ to belong to $\mc G$. We consider a
  convex path $\phi$ from (some member of) $G^{(x,y,z,w)}$ to $G^{(w,z,y,x)}$
  such that $\countof \{G \in \mc G: (\image \phi) \cap G \neq \emptyset \} $ is
  maximal. Via \ref{TQ}, one of $H^{\{x,y\}} , H^{\{y,z\}}$ and $H^{\{z,w\}}$
  supports $G^{(x,y,z,w)}$.  If it is $H^{\{y,z\}}$, then so does $H^{\{x,w\}}$
  and, and since $x$ and $w$ lie at opposite ends of the ordering $(x,y,z,w)$,
  this would imply that $\lvert\mc H \rvert = 1$ in violation of
  \cref{lem-c2dQ}. Since we are seeking a greatest lower bound on
  $\lvert \mc H \rvert$, take $ H_{1} = H^{\{x,y\}} = H^{\{z,w\}}$ to be the
  first hyperplane to intersect $\image \phi$ and the adjacent chamber is
  $G^{(y,x,w,z)}$.  Via \ref{TQ}, the second hyperplane to intersect
  $\image \phi$ is $H_{2} \defeq H^{\{x,w\}}$. If $H^{\{x,w\}} = H^{\{y,z\}}$
  then since $y$ and $z$ lie at opposite ends of $(y,x,w,z)$, this would once
  again lead to a violation of \cref{lem-c2dQ}. Then $G^{(y,w,x,z)}$ is next
  chamber to intersect $\image \phi$. Since we are seeking a greatest lower
  bound, via \ref{TQ}, take $ H_{3} \defeq H^{\{y,w\}} = H^{\{x,z\}}$ to be the
  third hyperplane to intersect $\image \phi$. Then $G^{(w,y,z,x)}$ is the next
  chamber to intersect $\image \phi$. The final hyperplane to intersect
  $\image \phi$ is $H_{4} = H^{\{y,z\}}$ which brings us to the destination
  chamber $G^{(w,z,y,x)}$.  We conclude that at most two pairs of the six
  hyperplanes coincide, so that $\lvert \mc H \rvert \geq 4 $.
  
  We now prove that \threejac\ and $ H^{\{x,y\}} = H^{\{z,w\}}$ together imply
  $\mathbf r = 2$. Consider equations \eqref{eq-xy}--\eqref{eq-xz} (so
  \threejac\ holds, but \fourjac\ need not). Via \eqref{eq-xy},
  $S = \{ u^{\xw}, u^{\wy}, u^{\xy}\}$ is 2-dimensional. Since $ u^{\xy} $ and
  $ u^{\zw }$ are collinear, $ u^{\zw} $ belongs to $ S $. Finally, equations
  \eqref{eq-yz} and \eqref{eq-xz} yield $u^{yz} , u^{\xz} \in S $.
\end{proof}
\begin{lemma}\label{lem-r3-Y4}
  If $Y = \{x,y,z,w\}$, $ \mathbf{r} = 3 $ and \threejac\ holds on $Y$, then \fourpru\
  and \fourjac\ hold on $Y$
\end{lemma}
\begin{proof}[Proof of \cref{lem-r3-Y4}]
  To see that, \fourjac\ holds, we appeal to the proof of lemma 3 of \gsii: if
  \threejac\ holds and \fourjac\ does not, then $\{u^{\xw}, u^{\yw}, u^{\zw}\}$
  is linearly dependent. This is in contradiction of $\mathbf{r} = 3$.
  
  It remains for us to verify that \fourpru\ also holds. We will show that, for
  every $J\in \mbbj$ such that $\ext_{J}$ is total, the $Y$-extension $\aext$ of
  \cref{prop-central-testworthy} satisfies
  $\lvert \acute{\mc G}_{\mdoubleplus} \rvert = 4!$. We then confirm that
  $\aext$ satisfies \ref{TQ}, so that every member of $\acute{\mc G}$ is
  associated with a total ordering of $Y$. Then $\total(\acute{\ext})$ is
  maximal and, if $ \cext $ is any novel, testworthy extension that satisfies
  $ \cextb_{\novel }=\extb_{J }^{-1}$, then $\aext$ is a nondogmatic
  perturbation of $\cext$ that satisfies \ref{TQ}-\ref{AQ}.

  Since $\mathbf{r}=3$, the contrapositive of
  \cref{lem-arrangement-cardinality-rank} implies that
  $\lvert\mc H_{\mdoubleplus}\rvert = 6 $.  The proof that now follows is a
  simple extension of \cref{lem-Y3-r2} to allow for the fact that
  $\lvert Y \rvert=4$.  Fix any $J \in \mbbj$ such that $\ext_{J}$ is total and
  \withoutlog\ take $J>0$.  Via \cref{prop-central-testworthy}, there exists a
  testworthy $Y$-extension $\aext$ such that
  $ \aextb_{\novel }= \extb_{J }^{-1}$ is total and
  $L = (1- \lambda)J \times \lambda$ belongs to both $ \mbbjp$ and the center of
  $\acute{\mc H}_{\mdoubleplus}$ and $\acute{v}^{\dd}$ satisfies \fourjac. Via
  \cref{rem-lattice}, since $\acute{\mc H}$ is central,
  $\acute{\mc G}_{\mdoubleplus}$ and $\acute{\mc G}$ are isomorphic, so we work
  with $\acute{\mc G}$.  Since $\mathbf r = 3$,
  \cref{lem-arrangement-cardinality-rank} implies $\lvert \mc H \rvert = 6$, so
  that the same is true of $\lvert \acute{\mc H} \rvert$.  The rank of
  subarrangements with cardinality $4$ or more is $\acute{\mathbf r}$.  Let
  $\acute{\mathbf{\uptau}}$ denote the number of subarrangements $\mc A $ that
  have cardinality $3$ and rank $2$. Each of the other
  $\binom{6}{3}- \acute{\uptau}$ subarrangements with cardinality $3$ have rank
  $ \acute{\mathbf r} $. All other subarrangements have rank equal to their
  cardinality.
\begin{linenomath*}
  \begin {equation}\label{eq-zaslavski-4}
    \begin{aligned}
      \lvert \acute{ \mc G} \rvert =&\binom{6}{6} (-1)^{6 -
        \acute{\mathbf{r}}}+\binom{6}{5}
      (-1)^{5 - \acute{\mathbf{r}}}+  \binom{6}{4}(-1)^{4-\acute{\mathbf{r}}}\\
      +&\binom{6}{3}(-1)^{3-\acute{\mathbf{r}}}-\,\,\,\,\acute{\mathbf{\uptau}}
      \,\,\,\,\,(-1)^{3-\acute{\mathbf{r}}}
      + \,\,\,\,\acute{\mathbf{\uptau}} \,\,\,\,\,(-1)^{3-2}\\
      +&\binom{6}{2}(-1)^{2-2}+\binom{6}{1}(-1)^{1-1} + \binom{6}{0}(-1)^{0-0}
    \end{aligned}
  \end{equation}
\end{linenomath*}  
Via \cref{prop-central-testworthy}, $\acute{\mathbf r} = \mathbf r$, so that
$\acute{\mathbf r} = 3$.  It remains for us to calculate the value of
$\acute{\uptau}$. Each of the $\binom{4}{3} = 4$ subsets of $Y$ that have
cardinality $3$ generates a subarrangement of cardinality $ \binom{3}{2} =
3$. (For instance,
$\mc A^{\{x,y,z\}} =
\{\acute{H}^{\{x,y\}},\acute{H}^{\{y,z\}},\acute{H}^{\{x,z\}}\}$.)  For such
subarrangements, \fourjac\ implies a rank of $2$. Arguments from the final step
in the proof of \cref{lem-arrangement-cardinality-rank} confirm that every other
subarrangement with cardinality $3$ has rank $ 3 $.  Substituting into
\cref{eq-zaslavski-4}, we obtain
\begin{linenomath*}
  \begin{equation*}
    \begin{aligned}
    \lvert  \acute{ \mc G} \rvert =
    -1+6-15+20-4-4+15+6+1=24 = 4!\hskip.5pt.
   \end{aligned}
 \end{equation*}
\end{linenomath*}
 The proof that every member of $\mc G$ is associated with a transitive ordering
 is as follows. Via \cref{prop-central-testworthy}, \fourjac\ holds for $\aext$
 if, and only if, it holds for $\ext$. We have seen that \fourjac\ holds. Then,
 via a straightforward application of the Jacobi identity and, for arbitary
 $K \in \mbbjp$, the definition of transitivity of $\aext_{K}$ yields
 \ref{TQ}. Thus, every member of $\acute{\mc G}$ is total.
  \end{proof}

  \paragraph{In the remaining case, where $Y = \{x,y,z,w\}$ and
    $\mathbf r = 2$,\hskip-8pt}
 the proof is complicated by the fact that maximally-diverse extensions have centerless
 arrangements.
 We begin by choosing $\acute{\eta}^{\dd}$ so as to construct
 $\acute{u}^{\dd} = u^{\dd} \times \acute{\eta}^{\dd}$ with $\acute{\mathbf r} = 3$.

 Since $\mathbf r = 2$, it follows that $u^{\xz}, u^{\yz}$ and $u^{\wz}$ form a
 linearly dependent set. Thus, for some $ \pi, \rho \in \R$,
\begin{linenomath*} 
  \begin{equation}\label{eq-iz}
   \pi  u^{\xz} + \rho u^{\yz} = u^{\wz}.
 \end{equation}
\end{linenomath*}
Fix arbitrary $J \in \mbbj$ such that $\ext_{J}$ is total, and, as in
\cref{prop-central-testworthy}, \withoutlog, we take
$J \in G_{\mdoubleplus}^{(x,y,z,w)}$.  Then, for $\acute{\iota} = \half $ and
$\acute{J}^{\{x,y,z\}} = (1- \acute{\iota}) J \times \acute{\iota} $, let
\begin{linenomath*}
  \begin{equation}\label{eq-xyz}
   \acute{\eta}^{\ij} = - \tfrac{1-\acute{\iota}}{\acute{\iota}} \langle
   u^{\ij}, J \rangle = - \langle u^{\ij}, J \rangle, \quad \text{for every $i,j\in \{x,y,z\}$.}
 \end{equation}
\end{linenomath*}
Note that $\acute{\eta}^{\yz}>0$ is determined by \cref{eq-xyz}. Since
 $G^{(x,y,z,w)}_{\mdoubleplus}$ is open and the inner product is continuous,
 there exists a compact neighbourhood $N_{J} \subset G^{(x,y,z,w)}_{\mdoubleplus}$ of
 $J$ such that, for every $L\in N_{J}$,
 $\langle u^{\yz}, L\rangle < - \langle u^{\wz}, L \rangle$ if, and only if,
 $\langle u^{\yz},J \rangle <-\langle u^{\wz},J \rangle$. For any such
 $L \neq J$, take $\acute{\lambda}$ and $\epsilon$ to be solutions to
\begin{linenomath*}
  \begin{equation}\label{eq-yzw}
    \acute{\eta}^{\yz} = -\tfrac{1 - \acute{\lambda}}{\acute{\lambda}} \langle
    u^{\yz}, L \rangle \quad \text{and} \quad \epsilon = \langle u^{\wz},
    \tfrac{1-\acute{\lambda}}{\acute{\lambda}} L - J  \rangle.
  \end{equation}
\end{linenomath*}
A rearrangement of the first equality in \cref{eq-yzw} and, via \cref{eq-xyz}, a
substitution for $\acute{\eta}^{\yz}$ yields
$\acute{\lambda} = \frac{\langle u^{\yz}, L\rangle}{ \langle u^{\yz}, J \rangle
  + \langle u^{\yz}, L \rangle}$ and
$ \frac{1- \acute{\lambda}}{\acute{\lambda}} = \frac{\langle u^{\yz},J
  \rangle}{\langle u^{\yz}, L \rangle} $. Since $N_{J} \subset G^{\yz}$,
$0 < \acute{\lambda}<1$ and
$\acute{L}^{\{y,z,w\}} \defeq (1 - \acute{\lambda}) L \times \acute{\lambda} \in
\posreal^{\mbbtp}$ for every such $L$. On $N_{J}$, the denominator of the map
$L \mapsto \acute{L}^{\{x,y,z\}}$ is bounded away from zero, and as the quotient
of continuous functions the map continuous, and, furthermore,
$\lim_{L \rightarrow J}\acute{L}^{\{y,z,w\}} =
\acute{J}^{\{y,z,w\}}$. Similarly, on $N_{J}$, the fact that
$\langle u^{\yz}, L \rangle>0$ ensures that the map
$L \mapsto \epsilon$ is well-defined, continuous and $\lim_{L \rightarrow J} \epsilon =
0$.\footnote{Substituting for $\frac{1- \acute{\lambda}}{\acute{\lambda}}$ in the
equation for $\epsilon$ in \eqref{eq-yzw} and rearranging yields a linear
equation
\begin{equation}\label{eq-k-linear}
  k \langle u^{\yz} , L \rangle =\langle u^{\wz} , L \rangle
  \end{equation}
  in $L$, where
  $k = \tfrac{\epsilon+\langle u^{\wz},J \rangle}{\langle
    u^{\yz},J\rangle}$. Note that, since
  $L \in G^{\zw}_{\mdoubleplus}\cap G^{\yz}_{\mdoubleplus}$,
  $\langle u^{\wz}, L \rangle <0 $ and $ \langle u^{\yz}, L \rangle > 0$;
  moreover, the same is true of $J$. Thus, $k<0$ or, equivalently,
  $\epsilon < -\langle u^{\wz},J \rangle$.  If
  $ \langle u^{\yz},L \rangle < - \langle u^{\wz},L \rangle $, then, via
  \cref{eq-k-linear}, $k < -1$. Moreover, our choice of $N_{J}$ is such that
  $\frac{\langle u^{\wz},J \rangle}{\langle u^{\yz},J \rangle} < -1 $. This
  ensures that, on $N_{J}$, $\epsilon$ lies in a neighbourhood of zero, as
  required.}

For every $i,j \in \{y,z,w\}$, let
$\acute{\eta}^{\ij} = - \frac{1 - \acute{\lambda}}{\acute{\lambda}} \langle
u^{\ij}, L \rangle $.  Now note that, via equations
\eqref{eq-iz}--\eqref{eq-yzw}, we obtain
\begin{linenomath*}
  \begin{equation}\label{eq-3d-epsilon}
    \pi  \acute{\eta}^{\xz} + \rho  \acute{\eta}^{\yz} = \acute{\eta}^{\wz} + \epsilon \neq
    \acute{\eta}^{\wz},
  \end{equation}
\end{linenomath*}
so that, for every $\epsilon\neq 0$,
$\{\acute{u}^{\ij} = u^{\ij} \times \acute{\eta}^{\ij}: \ij = \xz, \yz, \wz\}$
forms a linearly independent set. Let
$\{\acute{u}^{\ij}: i,j \in \{x,y,z\}\cup \{y,z,w\}\}$ be defined similarly. To
complete the definition of $\acute{u}^{\dd}$, we appeal to the fact that, via
\threejac, $u^{\dd}$ satisfies equations \eqref{eq-xy}--\eqref{eq-xz}. In
particular, from these equations, take parameters $\alpha$, $\beta$, and $\gamma$
and let $\acute{\eta}^{\xw}$ be the (unique) solution to the Jacobi identity
\begin{linenomath*}
  \begin{equation}\label{eq-xw}
    \alpha  \acute{\eta}^{\xw} =  \gamma  \acute{\eta}^{\xy} + \beta \acute{\eta}^{\yw} =
    -\langle \gamma  u^{\xy} , J \rangle -
    \tfrac{1-\acute{\lambda}}{\acute{\lambda}}\langle\beta u^{\yw}, L  \rangle.
  \end{equation}
\end{linenomath*}
For the same set of parameter values, $\acute{u}^{\dd}$ also satisfies
\eqref{eq-xy}--\eqref{eq-xz}. That is, for $\{x,y,z\}$, via \cref{eq-xyz} and \eqref{eq-xz},
$\gamma \acute{\eta}^{\xy}+\tau \acute{\eta}^{\yz} = \phi \acute{\eta}^{\xz}$,
so that $\acute{u}^{\dd}$ satisfies \eqref{eq-xz}. For $\{y,z,w\}$, Via \cref{eq-yzw} and
\eqref{eq-yz},
$\hat{\beta} \acute{\eta}^{\yw}+\sigma \acute{\eta}^{\wz} = \tau
\acute{\eta}^{\yz}$, so that $\acute{u}^{\dd}$ satisfies \eqref{eq-yz}. For $\{x,y,w\}$,
via \cref{eq-xw}, $\acute{u}^{\dd}$ satisfies \eqref{eq-xy}.  We now demonstrate
that for the final triple $\{x,z,w\}$,  the Jacobi identity holds if
$\hat{\beta} = \beta$, and  $\{\acute{u}^{\xw},\acute{u}^{\wz},\acute{u}^{\xz}\}$
has rank $3$ otherwise.

First extract the parameters from equations \eqref{eq-xy}--\eqref{eq-xz} to
obtain the matrix form
\begin{linenomath*}
  \begin{equation}\label{eq-matrix-6}
   \begin{blockarray}{ccccccc}
     \text{\footnotesize$\xw$} & \text{\footnotesize$\yw$} &
     \text{\footnotesize$\xy$} &\text{\footnotesize $\wz$}& \text{\footnotesize
       $\yz$} & \text{\footnotesize $\xz$}\\
     \begin{block}{[ccc|ccc]c}
       \alpha & -\beta &  -\gamma  & 0 & 0 & 0 & \text{\footnotesize\eqref{eq-xy}} \\
       0 & \hat{\beta} &  0 & \sigma&  -\tau & 0 & \text{\footnotesize\eqref{eq-yz}}\\
       0 & 0 & \gamma  & 0  & \tau & -\phi & \text{\footnotesize\eqref{eq-xz}}\\
     \end{block}
   \end{blockarray}
 \end{equation}
\end{linenomath*}
Since the triple $\{\acute{u}^{(i,z)}: i = x,y,w\}$ provides a basis for
$\spann(\acute{u}^{\dd})$, we will write all vectors in terms of this basis. To
this end, we derive the reduced row echelon form of \cref{eq-matrix-6}. In
particular, letting $r_{i}$ denote the rows of the matrix, we perform the
operation $r_{1} \mapsto r_{1} + \frac{\beta}{\hat{\beta}} r_{2} + r_{3}$ to
obtain
\begin{linenomath*}
  \begin{equation}\label{eq-matrix-6-echelon}
    \begin{blockarray}{ccccccc}
      \text{\footnotesize$\xw$} & \text{\footnotesize$\yw$} &
      \text{\footnotesize $\xy$} & \text{\footnotesize$\wz$} &
      \text{\footnotesize
        $\yz$} & \text{\footnotesize $\xz$}\\
      \begin{block}{[ccc|ccc]c}
        \alpha & 0 & 0 & \frac{\beta}{\hat{\beta}} \sigma&
        (1-\frac{\beta}{\hat{\beta}})
        \tau & -\phi & \text{\footnotesize\eqref{eq-xy}} \\
        0 & \hat{\beta} &  0 & \sigma&  -\tau & 0 & \text{\footnotesize\eqref{eq-yz}}\\
        0 & 0 & \gamma  & 0  & \tau & -\phi & \text{\footnotesize\eqref{eq-xz}}\\
      \end{block}
    \end{blockarray}.
  \end{equation}
\end{linenomath*}
In \cref{eq-matrix-6-echelon}, the fact that $\hat{\beta}$ (instead of $\beta$)
that appears as a pivot in column $2$ indicates that, in this derivation, we
are, \withoutlog, choosing $\acute{v}^{\yw} = \hat{\beta}
\acute{u}^{\yw}$. \emph{Mutatis mutandis}, the conclusions we will draw are the
same if instead we chose to define $\acute{v}^{\yw}$ using $\beta$. The other
(relevant) rows of $\acute{v}^{\dd}: Y^{2} \times \mbbtp \rightarrow \R $ are
$\acute{v}^{\xw} = \alpha \acute{u}^{\xw}$,
$\acute{v}^{\xy} = \gamma \acute{u}^{\xy}$,
$\acute{v}^{\zw}= \sigma\acute{u}^{\zw}$,
$\acute{v}^{\yz} = \tau \acute{u}^{\yz}$ and
$\acute{v}^{\xz} = \phi \acute{u}^{\xz}$.  
\begin{figure}[t]
    \begin{center}
      \begin{tikzpicture}
        
    \node[color=lightgray] (max) at (0,4) {\footnotesize$\{x,y,z,w\}$};
    
    \node (xyz) at (-6.55,2.75){\footnotesize$\{x,y,z\}$};

    \node (xyw) at (-4.66,2.75) {\footnotesize$\{x,y,w\}$};

    \node[color=lightgray] (xzlyw) at (-2.5,2.75) {\footnotesize$\{x,z\}\{y,w\}$};

    \node (xylzw) at (0,2.75) {\footnotesize$\{x,y\}\{z,w\}$};

    \node (xwlyz) at (2.5,2.75) {\footnotesize$\{x,w\}\{y,z\}$};

    \node (xzw) at (4.66,2.75) {\footnotesize$\{x,z,w\}$};

    \node (yzw) at (6.55,2.75) {\footnotesize$\{y,z,w\}$};
    
    \node (xy) at (-5,0.25) {\footnotesize$\{x,y\}$};

    \node (xz) at (-3,0.25) {\footnotesize$\{x,z\}$};

    \node (xw) at (3,0.25) {\footnotesize$\{x,w\}$};
    
    \node (yz) at (-1,0.25) {\footnotesize$\{y,z\}$};
   
    \node (yw) at (1,0.25) {\footnotesize$\{y,w\}$};

    \node (zw) at (5,0.25) {\footnotesize$\{z,w\}$};
    
    \node (min) at (0,-1) {\footnotesize$\emptyset$};

    \draw (min) -- (xy);
    \draw (min) -- (xz);
    \draw (min) -- (xw);
    \draw (min) -- (yz);
    \draw (min) -- (yw);
    \draw (min) -- (zw);

  \draw[color=lightgray] (max) -- (xyz);
    \draw[color=lightgray] (max) -- (xyw);
    \draw[color=lightgray] (max) -- (xzw);
    \draw[color=lightgray] (max) -- (yzw);
    \draw[color=lightgray] (max) -- (xylzw);
    \draw[color=lightgray] (max) -- (xzlyw);
    \draw[color=lightgray] (max) -- (xwlyz);

    \draw[color=blue](xyz)--(xy);
    \draw[preaction={draw=white,-,line width=3pt},color=blue] (xzw)  -- (xz);
    \draw[preaction={draw=white,-,line width=3pt},color=blue] (yzw) -- (yz);
    \draw[preaction={draw=white,-,line width=3pt},color=blue] (xzw)  -- (zw);
    \draw[preaction={draw=white,-,line width=3pt},color=blue] (yzw) -- (yw);
    \draw[preaction={draw=white,-,line width=3pt},color=lightgray] (xzlyw) -- (yw);
    \draw[preaction={draw=white,-,line width=3pt},color=patrickcolor1] (xwlyz) -- (yz);
    \draw[preaction={draw=white,-,line width=3pt},color=patrickcolor1] (xwlyz) -- (xw);
    \draw[preaction={draw=white,-,line width=3pt},color=blue] (xyw) -- (xw);
    \draw[preaction={draw=white,-,line width=3pt},color=blue] (xyw) -- (yw);
    \draw[preaction={draw=white,-,line width=3pt},color=lightgray] (xzlyw) -- (xz);
    \draw[preaction={draw=white,-,line width=3pt},color=blue] (xyz) -- (yz);
    \draw[preaction={draw=white,-,line width=3pt},color=patrickcolor1] (xylzw) -- (xy);
    \draw[preaction={draw=white,-,line width=3pt},color=patrickcolor1] (xylzw) -- (zw);
    \draw[preaction={draw=white,-,line width=3pt},color=blue] (xyw)-- (xy);
    \draw[preaction={draw=white,-,line width=3pt},color=blue] (xzw)  -- (xw);
    \draw[preaction={draw=white,-,line width=3pt},color=blue] (yzw) -- (zw);
    \draw[preaction={draw=white,-,line width=3pt},color=blue] (xyz) -- (xz);
\end{tikzpicture}
\caption{\label{fig-hasse-centerless-Y4} The intersection semilattices
  $\acute{\mc L}$ and
  $\acute{\mc L}_{\mdoubleplus} = \acute{\mc L} \bs \{A^{Y},
  \acute{A}^{\{x,z\}\{y,w\}}\}$ when $\countof \mbbt=2$, $\hat{\beta} = \beta$
  and $\epsilon$ is sufficiently small but distinct from zero.}
\end{center}
\end{figure}
The matrix of the equation that now follows, is invertible if, and only if,
$(1-\frac{\beta}{\hat{\beta}}) \neq 0$.
\begin{linenomath*}
  \begin{equation}\label{eq-matrix-3}
    \begin{blockarray}{[c]}
      \acute{v}^{\xw}\\
      \acute{v}^{\xz}\\
      \acute{v}^{\zw}
    \end{blockarray}
    =
    \begin{blockarray}{[ccc]}
      -\frac{\beta}{\hat{\beta}}  & -(1-\frac{\beta}{\hat{\beta}}) & 1 \\
      0 & 0 & 1\\
      -1 & 0 & 0\\
    \end{blockarray}
    \begin{blockarray}{[c]}
      \acute{v}^{\wz}\\
      \acute{v}^{\yz}\\
      \acute{v}^{\xz}
    \end{blockarray}
  \end{equation}
\end{linenomath*}
Thus, unless $\hat{\beta} = \beta$, we conclude that
$\{\acute{v}^{\xw}, \acute{v}^{\xz},\acute{v}^{\zw}\}$ has the same rank as
$\{\acute{v}^{\wz}, \acute{v}^{\yz}, \acute{v}^{\xz}\}$ which, by construction,
has rank $3$ for every choice of $\epsilon \neq 0$.

Since $\hat{\beta}= \beta$ if, and only if, \fourjac\ holds for $\ext$, we
conclude that \fourjac\ holds for $\ext$ if, and only if, it holds for $\aext$.
It remains for us to show that, for $\epsilon$ sufficiently small
$\acute{\mc G}$ is maximal. For if $\acute{\mc G} $ is maximal, then via
\cref{lem-Y3-r2} and \cref{lem-Y3-r3}, \ref{TQ} holds if, and only if,
$\rank\{\acute{v}^{\xw}, \acute{v}^{\xz},\acute{v}^{\zw}\} = 2$.

By Zaslavski's theorem, it suffices to show that every member of the
intersection lattice $\acute{\mc L}$, other than the center
$\acute{A}^{\{x,y,z,w\}}$, has nonempty intersection with
$\posreal^{\mbbtp}$. We now make explicit the dependence of the extension
$\aext $ on our choice of $L \in N_{J}$, though we do so indirectly via
$\epsilon$.  Let
$d^{\epsilon} = \max \{\diff(\acute{J},\acute{A}): \acute{A} \in \acute{\mc
  L}^{\epsilon}$, where $\diff(\acute{J},\acute{A})$ is the minimum (Euclidean)
distance between $\acute{J}$ and the (closed) linear subspace $\acute{A}$ of
$\R^{\mbbtp}$. Note that for $\epsilon = 0 $, we obtain a central arrangement of
the form of \cref{prop-central-testworthy} with $d^{\epsilon} = 0$. Moreover,
since the Euclidean metric is continuous in its arguments, and, for every
$\epsilon$, $ \acute{\mc L}^{\epsilon}$ is finite, the map
$\epsilon \mapsto d^{\epsilon}$ is continuous and
$\lim_{\epsilon \rightarrow 0}d^{\epsilon} = 0$. Thus, for sufficiently small
$\epsilon \neq 0$, every $\acute{A} \in \acute{\mc L}^{\epsilon}$ intersects
$\posreal^{\mbbtp}$.

\begin{remark*}
  We note that the above arguments apply without modification to the case where
  $\countof \mc H = 4,5$. Consider, for example, the Hasse diagram of
  \cref{fig-hasse-centerless-Y4}. That case arises when $u^{\xz}$ and $u^{\yw}$
  are collinear, so that $A^{\{xz\}\{y,w\}}$ is a hyperplane of dimension
  $\countof \mbbt -1$. Assuming the same construction, with $\epsilon \neq 0$,
  so that $\acute{u}^{\dd}$ has rank $3$ and, via
  \cref{lem-arrangement-cardinality-rank}, $\acute{u}^{\xz}$ and
  $\acute{u}^{\yw}$ are linearly independent. Thus $\acute{A}^{\{xz\}\{y,w\}}$ is
  of dimension $\countof \mbbtp - 2 = \countof \mbbt - 1$. Thus,
  $\acute{A}^{\{x,z\}\{y,w\}} = A^{\{x,z\}\{y,w\}} \times \{0\} $ which belongs
  to the boundary of $\posreal^{\mbbtp}$. At $\epsilon = 0$,
  $\acute{A}^{\{x,z\}\{y,w\}}$ increases by one dimension and the upper two
  levels of the Hasse diagram collapse to equal $A^{\{x,y,z,w\}}$.
\end{remark*}
\newpage
\makeatletter
\def\@seccntformat#1{Online Appendix\,\csname the#1\endcsname.\quad}
\makeatother
\section{Proofs}
\printProofs
\end{appendices}


\end{document}